\documentclass[11pt]{article}
\pdfoutput=1
\usepackage{fullpage}
\usepackage[bookmarks]{hyperref}
\usepackage{amssymb}
\usepackage{amsmath}
\usepackage{amsthm}
\usepackage{graphicx,color,colordvi}
\usepackage{bbm}
\usepackage{varioref}
\usepackage{cleveref}
\usepackage{stmaryrd}
\usepackage[utf8]{inputenc}
\usepackage[blocks]{authblk}
\usepackage{dsfont}
\usepackage{mathtools}
\usepackage{authblk}
\newcommand{\norm}[1]{\left\| #1 \right\|}

\usepackage{comment}
\usepackage{tikz}
\usepackage{graphicx}
\usepackage{bbm}
\usepackage[T1]{fontenc}
\usepackage{url}
\usepackage{stmaryrd}
\usepackage{varioref}
\usepackage{cleveref}
\usepackage{braket}
\usepackage{ upgreek }
\usepackage{dsfont}
\usepackage{makecell}
\usepackage{fancyref} %references with object type, no pages
\usepackage{algorithm}
\usepackage{algpseudocode}
\newcommand{\identity}{\ensuremath{\mathds{1}}}
\usepackage{tikz}
\usepackage{ifthen}
\usepackage{pgfplots}
\pgfplotsset{compat=1.9}
\usetikzlibrary{shapes,arrows}
\usetikzlibrary{positioning}
\usetikzlibrary{shapes.geometric}

\usepackage{float}
\newfloat{algorithm}{t}{lop}

\newcommand{\abs}[1]{\left| #1 \right|}

\usepackage{hyperref}

\theoremstyle{plain}
\newtheorem{thm}{Theorem}[section]
\newtheorem{lem}[thm]{Lemma}
\newtheorem{cor}[thm]{Corollary}
\newtheorem{prop}[thm]{Proposition}
\newtheorem{defi}[thm]{Definition}
\theoremstyle{remark}
\newtheorem{rem}[thm]{Remark}

\newcommand{\supp}{\operatorname{supp}}

\newcommand*{\fancyrefthmlabelprefix}{thm}
\newcommand*{\fancyreflemlabelprefix}{lem}
\newcommand*{\fancyrefcorlabelprefix}{cor}
\newcommand*{\fancyrefdefilabelprefix}{defi}
\frefformat{plain}{\fancyreflemlabelprefix}{lemma\fancyrefdefaultspacing#1}
\Frefformat{plain}{\fancyreflemlabelprefix}{Lemma\fancyrefdefaultspacing#1}
\frefformat{plain}{\fancyrefthmlabelprefix}{theorem\fancyrefdefaultspacing#1}
\Frefformat{plain}{\fancyrefthmlabelprefix}{Theorem\fancyrefdefaultspacing#1}
\frefformat{plain}{\fancyrefcorlabelprefix}{corollary\fancyrefdefaultspacing#1}
\Frefformat{plain}{\fancyrefcorlabelprefix}{Corollary\fancyrefdefaultspacing#1}
\frefformat{plain}{\fancyrefdefilabelprefix}{definition\fancyrefdefaultspacing#1}
\Frefformat{plain}{\fancyrefdefilabelprefix}{Definition\fancyrefdefaultspacing#1}
\newcommand*{\fancyrefalglabelprefix}{alg}
\newcommand*{\frefalgname}{algorithm}
\newcommand*{\Frefalgname}{Algorithm}
\frefformat{plain}{\fancyrefalglabelprefix}{%
  \frefalgname\fancyrefdefaultspacing#1%
}%
\Frefformat{plain}{\fancyrefalglabelprefix}{%
  \Frefalgname\fancyrefdefaultspacing#1%
}%

\newcommand*{\fancyrefapplabelprefix}{app}
\newcommand*{\frefappname}{appendix}
\newcommand*{\Frefappname}{Appendix}
\frefformat{plain}{\fancyrefapplabelprefix}{%
  \frefappname\fancyrefdefaultspacing#1%
}%
\Frefformat{plain}{\fancyrefapplabelprefix}{%
  \Frefappname\fancyrefdefaultspacing#1%
}%

\def\Block[#1,#2,#3]{

\def\r{0.3};

\ifthenelse{\NOT #3=0}{
\fill [#2] (-0.5,-0.5) rectangle ({#1-0.5},0.5);
}

\foreach \n in {1,...,#1}{

\shade[shading=ball, ball color=coolblue] ({\n-1},0) circle (\r);

}

}

\definecolor{Green}{HTML}{00AD69}  % "Pantone 3405"
\definecolor{coolblue}{RGB}{0,51,102}
\definecolor{lightblue}{RGB}{102,210,255}
\definecolor{lightpurple}{RGB}{140,30,255}
\definecolor{lightpink}{RGB}{204,0,204}
\definecolor{midblue}{RGB}{0,102,204}
\definecolor{midpink}{RGB}{153,0,153}
\definecolor{darkblue}{RGB}{0,0,153}
\definecolor{cyan}{RGB}{0,204,204}
\definecolor{lightgreen}{RGB}{0,255,128}
\definecolor{midgreen}{RGB}{0,204,0}
\definecolor{midyellow}{RGB}{204,204,0}
\definecolor{darkyellow}{RGB}{153,153,0}
\definecolor{darkpurple}{RGB}{102,0,102}

\def\beq{\begin{equation}}
\def\eeq{\end{equation}}
\def\bq{\begin{quote}}
\def\eq{\end{quote}}
\def\ben{\begin{enumerate}}
\def\een{\end{enumerate}}
\def\bit{\begin{itemize}}
\def\eit{\end{itemize}}

\def\l|{\left|}
\def\r|{\right|}

\newcommand\M{\mathcal{M}}

\newcommand\cB{\mathcal{B}}

\newcommand{\cD}{\mathcal{D}}

\newcommand{\cK}{\mathcal{K}}

\newcommand{\cL}{\mathcal{L}}

\newcommand{\cM}{\mathcal{M}}
\newcommand{\Ncal}{\mathcal{N}}

\newcommand{\tr}{{\operatorname{tr}}}
\newcommand{\si}{\sigma}

\newcommand{\Id}{\mathds{1}}

\newcommand{\id}{\operatorname{id}}

\newcommand{\ten}{\otimes}

\newcommand{\cH}{\mathcal{H}}

\newcommand{\eps}{\epsilon}
\newcommand{\cN}{\mathcal{N}}
\newcommand{\cT}{\mathcal{T}}
\newcommand{\Tr}{{\operatorname{tr}}}
\newcommand{\DF}{\mathcal{N}}

\begin{document}

% \date{\today}

\title{Entropy decay for Davies semigroups\\ of a one dimensional quantum lattice}

% \author[Lucia]{Angelo Lucia}
% \author[Pérez-García]{David Pérez-García}
% \author[Rouzé]{Cambyse Rouzé}

% \address[Bardet]{Inria Paris, France}
% \email{bardetivan@gmail.com}

% \address[Capel]{Instituto de Ciencias Matemáticas (CSIC-UAM-UC3M-UCM), C/ Nicolás Cabrera 13-15, Campus de Cantoblanco, 28049 Madrid, Spain}
% \email{angela.capel@icmat.es}

% \address[Lucia]{California Institute of Technology, 1200 E California Blvd, MC 305-16, Pasadena, CA 91125, Unites States of America}

% \email{alucia@caltech.edu}

% \address[Pérez-García]{Departamento de Análisis Matemático, Universidad Complutense de Madrid, 28040 Madrid, Spain and Instituto de Ciencias Matemáticas (CSIC-UAM-UC3M-UCM), C/ Nicolás Cabrera 13-15, Campus de Cantoblanco, 28049 Madrid, Spain}

% \email{dperezga@ucm.es}
\author[1]{Ivan Bardet\thanks{ivan.bardet@inria.fr}}
\author[2,3,4]{{\'A}ngela Capel\thanks{angela.capel@uni-tuebingen.de}}
\author[2,3,5]{Li Gao\thanks{lgao20@central.uh.edu}}
\author[6,7]{Angelo Lucia\thanks{anglucia@ucm.es}}
\author[6,7]{David P\'{e}rez-Garc\'{i}a\thanks{dperezga@ucm.es}}
\author[2,3]{Cambyse Rouz\'{e}\thanks{cambyse.rouze@tum.de}}
\affil[1]{Inria Paris, 75012 Paris, France}
\affil[2]{Department of Mathematics, Technische Universit\"at M\"unchen, 85748 Garching, Germany}
\affil[3]{Munich Center for Quantum Science and Technology, 80799 M\"unchen, Germany}
\affil[4]{Fachbereich Mathematik, Universit\"at T\"ubingen, 72076 T\"ubingen, Germany}
\affil[5]{Department of Mathematics, University of Houston, 77204 Houston, USA}
\affil[6]{Departamento de An\'{a}lisis Matemático y Matemática Aplicada, Universidad Complutense de Madrid, 28040 Madrid, Spain}
\affil[7]{Instituto de Ciencias Matemáticas, 28049 Madrid, Spain}

\maketitle

\begin{abstract}
Given a finite-range, translation-invariant commuting system Hamiltonians on a spin chain, we show that the Davies semigroup describing the reduced dynamics resulting from the joint Hamiltonian evolution of a spin chain weakly coupled to a large heat bath thermalizes rapidly at any temperature. More precisely, we prove that the relative entropy between any evolved state and the equilibrium Gibbs state contracts exponentially fast with an exponent that scales logarithmically with the length of the chain. Our theorem extends a seminal result of Holley and Stroock \cite{holley1989uniform} to the quantum setting, up to a logarithmic overhead, as well as provides an exponential improvement over the non-closure of the gap proved by Brandao and Kastoryano \cite{kastoryano2016quantum}. This has wide-ranging applications to the study of many-body in and out-of-equilibrium quantum systems. Our proof relies upon a recently derived strong decay of correlations for Gibbs states of one dimensional, translation-invariant local Hamiltonians, and tools from the theory of operator spaces.
\end{abstract}

% \tableofcontents
\newpage

\section{Introduction}

The mitigation of errors arising from noise represents a key challenge in the development of large-scale quantum architectures \cite{Preskill2018}. Unfortunately, noise adversely affects all stages of quantum
computation, from the storage of quantum information in quantum memories to its manipulation by means of quantum gates and quantum channels. In the future, error correcting codes will be used to suppress noise and achieve fully fault-tolerant computation. However, full fault-tolerance requires an unattainable overhead for current and near-term hardware. This observation leads to the need for a better understanding and
characterization of the underlying noise dynamics in current quantum devices, both in
order to inform hardware design, and for the preparation of noise-dependent optimized error correction and mitigation
protocols.

Davies generators \cite{Davies1976,Davies1979,archbold_1983} model the thermal dynamics that emerge for a system weakly interacting with a large thermal reservoir. They are the natural candidates for the description of the thermal noise occurring on self-correcting quantum memories \cite{Alicki2010,Alicki_2009,lucia2021thermalization}. The weak coupling between the system and its environment allows to consider only the reduced dynamics on the system alone, which can be modeled by a quantum Markov semigroup, that is a semigroup of quantum channels. Given a spin chain $\Lambda=\llbracket 1,n\rrbracket$ made of $n$ qudits, of corresponding quantum system Hilbert space $\cH_\Lambda:=\bigotimes_{k\in \Lambda}\cH_k$, $\cH_k\equiv \mathbb{C}^d$, the closed evolution of the system in absence of environmental noise is described by a local Hamiltonian $H_\Lambda$ acting on $\cH_\Lambda$ and of the form:
\begin{align*}
    H_\Lambda:=\sum_{A\subset \Lambda}h_A\,,
\end{align*}
where for each subregion $A\subset \Lambda$, the local self-adjoint operator $h_A$ acts non-trivially on $A$. In this article, we further impose that the Hamiltonian $H_\Lambda$ is commuting, which means that for any two regions $A,A'\subset \Lambda$, $[h_A,h_{A'}]=h_Ah_{A'}-h_{A'}h_A=0$; and translation-invariant, i.e. for any $A\subset\!\subset\mathbb{Z}$ and $k\in \mathbb{Z} $, $h_{A+k}=h_A$, were $A+k:=\{x+k|\,x\in A\}$.
We also assume that $H_\Lambda$ has finite-range $r$ and bounded interaction strength $J:=\max_{A\subset \Lambda}\|h_A\|_\infty$ independently of the system size $n$. The former means that for any subregion $A$ such that $\operatorname{diam}(A):=\max\{k|k\in A\}-\min\{k|k\in A\}>r$, $h_A=0$. Under some standard conditions on the generator $\cL_{\Lambda*}^D$ acting on the algebra $\cB(\cH_\Lambda)$ of linear operators on $\cH_\Lambda$ (see e.g.~\cite{SL78}), the Davies semigroup $(e^{\cL_{\Lambda*}^D})_{t\ge 0}$ admits the following asymptotic behavior: for all initial states $\rho\in \cD(\cH_\Lambda)$,
\begin{align}\label{thermalization}
    e^{t\cL_{\Lambda*}^D}(\rho)\underset{t\to\infty}{\longrightarrow}\,\sigma^\Lambda\,,\qquad \text{ where }\qquad \sigma^\Lambda:=\frac{e^{-\beta H_\Lambda}}{\tr[e^{-\beta H_\Lambda}]}
\end{align}
denotes the Gibbs state of the Hamiltonian $H_\Lambda$ at the inverse temperature $\beta>0$ of the reservoir.

A natural question emerging from the above discussion is that of the speed at which the thermalization \eqref{thermalization} occurs. In order to study the latter, we introduce the \textit{mixing time} of the semigroup: for any $\eps>0$,
\begin{align}\label{mixing_time}
    t_{\operatorname{mix}}(\eps):=\inf\big\{t\ge 0|\,\forall\rho\in\cD(\cH_\Lambda),\, \|e^{t\cL_{\Lambda*}^D}(\rho)-\sigma^\Lambda\|_1\le \eps\big\}\,,
\end{align}
where $\|\rho-\sigma\|_1:=\tr\big|\rho-\sigma\big|$ denotes the trace distance between the states $\rho$ and $\sigma$. In \cite{kastoryano2016quantum}, the authors proved that the spectral gap of $\cL_{\Lambda *}^D$ is uniformly lower bounded by a constant $\lambda_0>0$ independent of the system size $n$. This can be proven to imply the following bound on the mixing time \cite{KastoryanoTemme-LogSobolevInequalities-2013}:
\begin{align}\label{equatrapidmixing}
    t_{\operatorname{mix}}(\eps)= \frac{1}{\lambda_0}\,\mathcal{O}\Big(\ln\Big(\frac{1}{\eps }\Big)+n\Big)\,.
\end{align}
However, by analogy with the classical setting of Glauber dynamics, one would expect a logarithmic scaling with system size \cite{holley1989uniform,Zegarlinski1990}. In our main result, we answer this conjecture in the positive. In fact, we prove the following stronger result:
\begin{thm}[Informal formulation]\label{thm:main}
For any reservoir inverse temperature $\beta>0$ and all states $\rho\in\cD(\cH_\Lambda)$:
\begin{align}\label{eq:entdecay}
    D(e^{t\cL_{\Lambda*}^D}(\rho)\|\sigma^\Lambda)\le e^{-4\alpha_n t}\,D(\rho\|\sigma^\Lambda)\,,
\end{align}
where $\alpha_n=\Omega(\ln(n)^{-1})$, and where $D(\rho\|\sigma^\Lambda):=\tr(\rho\ln\rho)-\tr(\rho\ln\sigma^\Lambda)$ denotes the relative entropy between a state $\rho$ and $\sigma^\Lambda$. Moreover, the constant $\alpha_n$ scales exponentially with $\beta$.
\end{thm}

Note that this theorem is indeed stronger than a logarithmic scaling of the mixing time of Equation \eqref{mixing_time} as a consequence of Pinsker's inequality and Equation \eqref{eq:entdecay}. The inequality \eqref{eq:entdecay} is referred to as a \textit{modified logarithmic Sobolev inequality} (MLSI) \cite{KastoryanoTemme-LogSobolevInequalities-2013} and the constant $\alpha_n$ satisfying this bound is called the \textit{MLSI constant}. The problem of determining whether a quantum Markov semigroup satisfies a MLSI has been addressed in various settings in the last years. Some examples appear in \cite{MullerHermesFrancaWolf-DepolarizingChannels-2016,MullerHermesFrancaWolf-EntropyProduction-2016}, where it was shown that the MLSI constant of the depolarizing channel can be lower bounded by $1/2$. This was subsequently extended to the generalized depolarizing channel in \cite{BeigiDattaRouze-ReverseHypercontractivity-2018, capel2018quantum}. \cite{bardet2019modified} constitutes the first attempt to prove the inequality in the setting of spin systems, where the heat-bath  generator in 1D was shown to satisfy a MLSI under two conditions of decay of correlations on the Gibbs state. More recently, \cite{capel2020MLSI} proved the modified logarithmic Sobolev inequality for a family of Gibbs samplers of nearest neighbor commuting Hamiltonians above a certain threshold temperature. However, the drawback of \cite{capel2020MLSI} stems from the lack of physicality of the generator considered.

Although the proof of \Cref{thm:main} relies on a notion of clustering of correlations similar in spirit to that used in \cite{capel2020MLSI}, it also requires advanced techniques from the operator space structure of amalgamated $L_p$ spaces \cite{junge2010mixed}. To our knowledge, this is the first time that such techniques are being used to solve a problem from many-body quantum systems. It is also worth noticing that our result holds for any inverse temperature $\beta$.

\subsection{Applications}

As mentioned above, \Cref{eq:entdecay} implies the following corollary, as a consequence of Pinsker's inequality:
\begin{cor}\label{corollaryrapidmixing}
For any $\beta>0$, the semigroup $(e^{t\cL_\Lambda^D})_{t\ge 0}$ satisfies the following rapid mixing property:
 \begin{align*}
    t_{\operatorname{mix}}(\eps)\le \frac{1}{\alpha_n}\mathcal{O}\Big(\ln\Big(\frac{1}{\eps }\Big)+\ln(n)\Big)\,.
\end{align*}
\end{cor}
Combining \Cref{corollaryrapidmixing} with the main result in  \cite{CubittLuciaMichalakisPerezGarcia_StabilityRapidMixing_2015}, we can conclude that local observables and correlation functions are stable against local perturbations in the generator of  Davies semigroups corresponding to a one dimensional, short-range, translation-invariant commuting Hamiltonian at any inverse temperature. Moreover, entropic decays of the form of that of \Cref{thm:main} are known to imply a strengthening of Pinsker's inequality \cite{Rouz2019,DePalma2021,de2021quantum}, given the following quantum generalization of the Lipschitz constant corresponding to the Hamming distance: for any $H=H^*\in\cB(\cH_\Lambda)$,
\begin{align*}
    \|H\|_L:=2\max_{i\in[n]}\,\min\Big\{\|H-H^{(i)}\|_\infty|\,H^{(i)}=(H^{(i)})^*\,\text{ does not act on $i$-th qudit} \Big\}\,,
\end{align*}
we define the corresponding quantum Wasserstein distance between two quantum states $\rho,\sigma\in\cD(\cH_\Lambda)$ as
\begin{align}
    W_1(\rho,\sigma)=\max\big\{\tr[(\rho-\sigma)\,H]|\,H=H^*,\,\|H\|_L\le 1\big\}\,.
\end{align}
Next, a state $\sigma\in\cD(\cH_\Lambda)$ satisfies a transportation-cost inequality with constant $C>0$ if for any $\rho\in\cD(\cH_\Lambda)$,
\begin{align}\label{TCineq}
    W_1(\rho,\sigma)\le \sqrt{C\,D(\rho\|\sigma)}\,.
\end{align}
It was shown in \cite[Section 5]{de2021quantum} that fixed points of local quantum Markov semigroups satisfy a transportation-cost inequality with constant $C$ proportional to $n/\alpha$, where $\alpha$ is the MLSI constant of the semigroup. Combining this with \Cref{eq:entdecay}, we get
\begin{prop}
Gibbs states of finite-range, commuting, one dimensional Hamiltonians of bounded interaction strength at the inverse temperature $\beta>0$ satisfy a transportation-cost inequality with constant $C=\mathcal{O}(n\operatorname{polylog}(n))$.
\end{prop}
Transportation-cost inequalities in the form of \eqref{TCineq} were recently used to produce sample efficient Gibbs state learning algorithms in \cite{rouze2021learning}. They also imply Gaussian concentration inequalities for the statistics of quasi-local observables in the state $\sigma$ and can be used to prove an exponential improvement over the weak
Eigenstate Thermalization Hypothesis \cite{DePalma2021}.

Finally, we remark that Theorem \ref{thm:main} has implications in the context of Symmetry Protected Topological (SPT) phases \cite{Chen2011},  i.e. in systems protected against perturbations that keep the symmetry of the system. There has been a recent increase on the interest of SPT in open quantum systems \cite{coser2019classification, mcginley2019interacting,mcginley2020fragility}, partially motivated by their connection with Majorana modes and their potential use in quantum information theory. In the weak coupling limit between the system and the reservoir, it is shown in \cite{mcginley2020fragility} that, for unitary symmetries, the decoherence time scales exponentially with the inverse temperature $\beta$, even in the case in which all interactions involved in the problem preserve the symmetry that gives the topological protection to the (closed) system. Therefore, the decoherence time increases exponentially with the decrease of the temperature. However, there is in principle no analysis on the dependence of the decoherence time on the system size.

The particular case of the 1D cluster state \cite{Briegel2001}  has a commuting Hamiltonian and it is a non-trivial SPT phase under a $\mathbb{Z}_2\times \mathbb{Z}_2$ symmetry \cite{Son2011}. Thus, as a consequence of Theorem \ref{thm:main} for this state we obtain an example of a non-trivial SPT phase with decoherence time scaling only logarithmically on the size of the system in the presence of thermal noise, under the assumption that all relevant interactions preserve the symmetry. A previous result in this direction was already presented in \cite{coser2019classification} for on-site depolarizing noise, for which a similar behavior was shown. For further information on this, see our companion paper \cite{bardet2021MLSIDavies1Dshort}.

\subsection{Outline}
In Section \ref{sec:preliminaries}, we introduce some notations and preliminary notions which are necessary for the rest of the paper. The main result is presented in Section \ref{sec:main_result}, where we also include a complete scheme of its proof, leaving the specifics of the most technical parts
 to Sections \ref{sec:mixingcondition} and \ref{sec:localcontrol} for sake of clarity.

\section{Preliminary facts and definitions}\label{sec:preliminaries}

\subsection{Basic notations}\label{subsec:notation}

Let $\cH$ be a finite dimensional Hilbert space and $\cB(\cH)$ (resp. $\cT_1(\cH)$) be the bounded operators (resp. trace class operators) on $\cH$. We denote $\tr$ for the standard matrix trace, $\langle \cdot,\cdot \rangle_{\operatorname{HS}}$ for the trace inner product and $\|.\|_{2}$ for the Hilbert-Schmidt norm. The corresponding Hilbert-Schmidt space is denoted by $\cT_2(\cH)$. Operators will be sometimes denoted by capital letters, and sometimes by lowercase letters, e.g. in order to emphasize their belonging to a subalgebra. We write $A^\dagger$ for the adjoint of an operator $A\in \cB(\cH)$, and  $\Phi^*$ (or $\Phi_*$) for the adjoint (or preadjoint) of a map $\Phi:\cB(\cH)\to \cB(\cH)$. The identity operator on $\cH$ is denoted as $\Id_{\cH}$ and the identity map on a von Neumann subalgebra $\cM\subseteq \cB(\cH)$ is $\id_{\cM}$. We also denote the dimension of $\cH$ by $d_\cH=\text{dim}(\cH)$. Given two maps $\Phi,\Psi:\cM\to \cM$ on a von Neumann subalgebra $\cM\subseteq \cB(\cH)$, we say that $\Phi\le_{\operatorname{cp}} \Psi$ if $\Psi-\Phi$ is completely positive.

We say that an operator $\rho$ is a state (or density operator) if $\rho\ge 0$ and $\tr(\rho)=1$. We denote by $\cD(\cH)$ the set of states on $\cH$. Given a system $A$, we write $\rho_A$ to emphasize that $\rho_A \in \mathcal{D}(\mathcal{H}_A)$, though we will drop this subindex whenever it is clear from the context. With a slight abuse of notation, we write $\Psi(\rho):=(\Psi\otimes  \id) (\rho)$ for a bipartite state $\rho\in\cD(\cH\otimes \mathbb{C}^n)$ and a quantum channel $\Psi:\cM_*\to\cM_*$. For two states $\rho$ and $\si$, their relative entropy is defined as \begin{align*}D(\rho\|\si)=\begin{cases}
                             \tr(\rho\ln\rho-\rho\ln\si), & \mbox{if } \supp(\rho)\subseteq \supp(\si) \\
                             +\infty, & \mbox{otherwise},
                           \end{cases}
\end{align*}
where $\supp(\rho)$ (resp. $\supp(\si)$) is the support projection of $\rho$ (resp. $\si$). Let $\cN\subseteq \cM\subseteq  \cB(\cH)$ be two von Neumann subalgebras. Recall that a conditional expectation onto $\cN$ is a completely positive unital map $E_\cN:\cM\to \cN$ satisfying \begin{enumerate}
\item[i)]for all $X\in \cN$, $E_\cN(X)=X$
\item[ii)]for all  $a,b\in\cN,X\in \cB(\cH)$, $E_\cN(aXb)=aE_\cN(X)b$.
 \end{enumerate}
We denote by $E_{\cN*}$ its adjoint map with respect to the trace inner product, i.e.
\[\tr(E_{\cN*}(X)Y)=\tr(XE_\cN(Y))\,. \]
A quantum Markov semigroup (QMS) $(\mathcal{P}_t)_{t\ge 0}$ is a semigroup of completely positive, unital maps over the bounded operators $\cB(\cH)$ of a finite dimensional Hilbert space $\cH$. By the semigroup property, the QMS is generated by a map $\cL:\cB(\cH)\to\cB(\cH)$ that is referred to as a Lindbladian. The QMS is said to be GNS-symmetric with respect to a state $\sigma\in\cD(\cH)$ if for any two $X,Y\in\cB(\cH)$,
\begin{align}
    \tr[\sigma\,\cL(X)Y]=\tr[\sigma X\cL(Y)]\,.
\end{align}
Whenever the state $\sigma$ is full-rank, there exists a conditional expectation $E$ onto the kernel of $\cL$ such that for any $X\in\cB(\cH)$ \cite{frigerio1982long}
\begin{align}
    e^{t\mathcal{L}}(X)\underset{t\to\infty}{\longrightarrow}E[X]\,.
\end{align}

\subsection{Hamiltonians and Davies local generators}\label{subsec:Davies}

In this paper, we consider a spin chain $\Lambda:=\llbracket1,n\rrbracket$ made of $n$ sites, where at each site $k\in \Lambda$ lies a qudit system $\cH_k\equiv \mathbb{C}^d$, and denote by $\cH_\Lambda:=\bigotimes_{k\in\Lambda} \cH_k$ the Hilbert space of the whole system. At thermal equilibrium, the system is in a Gibbs state $\sigma^\Lambda:=e^{-\beta H_\Lambda}/\tr{e^{-\beta H_\Lambda}}$ at inverse temperature $\beta>0$, for some fixed Hamiltonian $H_\Lambda$. Here, we further assume $H_\Lambda$ to be finite-range and denote by $r$ its interaction range. This means that there exist self-adjoint operators $h_A$ supported on regions $A\subset\Lambda$ such that $H_\Lambda=\sum_{A\subset\Lambda}h_A\otimes\Id_{A^c}$ and $h_A=0$ whenever $\operatorname{diam}(A)>r$, where $\operatorname{diam}(A)$ denotes the diameter of $A$. The interaction strength of $H_\Lambda$ is defined as $J:=\max_{A\subset \Lambda}\|h_A\|_\infty$. In this paper, we assume that neither $r$, nor $J$ depend on the system size $\Lambda$.

The thermalization of the system due to its interaction with a heat-bath modeled by the Hamiltonian $H^{\operatorname{HB}}$ can be described as follows: assume a set of chain-bath interaction operators $\{ S_{\alpha,k}\otimes B_{\alpha,k} \}$, where $\alpha$ labels operators $S_{\alpha,k}$ and $B_{\alpha,k}$ associated to the site $k\in\Lambda$. For sake of simplicity, we assume that the operators $S_{\alpha,k}$ form an orthonormal basis of self-adjoint operators in $\cH_\Lambda$ with respect to the Hilbert-Schmidt inner product (e.g.~qubit Pauli matrices). The Hamiltonian of the universe composed of the spin chain and its heat-bath is thus given by
 \begin{align}
 H=H_\Lambda+H^{\operatorname{HB}}+\sum_{\alpha,k\in\Lambda}S_{\alpha,k}\otimes B_{\alpha,k}\,.
 \end{align}
 Assuming that the bath is in a Gibbs state, by a standard argument (e.g. weak coupling limit, see \cite{SL78}), the evolution on the system can be approximated by a quantum Markov semigroup whose generator is of the form
 \begin{align}\label{totalgenerator}
 {\cL}^D_{\Lambda*}(\rho):=-i[H_\Lambda,\rho]+\sum_{k\in\Lambda}\,\cL^{D}_{k*}(\rho)\,,
 \end{align}
for some local generators $\cL^{{D}}_{k*}$ in GKLS form
 \begin{align}\label{lindblad}
 &\cL^{{D}}_{k*}(\rho)=\sum_{\omega,\alpha}\,\chi^{\beta,\omega}_{\alpha,k}\,\Big(  S_{\alpha,k}^{\omega}\rho S_{\alpha,k}^{\omega,\dagger}-\frac{1}{2}\,\big\{ \rho, S_{\alpha,k}^{\omega,\dagger}S_{\alpha,k}^{\omega} \big\}   \Big)\,.
\end{align}
The sum in \eqref{lindblad} ranges over the index $\alpha$ of the local basis $\{S_{\alpha,k}\}$ as well as the Bohr frequencies $\omega$ of the Hamiltonian $H_\Lambda$. Similarly, we denote the generator $\cL^D_{A*}$ by restricting the sum in \eqref{totalgenerator} to the subregion $A$ in the interior of $ \Lambda$. Note that $\mathcal{L}^{{D}}_{A*}$ acts non-trivially on $A\partial:=\{k\in\Lambda:\,\operatorname{dist}(k,A)\le r\}$. Above, the Fourier coefficients of the two-point correlation functions of the environment $\chi_{\alpha,k}^{\beta,\omega}$, which we assume to be uniformly upper and lower bounded along the spin chain, $0<\chi_{\min}^\beta\le \chi_{\alpha,k}^{\beta,-\omega}\le \chi_{\max}^\beta$, satisfy the KMS condition $\chi_{\alpha,k}^{\beta,-\omega}=e^{-\beta\omega}\,\chi_{\alpha,k}^{\beta,\omega}$. The operators $S_{\alpha,k}^{\omega}$ are the Fourier coefficients of the system couplings $S_{\alpha,k}$, which means that they satisfy the following defining equation for any $t\in\mathbb{R}$:
 \begin{align}\label{eq!}
 e^{-itH_\Lambda}\,S_{\alpha,k}e^{it H_\Lambda}=\sum_\omega e^{it\omega}S_{\alpha,k}^{\omega}\,.
 \end{align}

\subsection{Modified logarithmic Sobolev inequality}\label{subsec:functineq}

%  Let us recall that the generator $\cL_{\Lambda*}^D$ is said to satisfy a  \textit{modified logarithmic Sobolev inequality} (MLSI for short) if there exists a positive constant $\alpha >0$ such that
% \begin{align*}
%   \alpha \, D(\rho\|\sigma) \leq - \tr[ \cL_{\Lambda*}^D(\rho) ( \ln \rho - \ln \sigma)]
% \end{align*}
% for every $\rho \in \mathcal{D}(\mathcal{H})$. The constant $\alpha$ for the above inequality is called the \textit{modified logarithmic Sobolev constant}. The problem of proving MLSIs for certain quantum Markov semigroups has attracted an increasing attention in recent years, in part due to its connections to mixing properties of many-body systems as well as applications in many other interesting contexts, such as for the estimation of noise and decoherence. In particular, in a series of works \cite{capel2018quantum, bardet2019modified, capel2019thesis, bardet2020approximate, capel2020MLSI, gao2021spectral}, some of the present authors have conceived a strategy to prove positivity of a MLSI via results of quasi-factorization (a.k.a. approximate tensorization) of the relative entropy. The strategy followed here is in spirit similar to those of.

% \angela{Rewrite section}

In this section, we consider the generator $\cL$ of a quantum Markov semigroup $(\mathcal{P}_t)_{t\ge 0}$ over the finite dimensional algebra $\cB(\cH)$, which we assume GNS-symmetric with respect to a full-rank invariant state $\sigma$, and we denote by $E:\cB(\cH)\to \operatorname{Ker}(\cL)$ the conditional expectation onto the kernel of $\cL$ such that $\mathcal{P}_t\to E$ as $t\to\infty$. The \textit{entropy production} of $\cL$ is defined for any other state $\rho\in\cD(\cH)$ by
\begin{align*}
\operatorname{EP}_{\cL}(\rho):=-\left.\frac{dD(\operatorname{e}^{t\cL_*}(\rho)\|E_*(\rho))}{dt}\right|_{t=0}\,.
\end{align*}
The entropy production is always non-negative, by the monotonicity of the relative entropy under quantum channels. Moreover, it satisfies the following useful property:
\begin{lem}\label{EPexpress}
For any other full-rank invariant state $\omega$:
	\begin{align*}
	\operatorname{EP}_{\cL}(\rho)=-\tr\big[\cL_*(\rho)\big( \ln(\rho)-\ln(\omega)  \big)\big]\,.
	\end{align*}
\end{lem}
\begin{proof}
	This simply follows from the fact that the difference of the logarithms of the two invariant states $\ln(\omega)-\ln(E_*[\rho])$ belongs to the fixed point algebra $\mathcal{F}(\cL)$ (see for instance the structure of invariant states in Equation (2.10) of \cite{bardet2018hypercontractivity}). Therefore
	\begin{align*}
	\tr\big[\cL_*(\rho)\big(\ln(E[\rho])-\ln\omega\big)\big]=\tr\big[\rho\cL\big(\ln(E[\rho])-\ln\omega\big)\big]=0\,.
	\end{align*}
\end{proof}

\begin{defi}[\cite{KastoryanoTemme-LogSobolevInequalities-2013,CarboneMartinelli_LogSobolev_2015, Bardet-NonCommFunctInequalities-2017,gao2020fisher}]
	The quantum Markov semigroup $(\mathcal{P}_t)_{t\ge 0}$ is said to satisfy a (non-primitive) \emph{modified logarithmic Sobolev inequality} $\operatorname{(MLSI)}$ if there exists a constant $\alpha>0$ such that, for all $\rho\in\cD(\cH)$:
	\begin{align}\tag{MLSI}\label{MLSI}
	4\alpha\, D(\rho\|E_*(\rho))\le \operatorname{EP}_\cL(\rho)\,.
	\end{align}
	The best constant satisfying \eqref{MLSI} is called the \emph{modified logarithmic Sobolev constant} and denoted by $\alpha(\cL)$. Moreover, the semigroup satisfies a \emph{complete modified logarithmic Sobolev inequality} $\operatorname{(CMLSI)}$ if, for any reference system $\cH_R$, the semigroup $(\operatorname{e}^{t\cL}\otimes \id_R)_{t\ge 0}$ satisfies a modified logarithmic Sobolev inequality with a constant $\alpha$ independent of $R$. In this case, the best constant satisfying $\operatorname{CMLSI}$ is called the \emph{complete modified logarithmic Sobolev constant} and is denoted by $\alpha_{\operatorname{c}}(\cL)$.
\end{defi}
	 The reason for the introduction of the complete modified logarithmic Sobolev constant is due to its tensorization property:
	 \begin{lem}[\cite{brannan2021complete2}]
	     Let $\cL$ and $\cK$ be two generators of $\operatorname{KMS}$-symmetric quantum Markov semigroups, and denote by $E_\cL$, resp. by $E_\cK$, their corresponding conditional expectations. Moreover, assume that $[\cL,\cK]=0$. Then
	     \begin{align*}
	        \alpha_{\operatorname{c}}(\cL+\cK)\ge \min\{ \alpha_{\operatorname{c}}(\cL),\,\alpha_{\operatorname{c}}(\cK)\}\,.
	     \end{align*}
	 \end{lem}
The following result was recently proved in \cite{gao2021spectral} (see also \cite{gao2021geometric} in the tracial setting):
\begin{thm}[\cite{gao2021spectral}]\label{CMLSItheorem}
	For any $\operatorname{GNS}$-symmetric $\operatorname{QMS}$ $(e^{t\cL})_{t\ge 0}$ over the algebra $\cB(\cH)$ of linear operators on a finite dimensional Hilbert space $\cH$, $\alpha_{\operatorname{c}}(\cL)>0$.
	\end{thm}
By Gr\"{o}nwall's inequality, the (complete) modified logarithmic Sobolev inequality is directly related to the exponential convergence of the evolution towards its equilibrium, as measured in relative entropy:
\begin{align*}
	D\big(\operatorname{e}^{t\cL_*}(\rho)\big\|E_*(\rho) \big)\le \operatorname{e}^{-4\alpha (\mathcal{L}) t}\,D(\rho\|E_*(\rho))\,.
	\end{align*}

 In this paper, we consider the MLSI constants $\alpha (\mathcal{L}^D_{\llbracket1,n\rrbracket})$ of the family $\{\cL^D_{\llbracket1,n\rrbracket}\}_{n\in\mathbb{N}}$ of Davies Lindbladians defined in \Cref{subsec:Davies}.
% $$\alpha (\cL^D) := \underset{n \nearrow \infty}{\text{lim inf}} \; \alpha (\mathcal{L}^D_{\llbracket1,n\rrbracket})\,.$$

% We immediately have:
% \begin{lem}[\cite{KastoryanoTemme-LogSobolevInequalities-2013}]
% Let $\cL:=\{\cL_\Lambda\}_{\Lambda\subset \subset \ZZ^d}$ be a primitive uniform family of Lindbladians. If $\alpha(\cL)>0$, then $\cL$ is rapidly mixing.
% \end{lem}

\subsection{Operator space theory}\label{sec:opspace}
In this section, we briefly recall some basic definitions from operator space theory, and refer the interested reader to the standard books \cite{effros2000operator,Pisier03} for further details. Given a complex vector space $V$ and an integer $n$, we denote by $\mathbb{M}_n(V)$ the space of $V$-valued $n\times n$ matrices. We also denote by $\mathbb{M}_{m,n}$ the space of $m\times n$ complex valued matrices, and write $\mathbb{M}_{n,n}\equiv \mathbb{M}_n$.
An (abstract) operator space $V$ is a complex vector space equipped with a sequence of norms on the spaces $\mathbb{M}_n(V)$, which satisfy Ruan's axioms:
\begin{itemize}
\item[(i)] For all integers $m,n$ and two elements $v\in \mathbb{M}_m(V),w\in\mathbb{M}_n(V)$,
\begin{align}
    \|v\oplus w\|_{\mathbb{M}_{m+n}(V)}=\max\{\|v\|_{\mathbb{M}_m(V)},\|w\|_{\mathbb{M}_n}(V)\}\,;
\end{align}
\item[(ii)] For all integers $m,n$ and $v\in \mathbb{M}_m(V)$, $A\in\mathbb{M}_{n,m},B\in \mathbb{M}_{m,n}$,
\begin{align}
\|A\cdot v\cdot B\|_{\mathbb{M}_{n}(V)}\le \|A\|_{\mathbb{M}_{m,n}}\,\|v\|_{\mathbb{M}_{m}(V)}\,\|B\|_{\mathbb{M}_{m,n}}\,,
\end{align}
\end{itemize}
where $v\oplus w$ denote the $(m+n) \times (m+n)$ diagonal matrix
$\left[\begin{smallmatrix} v& 0\\ 0&w\end{smallmatrix}\right]$,
the product $A\cdot v \cdot B$ is the multiplication on matrix coefficients and $\|.\|_{\mathbb{M}_{m,n}}$ denotes the usual matrix norm. With this definition, $\mathbb{M}_n(V)$ is again an operator space with the identification $\mathbb{M}_m(\mathbb{M}_n(V))\cong \mathbb{M}_{mn}(V)$ in the canonical way.

Given two operator spaces $V,W$ and a linear mapping $\varphi:V\to W$, we denote for each $n\in\mathbb{N}$ the map $\varphi_n\equiv \id_n\otimes \varphi:\mathbb{M}_n(V)\to \mathbb{M}_n(W)$ by simply mapping each matrix component of $v\in \mathbb{M}_n(V)$ to its image under the map $\varphi$, i.e. $[\varphi_n(v)]_{i,j}:=[\varphi(v_{ij})]_{i,j}$. The \textit{completely bounded norm} (in short $\operatorname{CB}$-norm) of $\varphi$ is defined as
\begin{align*}
    \|\varphi\|_{\operatorname{cb}}:=\sup_{n\in \mathbb{N}}\,\|\id_n\otimes \varphi: \mathbb{M}_n(V)\to \mathbb{M}_n(W)\|\,,
\end{align*}
given the above supremum is finite.
Here, for each $n$, $\|\id_n\otimes \varphi:\mathbb{M}_n(V)\to \mathbb{M}_n(W)\|$ is taken with respect to the norm structures of $\mathbb{M}_n(V)$ and $\mathbb{M}_n(W)$. A map $\varphi$ is said to be \textit{completely isometric} if for all $n\in\mathbb{N}$, $\id_n\otimes \varphi:\mathbb{M}_n(V)\to\mathbb{M}_n(W)$ is isometric. Two operator spaces $V$ and $W$ are said to be completely isometric, which we denote by $V\cong W$, if there exists a complete isometry $\varphi:V\to W$. We denote by $\operatorname{CB}(V,W)$ the space of completely bounded maps $\varphi:V\to W$ equipped with the CB-norm $\norm{\cdot}_{\operatorname{cb}}$. The space $\operatorname{CB}(V,W)$ itself is an operator space by the identification
\[ \mathbb{M}_n\big(\operatorname{CB}(V,W)\big)\cong \operatorname{CB}\big(V,\mathbb{M}_n(W)\big)\ .\]
Let $V^*$ be the Banach space dual of $V$. The dual operator space structure of $V^*$ is then given by $V^*\cong \operatorname{CB}(V,\mathbb{C})$. %In other words,
%we identify $\mathbb{M}_n(V^*)$ with $\mathbb{M}_n(V)^*$ via the duality $\langle [f_{ij}]_{i,j},\,[x_{i,j}]_{i,j}\rangle:=\sum_{ij}f_{ij}(v_{ij})$.
This induces a norm on each $\mathbb{M}_n(V^*)$, which can be easily seen to verify condition (i) and (ii) above. For example, $\cT_1(\cH)^*=\cB(\cH)$ as an operator space's dual space.

%: for any $A:=[A_{ij}]_{i,j}\in \mathbb{M}_n(\operatorname{CB}(V,W))$ and $x:=[x_{kl}]_{k,l}\in \mathbb{M}_m(V)$, consider $[A_{ij}(x_{kl})]_{i,j,k,l}\in \mathbb{M}_m(\mathbb{M}_n(W))\equiv \mathbb{M}_{mn}(W)$, and set
%\begin{align*}
 %   \|[A_{ij}]_{i,j}\|_n:=\sup\{\|[A_{ij}(x_{kl})]_{i,j,k,l}\|_{mn}\,:\|[x_{kl}]_{k,l}\|_m\le 1 \}\,,
%\end{align*}
%where it should be clear with respect to which operator space structures the different norms $\|.\|_q$ are being taken.

%Let $V^*$ be the Banach space dual of $V$. The dual operator

%Moreover, given an operator space $V$, its linear space dual $V^*$ can itself be provided with the structure of an operator space by identifying $\mathbb{M}_n(V^*)$ with $\mathbb{M}_n(V)^*$ via the duality $\langle [f_{ij}]_{i,j},\,[x_{i,j}]_{i,j}\rangle:=\sum_{ij}f_{ij}(v_{ij})$. This induces a norm on each $\mathbb{M}_n(V^*)$, which can be easily seen to verify condition (i) and (ii) above.

A concrete operator space $V$ is a norm-closed subspace of $\cB(\cH)$ for some Hilbert space $\cH$. The natural operator space structure is given by the embedding $\mathbb{M}_n(V)\subset \mathbb{M}_n(\cB(\cH))\cong \cB(\mathbb{C}^n\ten \cH)$. It was proved by Ruan \cite{Ruan88} that every abstract operator space is a concrete operator space.
We also need the notion of \textit{operator space minimal tensor product} (also known as \textit{operator space injective tensor product}, or \textit{spatial tensor product}). Let $V\subset \cB(\cH)$ and $W\subset \cB(\cK)$ be two (concrete) operator spaces. The operator space minimal tensor product $V\otimes_{\min}W$ is given by the inclusion
\[V\otimes_{\min}W\subset \cB(\cH\ten \cK)\,. \ \]
Note that this definition is independent of the choice of embedding $V\subset \cB(\cH)$ and $W\subset \cB(\cK)$.
More precisely,
\begin{align*}
   & \|u\|_{\min,n}\\&\quad :=\sup\{\|(f\otimes g)(u)\|_{n}\,: \,f:V \to  \mathbb{M}_p,\,g:W \to \mathbb{M}_q,\,\|f:V \to  \mathbb{M}_p\|_{\operatorname{cb}},\,\|g:W \to \mathbb{M}_q\|_{\operatorname{cb}}\le 1\}\,,
\end{align*}
where the supremum is over any pair of integers $q,p$. Here, the norm $\|.\|_n$ is taken with respect to the natural operator space structure on $\mathbb{M}_p\otimes \mathbb{M}_q$
 obtained when identifying $\mathbb{M}_n(\mathbb{M}_p\otimes \mathbb{M}_q)$ with $\cB((\mathbb{C}^p\otimes \mathbb{C}^q)^{\oplus n},(\mathbb{C}^p\otimes \mathbb{C}^q)^{\oplus n})$.

 In this manuscript, we are exclusively interested in the operator space structure of a non-commutative amalgamated $L_p$ space, which is discussed in Section \ref{sec:amalgamatedLp}. We will use the following proposition (see e.g. \cite[Proposition 8.1.2, Corollary 8.1.3]{effros2000operator}):
 \begin{prop}\label{thm:ruan}
 For any two operator spaces $V$ and $W$, the natural embedding $\theta:V^*{\otimes_{\min}} W\hookrightarrow \operatorname{CB}(V,W)$ defined by
 \begin{align}
     \theta(f\otimes w)(v)=f(v)\,w\,,\qquad \forall f\in V^*,\,w\in W,\,v\in V\,,
 \end{align}
 is completely isometric. In particular, for any operator space $W$ and $n\in\mathbb{N}$, $\mathbb{M}_n(W)\cong \mathbb{M}_n{\otimes_{\min}}\,W$.
 \end{prop}

 \subsection{Conditional expectations and subalgebra indices}\label{sec:indices}

 Let $\cN\subseteq \cM\subseteq  \cB(\cH)$ be two von Neumann subalgebras of $\cB(\cH)$. Recall that a conditional expectation onto $\cN$ is a completely positive unital map $E_\cN:\cM\to \cN$ satisfying \begin{enumerate}
\item[i)]for all $X\in \cN$, $E_\cN(X)=X$,
\item[ii)]for all  $a,b\in\cN,X\in \cM$, $E_\cN(aXb)=aE_\cN(X)b$  .
 \end{enumerate}
We denote by $E_{\cN*}$ its adjoint map with respect to the trace inner product, i.e.
\begin{align*}
\tr(E_{\cN*}(X)Y)=\tr(XE_\cN(Y))\,.
\end{align*}
For a state $\rho$, the relative entropy with respect to $\cN$ is defined as follows
\[D(\rho\|\mathcal{N}):=D(\rho\|E_{\cN*}(\rho))=\inf_{E_{\cN*}(\si)=\si} D(\rho\|\si)\,,\]
where the infimum is always attained by $E_{\cN*}(\rho)$. Indeed, for any $\si$ satisfying $E_{\cN*}(\si)=\si$,  we have the identity (see \cite[Lemma 3.4]{junge2019stability})
\[D(\rho\|\si)=D(\rho\|E_{\cN*}(\rho))+D(E_{\cN*}(\rho)\|\si)\,.\]
Hence the infimum is attained if and only if $D(E_{\cN*}(\rho)\|\si)=0$. More explicitly, a finite dimensional von Neumann (sub)algebra is always given by a direct sum of matrix algebras with multiplicity, i.e.
\[\cN=\bigoplus_{i=1}^n \cB(\cH_i)\ten \mathbb{C}\Id_{\cK_i}\,, ~~~~~~~ \cH=\bigoplus_{i=1}^n \cH_i\ten \cK_i\,.\]
Denote $P_i$ as the projection onto $\cH_i\ten \cK_i$. There exists a family of density operators $\tau_i\in \cD(\cK_i)$ such that
\begin{align}E_\cN(X)=\bigoplus_{i=1}^n \tr_{\cK_i}(P_iXP_i(\Id_{\cK_i}\ten \tau_i))\ten \Id_{\cK_i}\,, ~~~~E_{\cN*}(\rho)=\bigoplus_{i=1}^n \tr_{\cK_i}(P_i\rho P_i)\ten \tau_i\,, \label{cd}\end{align}
where $\tr_{\cK_i}$ is the partial trace with respect to $\cK_i$.
A state $\si$ satisfies $E_{\cN*}(\si)=\si$ if and only if
\[ \si=\bigoplus_{i=1}^n p_i\,\si_i\ten \tau_i\,\]
for some density operators $\si_i\in \cD(\cH_i)$ and a probability distribution $\{p_i\}_{i=1}^n$. Denote $\cD(E_\cN):=\{\si\in \cD(\cH) \, | \,  \si=E_{\cN*}(\si)\}$ as the subset of states that are invariant under $E_{\cN*}$. For any $\si\in \cD(E_\cN)$,
\begin{align*}
E_{\cN*}(\si^{\frac12}X\si^{\frac12})= \si^{\frac12}E_{\cN}(X)\si^{\frac12}\,. \end{align*}

Let $\cM\subset \cB(\cH)$ be a finite dimensional von Neumann algebra and $\cN\subset \cM$ be a subalgebra.
The trace preserving conditional expectation $E_{\cN,\tr}:\cM\to\cN$ is defined so that for any $ X\in \cM$ and $ Y\in\cN$,
\[ \tr(XY)=\tr(E_{\cN,\tr}(X)Y)\,.\]
$E_{\cN,\tr}$ is self-adjoint and corresponds to taking $\displaystyle \tau_i=d_{\cK_i}^{-1}\Id_{\cK_i}$ in \eqref{cd}. We recall the definition of the index associated to the algebra inclusion $\cN\subset\cM$,
\begin{align*}&C(\cM:\cN)=\inf\{ c>0\, | \, \rho\le  c\,E_{\cN,\tr}(\rho) \text{ for all states $\rho\in \M$} \}\,, \\ &C_{\operatorname{cb}}(\cM:\cN)=\sup_{n\in\mathbb{N}}C(\cM\ten \mathbb{M}_n:\cN\ten \mathbb{M}_n)\,,\end{align*}
where the supremum in $C_{\operatorname{cb}}(\cM:\cN)$ is taken over all finite dimensional matrix algebras $\mathbb{M}_n$.
The index $C(\cM:\cN)$ was first introduced by Pimsner and Popa in \cite{popapimser} for the connection to
subfactor index and Connes entropy, and the completely bounded version $C_{\operatorname{cb}}(\cM:\cN)$ was studied in \cite{gaoindex}.
These indices are closely related to the notion of maximal relative entropy. Recall that for two states $\rho,\omega$, their maximal relative entropy is \cite{datta2009min} \[D_{\max}(\rho\|\omega)=\ln \inf\{\, c>0 \,   | \, \rho\le c \,\omega \, \}\,.\]
Indeed, $$\displaystyle\ln C(\cM:\cN)=\sup_{\rho\in \cD(E_{\cM,\tr})} D_{\max}(\rho\|E_{\cN,\tr}(\rho))\,.$$ For all finite dimensional inclusion $\cN\subset\cM$, the index $C(\cM:\cN)$ is explicitly calculated in \cite[Theorem 6.1]{popapimser} (hence also $C_{\operatorname{cb}}(\cM:\cN)$). In particular, for $\cM=\cB(\cH)$ and $\cN=\bigoplus_{i=1}^n \cB(\cH_i)\ten \mathbb{C}\Id_{\cK_i}$,
\begin{align}
    C(\cB(\cH):\cN)=\sum_{i=1}^n \min\{d_{\cH_i},d_{\cK_i}\}\,d_{\cK_i}\,,~~~~~~ C_{\operatorname{cb}}(\cB(\cH):\cN)=\sum_{i=1}^n d_{\cK_i}^2\,. \label{formula}\end{align}
For example, if we take $\cD\subset \cB(\cH)$ to be the subalgebra of diagonal matrices and $\mathbb{C}$ as the multiple of identity
\begin{align}&C(\cB(\cH):\cD)=C_{\operatorname{cb}}(\cB(\cH):\cD)=d_\cH\,, \nonumber\\ &C(\cB(\cH):\mathbb{C})=d_\cH\,,~~~ C_{\operatorname{cb}}(\cB(\cH):\mathbb{C})=d_\cH^2\,.\label{eq:indicesex}
\end{align}
In \cite{gao2021spectral}, the authors considered a generalization of these indices for a general conditional expectation $E_{\cN}:\cM\to\cN$. We recall that here $\mathbb{M}_n$ is the $n$-dimensional matrix algebra and $E_{\cN}\ten \id_{\mathbb{M}_n}\equiv E_{\cN}\otimes \id_n$ is a conditional expectation from $\cM\ten \mathbb{M}_n\to\cN\ten \mathbb{M}_n$.
Denote \begin{align}\tau=\bigoplus_{i=1}^n \Id_{\cH_i}\ten \tau_i \,.\label{tau}\end{align} Note that
 $E_{\cN}$ and $E_{\cN*}$ are uniquely defined by $\tau$ as follows,
\begin{align}\label{eqENENtr}
E_{\cN}(X)= E_{\cN,\tr}(\tau^{\frac12}X\tau^{\frac12})\,,~~~~~~ E_{\cN*}(\rho)= \tau^{\frac12}E_{\cN,\tr}(\rho)\tau^{\frac12}.
\end{align}
In particular, $E_{\cN}$ is faithful if and only if $\tau$ is. Next, we define
\begin{align}&C_{\tau}(\cM:\cN):=\inf_{c}\{ c>0\, | \, \rho\le c\,E_{\cN*}(\rho) \text{ for all states $\rho\in \M$} \}\nonumber \\
&C_{\tau,\operatorname{cb}}(\cM:\cN):                            =\sup_{n\in\mathbb{N}}C_{\tau\otimes \Id_n}(\cM\otimes {\mathbb{M}_n}:\cN\otimes \mathbb{M}_n)\,. \label{index}\end{align}
 Since $\tau$ commutes with $\cN$,
\begin{align}\label{fromtracialtonontracial} C_\tau(\cM:\cN)\le \mu_{\operatorname{min}}(\tau)^{-1}C(\cM:\cN)\,,~~~~~~~~~~~ C_{\tau,\operatorname{cb}}(\cM:\cN)\le \mu_{\operatorname{min}}(\tau)^{-1}C_{\operatorname{cb}}(\cM:\cN)\end{align}
where $\mu_{\operatorname{min}}(\tau)=\min_{i}\mu_{\min}(\tau_i)$ is the minimal eigenvalue of $\tau$. Combined with \eqref{formula}, this implies $C_{\tau}(\cM:\cN)$ and $C_{\tau,\operatorname{cb}}(\cM:\cN)$ are finite iff $\tau$ is faithful. Moreover, for any invariant state $\sigma\in\cD(E_\cN)$, by the obvious bound $\sigma\le \tau$, we also have
\begin{align}\label{eq:CtoCsigma} C_\tau(\cM:\cN)\le \mu_{\operatorname{min}}(\sigma)^{-1}C(\cM:\cN)\,,~~~~~~~~~~~ C_{\tau,\operatorname{cb}}(\cM:\cN)\le \mu_{\operatorname{min}}(\sigma)^{-1}C_{\operatorname{cb}}(\cM:\cN)\,.\end{align}

\subsection{Bimodule maps and module Choi operators}\label{subsec:bimodule_maps}

Let $E_\cN:\cM\to \cN$ be a conditional expectation onto $\cN$. Recall that
there exists a module basis $\{\xi_i\}_{i=1}^n\in \cM$ satisfying (\cite[Theorem 3.15]{paschke1973inner}, see also \cite[Cons\'{e}quence 1.8]{baillet1988indice}):
\begin{align}\label{eq:modulebasis}
E_{\cN}(\xi_i^\dagger\xi_j)=\delta_{ij}p_i\,,
\end{align}
where $p_i\in \cN$ are some projections. Also recall that $\Phi:\cM\to \cM$ is a $\cN$-bimodule map if
$\Phi(aXb)=a\Phi(X)b$ for all $a,b\in \cN$ and $X\in \cM$. In particular, this implies
 $\Phi\circ E_\cN=E_{\cN}$ if $\Phi$ is unital. %and  $E_{\cN*}\circ \Phi=E_{\cN*}$ if $\Phi$ is trace preserving.
Next, we define the \textit{module Choi operator} of a bimodule map $\Phi$ as
\[ \chi_\Phi=\sum_{i,j=1}^n\ket{i}\bra{j}\ten \Phi(\xi_i^\dagger\xi_j)\in \cB(l_2^n)\ten \cM\,,\]
where $l_2^n$ denotes the space of $n$-dimensional vectors and $\{\ket{i}\}_{i=1}^n$ is a fixed orthonormal basis in $l_2^n$. Thus $\Phi$ and $\chi_\Phi$ determine each other because for each $x\in \cM$, we have a unique decomposition $x=\sum_{i}\xi_i x_i $ with $x_i\in\cN$ satisfying $p_ix_i=x_i$. Indeed, we have $x_i=E_\cN(\xi_i^\dagger x)$.
Moreover, $\Phi$ is completely positive if and only if $\chi_\Phi$ is a positive operator in $\cB(l_2^n)\ten \cM$. Indeed, for any finite family $y_1,\cdots,y_m\in \cM$, we assume the decomposition $y_j=\sum_{l}\xi_l x_{jl}$ with $x_{jl}\in \cN$. Then
\begin{align*}
(\id\ten \Phi)\big(\sum_{i,j}\ket{i}\bra{j}\ten y_i^\dagger y_j\big)
=& \sum_{i,j}\ket{i}\bra{j}\ten \Phi(y_i^\dagger y_j)
\\
=& \sum_{i,j,k,l}\ket{i}\bra{j}\ten x_{ik}^\dagger \Phi(\xi_k^\dagger \xi_l)x_{jl}
\\
=& \big(\sum_{i,k}\ket{i}\bra{k}\ten x_{ik}^\dagger \big)\chi_\Phi
\big(\sum_{j,l}\ket{l}\bra{j}\ten x_{jl}\big)\ ,
\end{align*}
from which the equivalence claimed directly follows. We remark that when $\cN=\mathbb{C}1$, $\chi_\Phi$ is the standard Choi matrix (up to a unitary equivalence).

\section{The main result}\label{sec:main_result}

In this section, we present the main result of the paper, namely the existence of a positive MLSI constant for any Davies generator in 1D converging to the Gibbs state of a finite-range, translation-invariant and commuting Hamiltonian, at any temperature, with a logarithmic dependence with the system size. We leave the specifics of the more technical parts of the proof, Lemma \ref{lemma:globaltolocal} and Lemma \ref{lemmatechnical1} below, to Sections \ref{sec:mixingcondition} and \ref{sec:localcontrol}, respectively, for sake of clarity.

\begin{thm}\label{thm:main_result_formal}
Let $\Lambda= \llbracket 1,n\rrbracket$. For any $\beta>0$, we denote by $\sigma \equiv \sigma^\beta$ the Gibbs state of a finite-range, translation-invariant, commuting Hamiltonian at inverse temperature $\beta>0$. Consider $\mathcal{L}^D_{\Lambda*}$ the Davies generator of a quantum Markov semigroup $\{e^{t\cL_{\Lambda*}^D}\}_{t \geq 0}$ with unique fixed point $\sigma$. Then, there exists $\alpha_n=\Omega(\ln(n)^{-1})$ such that, for all $\rho\in\cD(\cH_\Lambda)$ and all $t\ge 0$,
\begin{align}\label{eq:entropic_decay_appendix}
D(\rho_t\|\sigma)\le e^{-\alpha_n t}\,D(\rho\|\sigma)\,,
\end{align}
where $\rho_t:=e^{t\cL_{\Lambda*}^D}(\rho)$. Moreover, $\alpha_n=e^{-\mathcal{O}(\beta)}$ as a function of $\beta$.
\end{thm}

\begin{rem}
The scaling of the MLSI constant with the inverse temperature is optimal, as proved in \cite{nacu2003glauber}.
\end{rem}

Note that the previous result is equivalent to the existence of a positive MLSI constant $\alpha_n=\Omega(\ln(n)^{-1})$ for $\mathcal{L}^D_{\Lambda*}$ (cf. Section \ref{subsec:functineq}): Indeed, \eqref{eq:entropic_decay_appendix} holds if, and only if, for all $\rho\in\cD(\cH_\Lambda)$ and all $t\ge 0$,
\begin{equation}\label{eq:global_MLSI_definition}
    \alpha_n D(\rho \|\sigma)\le -\tr[\mathcal{L}^D_{\Lambda*}(\rho) (\ln \rho - \ln \sigma)]=:\operatorname{EP}(\rho) \, ,
\end{equation}
where $\alpha_n=\mathcal{O}(\ln(n))$ and the entropy production in the right-hand side of the inequality is obtained by
\begin{equation*}
        \operatorname{EP}(\rho)=-\left.\frac{d}{dt}\right|_{t=0} D(\rho_t\| \sigma) \, .
\end{equation*}
One way to prove \eqref{eq:global_MLSI_definition}, and therefore entropic convergence of the Davies semigroup $(e^{t\cL^D_{\Lambda*}})_{t\ge0}$ as claimed in Theorem \ref{thm:main_result_formal}, is by reducing the MLSI constant in $\Lambda$ to the complete MLSI constants in smaller regions $A, B \subset \Lambda$. This procedure was introduced in \cite{cesi2001quasi} and \cite{daipra2002classicalMLSI} to simplify the traditional strategy for proving positivity of MLSI constants for classical spin systems. In the past few years, a similar idea has been explored in a number of works regarding quantum spin lattice systems \cite{capel2018quantum, bardet2019modified, capel2019thesis, bardet2020approximate, capel2020MLSI, gao2021spectral}, in which   a strategy to prove positivity of a MLSI constant via results of quasi-factorization (a.k.a. approximate tensorization) of the relative entropy has been conceived. The proof of Theorem \ref{thm:main_result_formal} also follows this direction.

In the next few lines, we provide an intuition for the reduction from MLSI constant in $\Lambda$ to complete MLSI constants in $A,B \subset \Lambda$. Given any sub-region $A\subseteq \Lambda$, we denote the projection onto the fixed points of $\cL^D_A$ by $E_A$. Its adjoint map $E_{A*}$ is a quantum channel that corresponds to the infinite time limit of the semigroup generated by $\cL_{A*}^D$, i.e. $e^{t\cL_{A*}^D}\to E_{A*}$ as $t\to\infty$. Then, $e^{t\cL_{A*}^D}$ satisfies a positive complete MLSI if there exists $\alpha_A > 0$ such that for any $m\in \mathbb{N}$ and all states $\rho\in \mathcal{D}(\cH_\Lambda\otimes \mathbb{C}^m)$,
\begin{align}\label{eq:CMLSI}
   \alpha_A\,D(\rho\|(E_{A*}\otimes \id_m)(\rho))\le \operatorname{EP}_{A}(\rho)\,,
\end{align}
where $\operatorname{EP}_A(\rho)$ denotes the entropy production of $\rho$ in $A$ and is defined as
\begin{align*}
    \operatorname{EP}_A(\rho):=-\left.\frac{d}{dt}\right|_{t=0} D(\rho^A_t\|(E_{A*}\otimes \id_m)(\rho))
\end{align*}
for $\rho^A_t:=(e^{t\cL_{A*}^D}\otimes \id_m)(\rho)$. Indeed, \eqref{eq:CMLSI} is equivalent to the existence of a uniform MLSI constant for the semigroup $e^{t\cL_{A*}^D}$ coupled with environment of all dimensions.  By a straightforward derivation, the entropy production in $A$ can be shown to be equal to
$$\operatorname{EP}_A(\rho)=-\tr\Big[\cL_{A*}^D(\rho)\big( \ln (\rho)-\ln(\sigma\otimes \tr_\Lambda(\rho))\big)\Big]\, $$
where $\tr_\Lambda(\rho)$ is the reduced density on the environment system.
Note that, in particular, for any two non-overlapping regions $A,B\subseteq \Lambda$,
\begin{align}\label{eq:additivity}
    \operatorname{EP}_{A\cup B}(\rho)=  \operatorname{EP}_{A}(\rho)+  \operatorname{EP}_{ B}(\rho)\,.
\end{align}
This simple yet important observation regarding the linearity of the entropy production (inherited by the linearity of the generator) allows us to reduce the right-hand side of Equation \eqref{eq:global_MLSI_definition} from a chain $\Lambda$ to two sub-regions $A, B$ such that $A \cup B = \Lambda$. The remaining part of the reduction procedure concerns the splitting of the left-hand side, namely the relative entropy. For that, we need the following two-steps strategy:
\begin{itemize}
\item[(i)] \textbf{Global-to-local reduction} (\Cref{subsec:global-to-local}): We first prove the following \textit{approximate tensorization} of the relative entropy that for any state $\rho\in \cD(\cH_\Lambda)$,
%on the right-hand side of \eqref{eq:CMLSI} for $A=\Lambda$ and $m=1$:
\begin{align}\label{eq:AT}
D(\rho\|\sigma)    \le \mathcal{C} \sum_{i}D(\rho\|E_{A_i*}(\rho))\,,
\end{align}
for some universal constant $\mathcal{C} \ge 1$ and a covering $\{A_i\}_i$ of the chain $\Lambda$ by intervals $A_i$ of size $|A_i|=O(\ln(n))$. This turns out to be a consequence of a generic approximate tensorization for quantum states introduced in \cite{capel2018quantum,capel2017superadditivity} (see also \cite{bardet2019modified}) and a recently proven mixing property of $1$D quantum Gibbs states at any temperature \cite{bluhm2021exponential}.
\item[(ii)] \textbf{Local control of the constant} (\Cref{sec:quasilocalcontrol}): We then prove that the complete MLSI \eqref{eq:CMLSI} holds for the local subregions $A_i$ with a constant $\alpha_A:=\min_i\,\alpha_{A_i}=\min_i\Omega (|A_i|^{-1})$. In order to prove this second step, we will in fact resort to another approach of approximate tensorization for conditional expectations recently studied in \cite{laracuente2019quasi,gao2021spectral}.
\end{itemize}

The next two subsections detail the plan drawn above.

% Before that, let us conclude with a direct corollary to the previous result and \cite[Theorem 22]{KastoryanoTemme-LogSobolevInequalities-2013} regarding the mixing time of an evolution generated by a Davies generator in 1D.

% \begin{cor}\label{cor:rapid_mixing_formal}
% Let $\Lambda \subset \subset \mathbb{Z}$ with $|\Lambda|=n$. For any $\beta>0$, we denote by $\sigma \equiv \sigma^\beta$ the Gibbs state of a local, finite-range, translation-invariant, commuting Hamiltonian at inverse temperature $\beta$. Consider $\mathcal{L}^D_{\Lambda*}$ the Davies generator of a quantum Markov semigroup $\{e^{t\cL_{\Lambda*}^D}\}_{t \geq 0}$ with unique fixed point $\sigma$. Then, $\mathcal{L}^D_{\Lambda*}$ has rapid mixing.
% \end{cor}

\subsection{Global-to-local reduction}\label{subsec:global-to-local}

The first step in the proof of \eqref{eq:AT} consists in the following abstract approximate tensorization for the relative entropy proved in \cite{capel2018quantum,capel2017superadditivity} (see also \cite{bardet2019modified}). For any finite lattice $\Lambda \subset \subset \mathbb{Z}^d$ in any dimension $d$, consider any two regions $A,B\subset \Lambda$ with non-overlapping complements and such that $A \cup B = \Lambda$, as in the following picture:
\begin{figure}[H]
\begin{center}
   \begin{tikzpicture}[scale=0.7]
\Block[6,lightblue!60!white,1];
\begin{scope}[yshift=1cm]
\Block[6,lightblue!60!white,1];
\end{scope}
\begin{scope}[yshift=2cm]
\Block[6,lightblue!60!white,1];
\end{scope}
\begin{scope}[yshift=3cm]
\Block[6,lightblue!60!white,1];
\draw [thick,lightblue,decorate,decoration={brace,amplitude=10pt,mirror},xshift=-0.5pt,yshift=-0.6pt](-0.5,-3.6) -- (7.5,-3.6) node[black,midway,yshift=-0.7cm] { \textcolor{lightblue}{$A$}};
\end{scope}
\begin{scope}[xshift=6cm]
\Block[2,lightpurple!40!white,1];
\end{scope}
\begin{scope}[xshift=6cm,yshift=1cm]
\Block[2,lightpurple!40!white,1];
\end{scope}
\begin{scope}[xshift=6cm,yshift=2cm]
\Block[2,lightpurple!40!white,1];
\end{scope}
\begin{scope}[xshift=6cm,yshift=3cm]
\Block[2,lightpurple!40!white,1];
\end{scope}
\begin{scope}[xshift=8cm]
\Block[6,lightpink!50!white,1];
\end{scope}
\begin{scope}[xshift=8cm,yshift=1cm]
\Block[6,lightpink!50!white,1];
\end{scope}
\begin{scope}[xshift=8cm,yshift=2cm]
\Block[6,lightpink!50!white,1];
\end{scope}
\begin{scope}[xshift=8cm,yshift=3cm]
\Block[6,lightpink!50!white,1];
\draw [thick,lightpink,decorate,decoration={brace,amplitude=10pt},xshift=-0.5pt,yshift=-0.6pt](-2.5,0.6) -- (5.5,0.6) node[black,midway,yshift=0.7cm] { \textcolor{lightpink}{$B$}};
\end{scope}
\node at (16,3.8) {\huge $\Lambda$};
\end{tikzpicture}
\end{center}
  \caption{Possible splitting of a lattice $\Lambda$ into two subregions $A, B$ with non-overlapping complements and such that $\Lambda = A \cup B$.}
  \label{fig:0}
\end{figure}
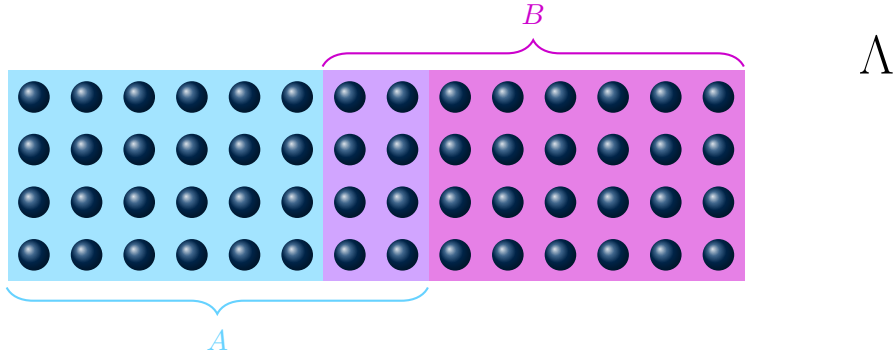
\noindent Given $\sigma  \in \cD(\cH_\Lambda)$, let us denote
\begin{equation}\label{eq:definition_h_mixing_condition}
    h(\sigma_{A^cB^c}):=\left(\sigma_{A^c}^{-1/2}\otimes \sigma_{B^c}^{-1/2}\right)\sigma_{A^cB^c}\left(\sigma_{A^c}^{-1/2}\otimes \sigma_{B^c}^{-1/2}\right)-\Id_{A^cB^c} \, .
\end{equation}
Then, for all states $\rho, \sigma \in \cD(\cH_\Lambda)$ such that $\norm{ h(\sigma_{A^cB^c})}_\infty < 1/2$, the following holds \cite{capel2018quantum}:
\begin{align}\label{ATmixingcondrel}
    D(\rho \| \sigma )\le \frac{1}{1-2\,\|h(\sigma_{A^cB^c})\|_\infty}\,\big[D_A(\rho \| \sigma )+D_B(\rho \| \sigma )\big]\,,
\end{align}
where for any region $C\subseteq \Lambda$, we define the \textit{conditional relative entropy between $\rho$ and $\sigma$ on $C$} as
\begin{align*}
    D_C(\rho \| \sigma ):=D(\rho\|\sigma)-D(\rho_{C^c}\|\sigma_{C^c})\,.
\end{align*}
The operator $h(\sigma_{A^cB^c})$, which constitutes a natural quantum generalization of the mixing condition of \cite{daipra2002classicalMLSI}, is a measure of independence of the regions $A$ and $B$ as measured in the Gibbs state $\sigma$.
Now, we focus on a unidimensional $\Lambda \subset \subset \mathbb{Z}$ and, following the lines of \cite{bardet2019modified}, we construct $A, B$ as two non-connected regions. Indeed, we choose the regions $A$ and $B$ to be unions of small intervals $\cup_{i=1}^m A_i$ and $\cup_{j=1}^m B_j$ respectively, such that subregions $A_i$ (resp.~$B_j$) do not overlap with each other, and $A\cup B$ covers the spin chain $\Lambda$. More specifically, for a certain $\ell \in \mathbb{N}$ to be determined later, we assume the following conditions on $\{ A_i \}_{i=1}^m $ and $\{ B_j \}_{j=1}^m $:
\begin{itemize}
    \item $\abs{A_i}=\abs{B_j} = 2 (r + \ell) -1 \, $ for every $1\leq i, j \leq m$, where $r$ is the range of the interaction of the Hamiltonian.
    \item $\abs{A_i \cap B_i} = \abs{B_i \cap A_{i+1}}= \ell \, $ for every $1 \leq i \leq m-1$.
\end{itemize}
For a better intuition on this construction, see Figure \ref{fig:1}:
\begin{figure}[H]
\begin{center}
   \begin{tikzpicture}[scale=0.7]
\Block[4,lightblue!60!white,1];
\draw [thick,midblue,decorate,decoration={brace,amplitude=10pt,mirror},xshift=-0.5pt,yshift=-0.6pt](-0.5,-0.6) -- (4.5,-0.6) node[black,midway,yshift=-0.7cm] { \textcolor{midblue}{$A_1$}};
\begin{scope}[xshift=4cm]
\Block[1,lightpurple!40!white,1];
\end{scope}
\begin{scope}[xshift=5cm]
\Block[3,lightpink!50!white,1];
\draw [thick,midpink,decorate,decoration={brace,amplitude=10pt},xshift=-0.5pt,yshift=-0.6pt](-1.5,0.6) -- (3.5,0.6) node[black,midway,yshift=+0.7cm] { \textcolor{midpink}{$B_1$}};
\end{scope}
\begin{scope}[xshift=8cm]
\Block[1,lightpurple!40!white,1];
\end{scope}
\begin{scope}[xshift=9cm]
\Block[3,lightblue!60!white,1];
\draw [thick,midblue,decorate,decoration={brace,amplitude=10pt,mirror},xshift=-0.5pt,yshift=-0.6pt](-1.5,-0.6) -- (3.5,-0.6) node[black,midway,yshift=-0.7cm] { \textcolor{midblue}{$A_2$}};
\end{scope}
\begin{scope}[xshift=12cm]
\Block[1,lightpurple!40!white,1];
\end{scope}
\begin{scope}[xshift=13cm]
\Block[4,lightpink!50!white,1];
\draw [thick,midpink,decorate,decoration={brace,amplitude=10pt},xshift=-0.5pt,yshift=-0.6pt](-1.5,0.6) -- (3.5,0.6) node[black,midway,yshift=+0.7cm] { \textcolor{midpink}{$B_2$}};
\end{scope}
\node at (18,0.8) {\huge $\Lambda$};
\end{tikzpicture}
\end{center}
  \caption{Representation of an interval $\Lambda$ split into two regions $A$ and $B$, which are the union of small intervals  $\cup_{i=1}^m A_i$ and $\cup_{j=1}^m B_j$, such that $A_i \cap B_i \neq \emptyset \neq B_i \cap A_{i+1}$ for all $i = 1, \ldots m-1$ and subregions $A_i \cap A_j = \emptyset = B_i \cap B_j$ for all $i \neq j$. In the picture, $m=2$, $r=2$ and $\ell=1$.}
  \label{fig:1}
\end{figure}
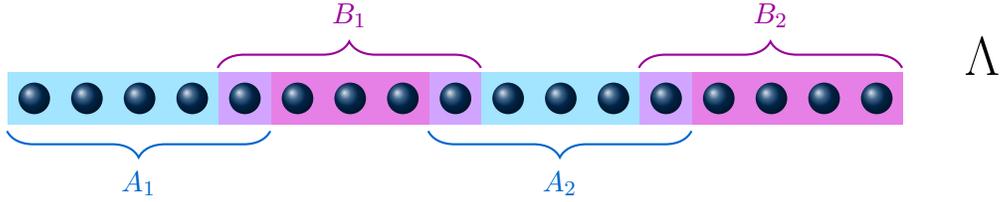

Next, using the geometry of the chain that we have just introduced, we are able to further upper bound the conditional relative entropies at the right-hand side of \eqref{ATmixingcondrel}. Indeed, since the $A_i$'s and $B_j$'s have been defined so that their boundaries do not overlap,  if we consider the splitting of the chain as $A_i \leftrightarrow \partial (A_i) \leftrightarrow (A_i \partial)^c$ for any $1 \leq i \leq m$ (and analogously for $B_j$), the Gibbs state $\sigma$ is, in particular, a quantum Markov chain between these three regions, as shown in Figure \ref{fig:2}:
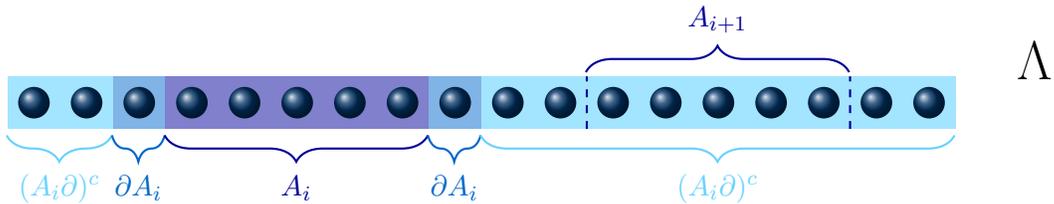
\begin{figure}[H]
\begin{center}
   \begin{tikzpicture}[scale=0.7]
\Block[2,lightblue!60!white,1];
\draw [thick,lightblue,decorate,decoration={brace,amplitude=10pt,mirror},xshift=-0.5pt,yshift=-0.6pt](-0.5,-0.6) -- (1.5,-0.6) node[black,midway,yshift=-0.7cm] { \textcolor{lightblue}{$(A_i \partial )^c$}};
\begin{scope}[xshift=2cm]
\Block[1,midblue!50!white,1];
\draw [thick,midblue,decorate,decoration={brace,amplitude=10pt,mirror},xshift=-0.5pt,yshift=-0.6pt](-0.5,-0.6) -- (0.5,-0.6) node[black,midway,yshift=-0.7cm] { \textcolor{midblue}{$\partial A_i  $}};
\end{scope}
\begin{scope}[xshift=3cm]
\Block[5,darkblue!50!white,1];
\draw [thick,darkblue,decorate,decoration={brace,amplitude=10pt,mirror},xshift=-0.5pt,yshift=-0.6pt](-0.5,-0.6) -- (4.5,-0.6) node[black,midway,yshift=-0.7cm] { \textcolor{darkblue}{$A_i$}};
\end{scope}
\begin{scope}[xshift=8cm]
\Block[1,midblue!50!white,1];
\draw [thick,midblue,decorate,decoration={brace,amplitude=10pt,mirror},xshift=-0.5pt,yshift=-0.6pt](-0.5,-0.6) -- (0.5,-0.6) node[black,midway,yshift=-0.7cm] { \textcolor{midblue}{$\partial A_i  $}};
\end{scope}
\begin{scope}[xshift=9cm]
\Block[9,lightblue!60!white,1];
\draw [thick,lightblue,decorate,decoration={brace,amplitude=10pt,mirror},xshift=-0.5pt,yshift=-0.6pt](-0.5,-0.6) -- (8.5,-0.6) node[black,midway,yshift=-0.7cm] { \textcolor{lightblue}{$(A_i \partial )^c$}};
\draw [thick,darkblue,decorate,decoration={brace,amplitude=10pt},xshift=-0.5pt,yshift=-0.6pt](1.5,0.6) -- (6.5,0.6) node[black,midway,yshift=0.7cm] { \textcolor{darkblue}{$A_{i+1}$}};
\draw [dashed,darkblue,thick](1.5,-0.5) -- (1.5,0.5) ;
\draw [dashed,darkblue,thick](6.5,-0.5) -- (6.5,0.5) ;
\end{scope}
\node at (19,0.8) {\huge $\Lambda$};
\end{tikzpicture}
\end{center}
  \caption{Splitting of $\Lambda$ into the three regions $A_i \leftrightarrow \partial (A_i) \leftrightarrow (A_i \partial)^c$. Note that, in particular, any other $A_j$ satisfies $A_j \subset (A_i \partial)^c$. }
  \label{fig:2}
\end{figure}

This property yields a privileged structural decomposition of $\sigma$ of the following form:
\begin{equation}\label{eq:decomposition_sigma_Gibbs_state}
    \sigma = \underset{j}{\bigoplus} \, q_j \, \sigma_{A_i \partial (A_i)_j^L  } \otimes \sigma_{\partial (A_i)_j^R (A_i \partial )^c  } \, .
\end{equation}
This decomposition allowed some of the authors to show in \cite[Step 2 of Theorem 7]{bardet2019modified} the following inequality for such a state $\sigma$:
\begin{align}\label{ATmixingcondrel2}
    D(\rho \| \sigma )\le \frac{1}{1-2\,\|h(\sigma_{A^cB^c})\|_\infty}\, \underset{i=1}{\overset{m}{\sum}} \big[D_{A_i}(\rho \| \sigma )+D_{B_i}(\rho \| \sigma )\big]\,.
\end{align}
Now, we aim at estimating the multiplicative term appearing in \eqref{eq:additivity} (and in \eqref{ATmixingcondrel2}) given this geometry. Intuitively, the quantity $\|h(\sigma_{A^cB^c})\|_\infty$ should decrease with the size of the overlap between $A$ and $B$, since in this case their complements get more separated. This intuition was recently given a rigorous justification in \cite[Proposition 8.1]{bluhm2021exponential}. There, building up on the seminal result of Araki \cite{araki1969gibbs}, the authors prove the following \textit{mixing condition}:  Consider any three convex regions $X,Y,Z\subset\mathbb{Z}$ with $Y$ shielding $X$ away from $Z$ as in Figure \ref{fig:3},
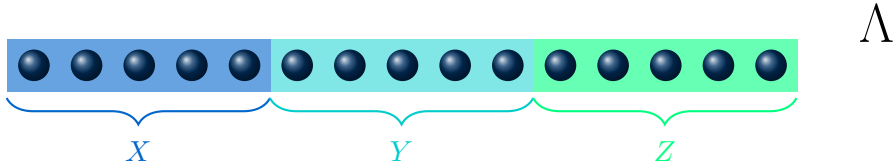
\begin{figure}[H]
\begin{center}
   \begin{tikzpicture}[scale=0.7]
\Block[5,midblue!60!white,1];
\draw [thick,midblue,decorate,decoration={brace,amplitude=10pt,mirror},xshift=-0.5pt,yshift=-0.6pt](-0.5,-0.6) -- (4.5,-0.6) node[black,midway,yshift=-0.7cm] { \textcolor{midblue}{$X$}};
\begin{scope}[xshift=5cm]
\Block[5,cyan!50!white,1];
\draw [thick,cyan,decorate,decoration={brace,amplitude=10pt,mirror},xshift=-0.5pt,yshift=-0.6pt](-0.5,-0.6) -- (4.5,-0.6) node[black,midway,yshift=-0.7cm] { \textcolor{cyan}{$Y $}};
\end{scope}
\begin{scope}[xshift=10cm]
\Block[5,lightgreen!60!white,1];
\draw [thick,lightgreen,decorate,decoration={brace,amplitude=10pt,mirror},xshift=-0.5pt,yshift=-0.6pt](-0.5,-0.6) -- (4.5,-0.6) node[black,midway,yshift=-0.7cm] { \textcolor{lightgreen}{$Z$}};
\end{scope}
\node at (16,0.8) {\huge $\Lambda$};
\end{tikzpicture}
\end{center}
  \caption{Splitting of a finite interval $\Lambda$ into three subintervals  $X,Y,Z$ with $Y$ shielding $X$ away from $Z$. }
  \label{fig:3}
\end{figure}
\noindent Then, given the Gibbs state of a possibly non-commuting, finite-range, translation-invariant 1D Hamiltonian, for any inverse temperature $\beta > 0$ there exist constants $\mathcal{K}\ge 0$ and $\gamma >0$ independent of $n$ such that,
\begin{align}
 \|\sigma_{XZ}(\sigma_{X}\otimes \sigma_{Z})^{-1}-\Id_{XZ}\|_\infty \le \mathcal{K}\,e^{-\gamma |Y|}\,,\label{eq:mixing}
\end{align}
Moreover, note that by \cite[Proposition IX.1.1]{bhatia2013matrix}, we have for any $P > 0$ and any real observable $Q$ that
\begin{equation}\label{equa:boundBSrelativeEntropyAux1}
\| P^{1/2} Q P^{1/2} - \identity\| \,\, \leq \,\, \| Q P - \identity\|\, .
\end{equation}
Therefore, putting both inequalities together, we have
\begin{align}
      \|h(\sigma_{XZ})\|_\infty&\le \mathcal{K}\,e^{-\gamma |Y|}\, \label{mixing2}
\end{align}

With the geometry $\Lambda=A\cup B$ described above, this bound can be used in cascade in order to yield a good enough control over $\|h(\sigma_{A^cB^c})\|_\infty$. This is the content of the following Lemma, whose proof we defer to Section \ref{sec:mixingcondition}.

\begin{lem}\label{lemma:globaltolocal}Let $\{A_i , B_i \}_{i=1}^m$ be a covering  of $\Lambda$ with $m=\Omega(n/\ln(n))$ intervals of size $|A_i|= |B_i|=\mathcal{O}(\ln(n))$ for every $1 \leq i \leq m$.
For any inverse temperature $\beta$, there exists a constant $\mathcal{C}$ independent of $|\Lambda|\equiv n$ such that for any state $\rho\in \cD(\cH_\Lambda)$,
\begin{align*}
    D(\rho\|\sigma)\le \mathcal{C} \sum_{i=1}^{m} \left[ D_{A_i}(\rho\| \sigma ) + D_{B_i}(\rho\| \sigma ) \right]\,,
\end{align*}
%for a covering $\{A_i , B_i \}_{i=1}^m$ of $\Lambda$ with $m=\Omega(n/\ln(n))$ intervals of size $|A_i|= |B_i|=\mathcal{O}(\ln(n))$ for every $1 \leq i \leq m$.
\end{lem}
To conclude the proof of \eqref{eq:AT}, we recall \cite[Proposition 5]{bardet2020approximate}, in which it was proven that the conditional relative entropy between $\rho$ and $\sigma$ in a region $A \subset \Lambda$ is upper bounded by the relative entropy between $\rho$ and its conditional expectation onto $A$, i.e.
\begin{align}\label{eq:ineq_conditional_relative_entropies}
    D_A(\rho \| \sigma)  \leq D( \rho \| E_{A^*}(\rho)) \, .
\end{align}
This, together with Lemma \ref{lemma:globaltolocal}, allows us to conclude the following approximate tensorization:
\begin{lem}\label{globaltolocalprop}
In the notations of \Cref{lemma:globaltolocal}, we have that for any state $\rho\in\cD(\cH_\Lambda)$,
\begin{align}\label{eq:final_global_to_local}
    D(\rho\|\sigma)\le \mathcal{C} \sum_{i=1}^{m} \left[ D( \rho \| E_{A_i^*}(\rho)) + D( \rho \| E_{B_i^*}(\rho)) \right]\,.
\end{align}
\end{lem}
\begin{proof}
The result follows directly from \Cref{lemma:globaltolocal} and  \eqref{eq:ineq_conditional_relative_entropies}.
\end{proof}

\subsection{Quasi-local control of the constant}\label{sec:quasilocalcontrol}

In the next step of the proof, we need to reduce the conditional expectations in the relative entropies on the  right-hand side of \eqref{eq:final_global_to_local} to single-site conditional expectations in each of the sites composing the region where the latter was conditioning. Inspired by the work of \cite{laracuente2019quasi}, some of the authors provide in
\cite[Corollary 5.5]{gao2021spectral} that, in the case of tracial conditional expectations $\{\tilde{E}_B\}_{B\subseteq \Lambda}$ and for any $A \subset \Lambda$,
\begin{align}\label{eq:corollaryAT}
    D&(\rho\|\tilde{E}_{A*}(\rho))\le 4k_A\,\sum_{i\in A}\,D(\rho\|\tilde{E}_{i*}(\rho))\, ,
    \end{align}
whenever $k_A\in\mathbb{N}$ satisfies
\begin{align}
   \frac{1}{2}\,\tilde{E}_A\le_{\operatorname{cp}} \Big(\prod_{i\in A} \tilde{E}_i\Big)^{k_A}\le_{\operatorname{cp}} \frac{3}{2}\,\tilde{E}_A \,, \label{eq:orderinequality}
\end{align}
where $\prod_{i\in A} \tilde{E}_i$ is a product of conditional expectations in an arbitrary ordering, and $\le_{\operatorname{cp}}$ stands for the completely positive partial order. In \Cref{sec:localcontrol}, we make use of operator space theory methods together with spectral gap estimates to (i) extend \eqref{eq:corollaryAT} to non-tracial conditional expectations like the ones corresponding to the infinite time limit of local Davies semigroups, and (ii) further control the integer $k_A$ appearing in \eqref{eq:corollaryAT} (see \Cref{lemmatechnicalgeneral} for technical details).

\begin{lem}\label{lemmatechnical1}
 The approximate tensorization \eqref{eq:corollaryAT} is satisfied for
 \begin{align}\label{eq:kA}
 k_A=\bigg\lceil\frac{\ln(2d^{2|A\partial|}e^{2\beta|A\partial|J})}{-\ln \lambda}\bigg\rceil\,,
 \end{align}
where $J$ is the interaction strength of $H_\Lambda$ and
$$\lambda:=\norm{\prod_{i\in A}E_i-E_A:L_2(\sigma)\to L_2(\sigma)}<1\,.$$
is a constant independent of the system size $n=|\Lambda|$.
\end{lem}

For $|A|=\mathcal{O}(\ln (n))$, the constant $k_A$ in \eqref{eq:corollaryAT} scales logarithmically with the system size $n$ if the constant $\lambda$ is independent of $n$. As we show in Section \ref{sec:localcontrol}, this is a direct consequence of the non-closure of the spectral gap proved for 1D commuting Gibbs samplers in \cite{kastoryano2016quantum} and the detectability lemma, which precisely relates the gap of a commuting, finite-range Gibbs sampler to $\lambda$ \cite{aharonov2009detectability,anshu2016simple,kastoryano2016quantum,gao2021spectral}. Moreover, the exponential dependence of $k_A$ on the inverse temperature $\beta$ causes the scaling of $\alpha_n=e^{-\mathcal{O}(\beta)}$ stated in \Cref{thm:main_result_formal}.

\subsection{Merging global and quasi-local analysis}
Theorem \ref{thm:main_result_formal} is now a simple consequence of the reasoning provided in the last two subsections, and of Lemmas \ref{globaltolocalprop} and \ref{lemmatechnical1} in particular. Indeed, putting both lemmas together, we have proven that
\begin{align}\label{eq:ATtot}
    D(\rho\|\sigma)\le \mathcal{O}(\ln(n))\sum_{i\in\Lambda}D(\rho\|E_{i*}(\rho))\,.
\end{align}
Moreover, it was proved in \cite[Theorem 3.3]{gao2021spectral} that the local generators $\cL_k^D$ always satisfy a complete modified logarithmic Sobolev inequality (cf. \eqref{eq:CMLSI}). That is, there exists a constant $\alpha_0> 0$ such that for all $i\in \Lambda$ and any state $\rho\in\cD(\cH_\Lambda)$,
\begin{align}\label{eq:onesiteMLSI}
 \alpha_0   \,D(\rho\|E_{i*}(\rho))\le \operatorname{EP}_{i}(\rho)\,.
\end{align}
Combining the bounds \eqref{eq:onesiteMLSI} and \eqref{eq:ATtot} together with the additivity of the entropy production \eqref{eq:additivity}, we conclude the existence of a constant $\alpha_n=\alpha_0\,\Omega(\ln(n)^{-1})=\Omega(\ln(n)^{-1})$ such that
\begin{align*}
\forall\,\rho \in \cD (\cH_\Lambda),\,\alpha_n  \, D(\rho\|\sigma)\le \operatorname{EP}_\Lambda(\rho)~~~~~\Rightarrow~~~~~\forall\,\rho \in \cD (\cH_\Lambda),\, D(e^{t\cL^D_{\Lambda*}}(\rho)\|\sigma)\le e^{-\alpha_nt}D(\rho\|\sigma)\,.
\end{align*}
This concludes the proof of Theorem \ref{thm:main_result_formal}.
\begin{flushright}
   \qedsymbol
\end{flushright}

\section{Proof of the mixing condition (Lemma \ref{lemma:globaltolocal})}\label{sec:mixingcondition}

This section is devoted to the proof of Lemma \ref{lemma:globaltolocal}. Let us first recall the construction devised in Section \ref{subsec:global-to-local} for the covering of $\Lambda \subset \mathbb{Z}$. We consider two regions $A$ and $B$ composed of small intervals, $A:= \cup_{i=1}^m A_i$, $B:= \cup_{j=1}^m B_j$ such that $A\cup B$ covers the spin chain $\Lambda$ and the following holds for a certain $\ell \in \mathbb{N}$:
\begin{itemize}
\item $A_i \cap A_j = B_i \cap B_j = \emptyset$ for all $1 \leq i, j \leq m$.
    \item $\abs{A_i}=\abs{B_j} = 2 (r + \ell) -1 \, $ for every $1\leq i, j \leq m$, where $r$ is the range of the interaction of the Hamiltonian.
    \item $\abs{A_i \cap B_i} = \abs{B_i \cap A_{i+1}}= \ell \, $ for every $1 \leq i \leq m-1$.
\end{itemize}
Next, let us write $C:= B^c$ and $D:= A^c$. Note that both of them are composed of $m$ disjoint segments, namely $C:= \cup_{i=1}^m \, C_i$ and $D:= \cup_{i=1}^m \, D_i$, respectively. Moreover, for every $i= 1 , \dots , m$ (resp. $i=1, \ldots, m-1$) let us denote by $E_i$, resp. $F_i$, the connected set that separates $C_i$ from $D_i$, resp. $D_i$ from $C_{i+1}$. See this construction in Figure \ref{fig:4}.
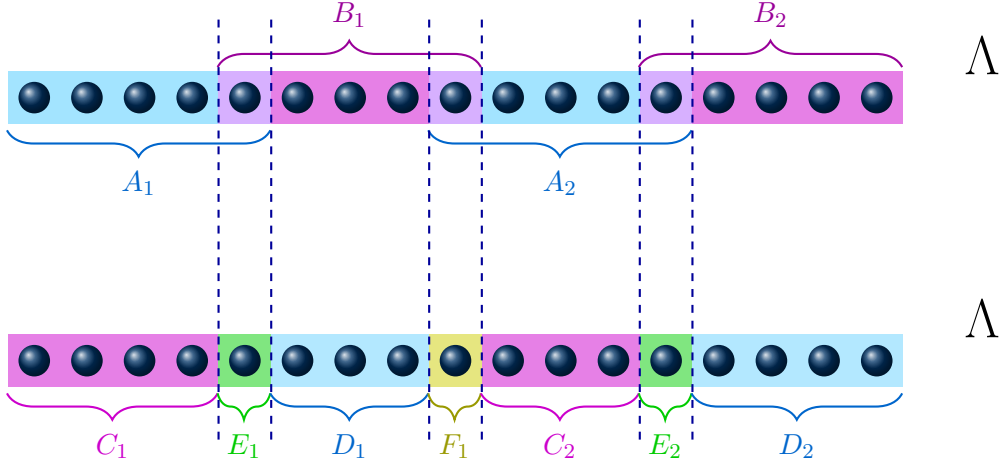
\begin{figure}[H]
\begin{center}
   \begin{tikzpicture}[scale=0.7]
\Block[4,lightblue!60!white,1];
\draw [thick,midblue,decorate,decoration={brace,amplitude=10pt,mirror},xshift=-0.5pt,yshift=-0.6pt](-0.5,-0.6) -- (4.5,-0.6) node[black,midway,yshift=-0.7cm] { \textcolor{midblue}{$A_1$}};
\begin{scope}[xshift=4cm]
\Block[1,lightpurple!35!white,1];
\end{scope}
\begin{scope}[xshift=5cm]
\Block[3,lightpink!50!white,1];
\draw [thick,midpink,decorate,decoration={brace,amplitude=10pt},xshift=-0.5pt,yshift=-0.6pt](-1.5,0.6) -- (3.5,0.6) node[black,midway,yshift=+0.7cm] { \textcolor{midpink}{$B_1$}};
\end{scope}
\begin{scope}[xshift=8cm]
\Block[1,lightpurple!35!white,1];
\end{scope}
\begin{scope}[xshift=9cm]
\Block[3,lightblue!60!white,1];
\draw [thick,midblue,decorate,decoration={brace,amplitude=10pt,mirror},xshift=-0.5pt,yshift=-0.6pt](-1.5,-0.6) -- (3.5,-0.6) node[black,midway,yshift=-0.7cm] { \textcolor{midblue}{$A_2$}};
\end{scope}
\begin{scope}[xshift=12cm]
\Block[1,lightpurple!35!white,1];
\end{scope}
\begin{scope}[xshift=13cm]
\Block[4,lightpink!50!white,1];
\draw [thick,midpink,decorate,decoration={brace,amplitude=10pt},xshift=-0.5pt,yshift=-0.6pt](-1.5,0.6) -- (3.5,0.6) node[black,midway,yshift=+0.7cm] { \textcolor{midpink}{$B_2$}};
\end{scope}
\begin{scope}[yshift=-5cm]
\Block[4,lightpink!50!white,1];
\draw [thick,lightpink,decorate,decoration={brace,amplitude=10pt,mirror},xshift=-0.5pt,yshift=-0.6pt](-0.5,-0.6) -- (3.5,-0.6) node[black,midway,yshift=-0.7cm] { \textcolor{lightpink}{$C_1$}};
\end{scope}
\begin{scope}[xshift=4cm,yshift=-5cm]
\Block[1,midgreen!50!white,1];
\draw [thick,midgreen,decorate,decoration={brace,amplitude=10pt,mirror},xshift=-0.5pt,yshift=-0.6pt](-0.5,-0.6) -- (0.5,-0.6) node[black,midway,yshift=-0.7cm] { \textcolor{midgreen}{$E_1$}};
\end{scope}
\begin{scope}[xshift=5cm,yshift=-5cm]
\Block[3,lightblue!50!white,1];
\draw [thick,midblue,decorate,decoration={brace,amplitude=10pt,mirror},xshift=-0.5pt,yshift=-0.6pt](-0.5,-0.6) -- (2.5,-0.6) node[black,midway,yshift=-0.7cm] { \textcolor{midblue}{$D_1$}};
\end{scope}
\begin{scope}[xshift=8cm,yshift=-5cm]
\Block[1,midyellow!50!white,1];
\draw [thick,darkyellow,decorate,decoration={brace,amplitude=10pt,mirror},xshift=-0.5pt,yshift=-0.6pt](-0.5,-0.6) -- (0.5,-0.6) node[black,midway,yshift=-0.7cm] { \textcolor{darkyellow}{$F_1$}};
\end{scope}
\begin{scope}[xshift=9cm,yshift=-5cm]
\Block[3,lightpink!50!white,1];
\draw [thick,lightpink,decorate,decoration={brace,amplitude=10pt,mirror},xshift=-0.5pt,yshift=-0.6pt](-0.5,-0.6) -- (2.5,-0.6) node[black,midway,yshift=-0.7cm] { \textcolor{lightpink}{$C_2$}};
\end{scope}
\begin{scope}[xshift=12cm,yshift=-5cm]
\Block[1,midgreen!50!white,1];
\draw [thick,midgreen,decorate,decoration={brace,amplitude=10pt,mirror},xshift=-0.5pt,yshift=-0.6pt](-0.5,-0.6) -- (0.5,-0.6) node[black,midway,yshift=-0.7cm] { \textcolor{midgreen}{$E_2$}};
\end{scope}
\begin{scope}[xshift=13cm,yshift=-5cm]
\Block[4,lightblue!50!white,1];
\draw [thick,midblue,decorate,decoration={brace,amplitude=10pt,mirror},xshift=-0.5pt,yshift=-0.6pt](-0.5,-0.6) -- (3.5,-0.6) node[black,midway,yshift=-0.7cm] { \textcolor{midblue}{$D_2$}};
\draw [dashed,darkblue,thick](-9.5,-1.5) -- (-9.5,6.5) ;
\draw [dashed,darkblue,thick](-8.5,-1.5) -- (-8.5,6.5) ;
\draw [dashed,darkblue,thick](-5.5,-1.5) -- (-5.5,6.5) ;
\draw [dashed,darkblue,thick](-4.5,-1.5) -- (-4.5,6.5) ;
\draw [dashed,darkblue,thick](-1.5,-1.5) -- (-1.5,6.5) ;
\draw [dashed,darkblue,thick](-0.5,-1.5) -- (-0.5,6.5) ;
\end{scope}
\node at (18,0.8) {\huge $\Lambda$};
\node at (18,-4.2) {\huge $\Lambda$};
\end{tikzpicture}
\end{center}
  \caption{Notation introduced for the splitting of $\Lambda$ into non-overlapping regions $\{ C_i\}_{i=1}^m, \{ D_i\}_{i=1}^m, \{ E_i\}_{i=1}^m$ and $\{ F_i\}_{i=1}^{m-1}$. Here we are taking $m=2$ for simplicity.}
  \label{fig:4}
\end{figure}
 \noindent Because of the definition of $A$ and $B$, it is clear that
\begin{itemize}
\item $\abs{C_i }=\abs{D_i}=2 r -1 $ for every $1 \leq i \leq m$.
    \item $\abs{E_i}= \abs{F_j} = \ell $, for every $1 \leq i \leq m$, $1 \leq j \leq m-1$.
\end{itemize}
Note that, with this new notation, we aim to prove the following inequality
\begin{equation}\label{eq:mixingconditionCD}
\norm{ \left(\sigma_C^{-1/2}  \otimes \sigma_D^{-1/2} \right) \; \sigma_{CD} \;  \left(\sigma_C^{-1/2}  \otimes \sigma_D^{-1/2} \right) - \identity_{CD}}_\infty  \leq \mathcal{C}  < 1/2 \, ,
\end{equation}
for $\mathcal{C}$ independent of $\Lambda$, whenever d$(C,D) = \mathcal{O} (\log n)$. The proof of this result follows from a repeated use of the following estimate:
\begin{equation}\label{eq:result_andreas_antonio}
\norm{ \sigma_{XZ} \, (\sigma_{X}^{-1} \otimes \sigma_Z^{-1}) - \identity_{XZ}}_\infty \leq \mathcal{K} \operatorname{e}^{- \gamma \ell} \, ,
\end{equation}
 for finite intervals $XYZ$ with $Y$ shielding $X$ from $Z$ and $\abs{Y}=\ell$. Let us denote $\eta(\ell):= \mathcal{K} \operatorname{e}^{- \gamma \ell}$ hereafter for simplicity. As the aforementioned inequality only holds for constructions such as the one presented in Figure \ref{fig:3} (i.e. $\Lambda$ consisting of three connected parts  $X, Y$ and $Z$ such that $Y$ shields $X$ from $Z$ and the given estimate scales with the size of $Y$) while the construction devised in Figure \ref{fig:4} consists of a more complex structure, we need to use \eqref{eq:result_andreas_antonio} recursively to prove the statement of the proposition. We divide the proof into two parts:
\begin{enumerate}
\item \textbf{Splitting step:} In this part, we use \eqref{eq:result_andreas_antonio} repeatedly to approximate
\begin{equation*}
\sigma_{CD} \sim \underset{i=1}{\overset{m}{\bigotimes}} \left( \sigma_{C_i} \otimes \sigma_{D_i} \right) \, .
\end{equation*}
\item \textbf{Joining step: } Now, we use \eqref{eq:result_andreas_antonio} again to entangle $\sigma$ in $C$ and $D$ separately, namely
\begin{equation*}
 \underset{i=1}{\overset{m}{\bigotimes}} \left( \sigma_{C_i} \otimes \sigma_{D_i} \right) \sim \sigma_C \otimes \sigma_D \,  .
\end{equation*}
\end{enumerate}
Note that the approximations stated above are in the sense that the distance of the product of one of the terms by the inverse of the other in operator norm decays exponentially with the distance between $C$ and $D$ (i.e. the size of the overlap between $A$ and $B$, denoted by $\ell$). Therefore, these two steps together allow to conclude the proof of the lemma.

\vspace{0.2cm}
\subsection{ Splitting step}
Given the construction introduced above,  let us further denote, for $j=1, \ldots , 2m-1$, the intervals obtained by considering the convex hull of joining the first $j$ segments from the set $\{ C_i, D_i \}_{i=1}^m$, namely
\begin{equation}\label{eq:def_Xi_Yi_odd}
X_j := \left( \underset{k=1}{\overset{\frac{j-1}{2}}{\bigcup}} \, C_k \cup E_k \cup D_k \cup F_k   \right) \cup \; C_{\frac{j+1}{2}} \text{ and } Z_j := D_{\frac{j+1}{2}}  \; \cup \left( \underset{k=\frac{j+3}{2}}{\overset{n}{\bigcup}} \, C_k \cup E_k \cup D_k \cup F_{k-1}   \right)   ,
\end{equation}
if $j $ is odd, and
\begin{equation}\label{eq:def_Xi_Yi_even}
X_j := \left( \underset{k=1}{\overset{\frac{j}{2}}{\bigcup}} \, C_k \cup E_k \cup D_k \cup F_k   \right) \setminus F_{\frac{j}{2}}  \text{ and } Z_j :=  \left( \underset{k=\frac{j}{2}+1}{\overset{n}{\bigcup}} \, C_k \cup E_k \cup D_k \cup F_{k-1}   \right)  \setminus F_{\frac{j}{2}}  ,
\end{equation}
if $j $ is even. Note that both constructions, as shown in Figure \ref{fig:5} below, are in the form of Figure \ref{fig:3}, where the role of $Y$ is played in each case by the unique $E_i$ or $F_i$ in between the corresponding $X_j$ and $Z_j$.
\begin{figure}[H]
\begin{center}
   \begin{tikzpicture}[scale=0.7]
\Block[4,lightpink!50!white,1];
\draw [thick,lightpink,decorate,decoration={brace,amplitude=10pt,mirror},xshift=-0.5pt,yshift=-0.6pt](-0.5,-0.6) -- (3.5,-0.6) node[black,midway,yshift=-0.7cm] { \textcolor{lightpink}{$C_1$}};
\draw [thick,darkpurple,decorate,decoration={brace,amplitude=10pt},xshift=-0.5pt,yshift=-0.6pt](-0.5,0.6) -- (11.5,0.6) node[black,midway,yshift=0.7cm] { \textcolor{darkpurple}{$X_3$}};
\begin{scope}[xshift=4cm]
\Block[1,midgreen!50!white,1];
\draw [thick,midgreen,decorate,decoration={brace,amplitude=10pt,mirror},xshift=-0.5pt,yshift=-0.6pt](-0.5,-0.6) -- (0.5,-0.6) node[black,midway,yshift=-0.7cm] { \textcolor{midgreen}{$E_1$}};
\end{scope}
\begin{scope}[xshift=5cm]
\Block[3,lightblue!50!white,1];
\draw [thick,midblue,decorate,decoration={brace,amplitude=10pt,mirror},xshift=-0.5pt,yshift=-0.6pt](-0.5,-0.6) -- (2.5,-0.6) node[black,midway,yshift=-0.7cm] { \textcolor{midblue}{$D_1$}};
\end{scope}
\begin{scope}[xshift=8cm]
\Block[1,midyellow!50!white,1];
\draw [thick,darkyellow,decorate,decoration={brace,amplitude=10pt,mirror},xshift=-0.5pt,yshift=-0.6pt](-0.5,-0.6) -- (0.5,-0.6) node[black,midway,yshift=-0.7cm] { \textcolor{darkyellow}{$F_1$}};
\end{scope}
\begin{scope}[xshift=9cm]
\Block[3,lightpink!50!white,1];
\draw [thick,lightpink,decorate,decoration={brace,amplitude=10pt,mirror},xshift=-0.5pt,yshift=-0.6pt](-0.5,-0.6) -- (2.5,-0.6) node[black,midway,yshift=-0.7cm] { \textcolor{lightpink}{$C_2$}};
\end{scope}
\begin{scope}[xshift=12cm]
\Block[1,midgreen!50!white,1];
\draw [thick,midgreen,decorate,decoration={brace,amplitude=10pt,mirror},xshift=-0.5pt,yshift=-0.6pt](-0.5,-0.6) -- (0.5,-0.6) node[black,midway,yshift=-0.7cm] { \textcolor{midgreen}{$E_2$}};
\end{scope}
\begin{scope}[xshift=13cm]
\Block[4,lightblue!50!white,1];
\draw [thick,midblue,decorate,decoration={brace,amplitude=10pt,mirror},xshift=-0.5pt,yshift=-0.6pt](-0.5,-0.6) -- (3.5,-0.6) node[black,midway,yshift=-0.7cm] { \textcolor{midblue}{$D_2$}};
\draw [thick,darkpurple,decorate,decoration={brace,amplitude=10pt},xshift=-0.5pt,yshift=-0.6pt](-0.5,0.6) -- (3.5,0.6) node[black,midway,yshift=0.7cm] { \textcolor{darkpurple}{$Z_3$}};
\end{scope}
\node at (18,0.8) {\huge $\Lambda$};
\end{tikzpicture}
\end{center}
  \caption{Definition of $X_j$ and $Z_j$ from $C_i$, $D_i$, $E_i$ and $F_i$. }
  \label{fig:5}
\end{figure}
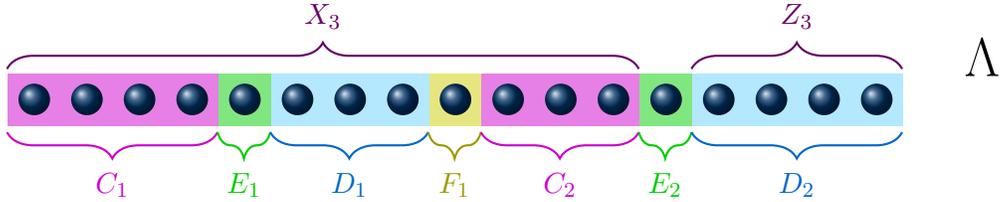
\noindent Moreover, note that in both cases we are in position to apply \eqref{eq:result_andreas_antonio} to $X_j$ and $Z_j$. Then, the following holds for every $j=1, \ldots, 2m-1$:
\begin{equation}\label{eq:recursion_splitting_part}
\norm{\sigma_{X_j Z_j} \left( \sigma_{X_j} \otimes \sigma_{Z_j} \right)^{-1} - \identity_{X_j Z_j}}_\infty \leq  \eta (\ell) \, .
\end{equation}
In particular, the previous expression also holds when restricting $X_j$ and $Z_j$ to $C \cup D$, i.e.
\begin{equation}\label{eq:recursion_splitting_part2}
\norm{\sigma_{CD} \left( \sigma_{(C \cup D) \cap  X_j} \otimes \sigma_{(C \cup D) \cap Z_j} \right)^{-1} - \identity_{CD}}_\infty \leq  \eta (\ell) \, .
\end{equation}
Let us further define for simplicity, for $j=1, \ldots , 2m$, the intervals $ R_j $ from the set $\{ C_i, D_i \}_{i=1}^m$ from left to right, such that
\begin{equation}\label{eq:def_Rj}
R_j := C_{(j+1)/2 } \quad\text{ if $j $ is odd, and } \quad \quad R_j := D_{j/2} \quad \text{if $j $ is even .}
\end{equation}
Furthermore, let us write, for $1 \leq k \leq 2m$, the union of the of all the segments $\{ R_j \}_{j=1}^{2m}$ except for the first $k-1$, namely:
\begin{equation}\label{eq:def_R(k)}
    R^{(k)}:= \underset{i=k}{\overset{2m}{\bigcup}} \,  R_i \, .
\end{equation}
 Then, \eqref{eq:def_R(k)} jointly with \eqref{eq:recursion_splitting_part2} imply  the following inequality for each $1 \leq k \leq 2m-1$:
\begin{equation}\label{eq:recursion_splitting_part3}
\norm{ \left( \bigg( \underset{i=1}{\overset{k-1}{\bigotimes}} \; \sigma_{R_i} \bigg)  \otimes \sigma_{R^{(k)}} \right) \left( \bigg( \underset{i=1}{\overset{k}{\bigotimes}} \; \sigma_{R_i} \bigg) \otimes \sigma_{R^{(k+1)}}  \right)^{-1} - \identity_{CD}}_\infty \leq   \eta(\ell) \, ,
\end{equation}
Let us write for every  $2 \leq k \leq 2m-1$
\begin{equation}\label{eq:definition_xi}
    \xi_k := \left( \bigg( \underset{i=1}{\overset{k-1}{\bigotimes}} \; \sigma_{R_i} \bigg)  \otimes \sigma_{R^{(k)}} \right) \left( \bigg( \underset{i=1}{\overset{k}{\bigotimes}} \; \sigma_{R_i} \bigg) \otimes \sigma_{R^{(k+1)}}  \right)^{-1}  \, ,
\end{equation}
and $\xi_1 := \sigma_{CD} \left(  \sigma_{R_1} \otimes \sigma_{R^{(1)}} \right)^{-1} $.  Hence, it is clear that
\begin{align*}
    \norm{ \sigma_{CD} \left(  \bigotimes_{k=1}^{2m} \, \sigma_{R_k}\right)^{-1} - \identity_{CD}  }_\infty & =    \norm{ \sigma_{CD} \left(  \sigma_{R_1} \otimes \sigma_{R^{(1)}} \right)^{-1} \left(  \sigma_{R_1} \otimes \sigma_{R^{(1)}} \right) \ldots  \left(  \bigotimes_{k=1}^{2m} \, \sigma_{R_k}\right)^{-1} - \identity_{CD}  }_\infty \\
    & = \norm{ \underset{k=1}{\overset{2m-1}{\prod}} \, \xi_k - \identity_{CD} }_\infty \, .
\end{align*}
Now, note that
\begin{align*}
    \norm{ \underset{k=1}{\overset{2m-1}{\prod}} \, \xi_k - \identity_{CD} }_\infty &= \norm{( \xi_1 - \identity_{CD}) \underset{k=2}{\overset{2m-1}{\prod}} \, \xi_k +  \underset{k=2}{\overset{2m-1}{\prod}} \, \xi_k   - \identity_{CD} }_\infty \\
    & \leq \norm{ \xi_1 - \identity_{R_1}}_\infty \norm{\underset{k=2}{\overset{2m-1}{\prod}} \, \xi_k  }_\infty  + \norm{ \underset{k=2}{\overset{2m-1}{\prod}} \, \xi_k - \identity_{R^{(1)}}}_\infty  \\
    & \leq \eta(\ell) \underset{k=2}{\overset{2m-1}{\prod}} \norm{\xi_k}_\infty  +  \norm{ \underset{k=2}{\overset{2m-1}{\prod}} \, \xi_k - \identity_{R^{(1)}}}_\infty   \, ,
\end{align*}
where we have used \eqref{eq:recursion_splitting_part3} in the last inequality, as well as triangle inequality and submultiplicativy of Schatten norms repeatedly. We further estimate each of the terms in the product of the first summand of the right-hand side by:
\begin{equation}\label{eq:bound_norm_xis_eta}
    \norm{\xi_k}_\infty \leq \norm{\xi_k - \identity_{R^{(k)}}}_\infty + 1 \leq \eta(\ell) + 1 \, .
\end{equation}
Therefore,
\begin{equation*}
        \norm{ \underset{k=1}{\overset{2m-1}{\prod}} \, \xi_k - \identity_{CD} }_\infty \leq \eta (\ell) \left(  \eta(\ell) + 1 \right)^{2m-2} +  \norm{ \underset{k=2}{\overset{2m-1}{\prod}} \, \xi_k - \identity_{R^{(1)}}}_\infty  \, .
\end{equation*}
Repeating the same procedure $2m - 2$ times on the last term in the right-hand side, we obtain
\begin{align}\label{eq:bound_product_xi}
     \norm{ \underset{k=1}{\overset{2m-1}{\prod}} \, \xi_k - \identity_{CD} }_\infty & \leq \eta(\ell) \underset{k=0}{\overset{2m-2}{\sum}} \left(  \eta(\ell) + 1 \right)^{k}  \nonumber \\
     & = \left(  \eta(\ell) + 1 \right)^{2m-1}  - 1 \, .
\end{align}

\subsection{ Joining step}

In the first part of the proof, we have provided an estimate for the distance of $\sigma_{CD}$ from being a tensor product between all the segments of the form $\{ C_i , D_i \}_{i=1}^{m}$. Now, we need to approximate:
\begin{equation*}
\underset{i=1}{\overset{m}{\bigotimes}} \, \sigma_{C_i} \sim \sigma_C \; \; , \; \; \underset{i=1}{\overset{m}{\bigotimes}} \,\sigma_{D_i} \sim \sigma_D \, .
\end{equation*}
The idea followed here is similar to that of the previous step, but in a reversed direction. Let us denote, for every $1 \leq k \leq m$,
\begin{equation}\label{eq:def_C(k)D(k)}
    C^{(k)}:= \underset{i=1}{\overset{k}{\bigcup}}  \, C_i \; \;  , \quad \;   D^{(k)}:= \underset{i=1}{\overset{k}{\bigcup}} \, D_i \, .
\end{equation}
Similarly to \eqref{eq:recursion_splitting_part3}, it is clear that the following holds for every $1 \leq k \leq m-1$,
\begin{equation}\label{eq:recursion_joining_part_C}
\norm{ \left( \sigma_{C^{(k-1)}}  \otimes  \bigg( \underset{i=k}{\overset{m}{\bigotimes}} \; \sigma_{C_i} \bigg)  \right) \left( \sigma_{C^{(k)}} \otimes  \bigg( \underset{i=k+1}{\overset{m}{\bigotimes}} \; \sigma_{C_i} \bigg)   \right)^{-1} - \identity_{C}}_\infty \leq   \eta(\ell) \, ,
\end{equation}
and analogously for $D$. Note that, in this case, we actually have $\eta(2 (r + \ell) - 1)$, since that is the distance between every two segments $C_i$ and $C_{i+1}$. However, since $\eta$ is monotonically decreasing, we just consider $\eta(\ell)$ as an upper bound for the previous norm for simplicity. Let us denote
\begin{align}\label{eq:definition_xi_CD}
    \xi^C_k & := \left( \sigma_{C^{(k-1)}}  \otimes  \bigg( \underset{i=k}{\overset{m}{\bigotimes}} \; \sigma_{C_i} \bigg)  \right) \left( \sigma_{C^{(k)}} \otimes  \bigg( \underset{i=k+1}{\overset{m}{\bigotimes}} \; \sigma_{C_i} \bigg)   \right)^{-1}  \, , \\
        \xi^D_k & := \left( \sigma_{D^{(k-1)}}  \otimes  \bigg( \underset{i=k}{\overset{m}{\bigotimes}} \; \sigma_{D_i} \bigg)  \right) \left( \sigma_{D^{(k)}} \otimes  \bigg( \underset{i=k+1}{\overset{m}{\bigotimes}} \; \sigma_{D_i} \bigg)   \right)^{-1}   \, .
\end{align}
Then, we have
\begin{align*}
    \norm{\left( \underset{i=1}{\overset{m}{\bigotimes}} \left( \sigma_{C_i} \otimes \sigma_{D_i} \right) \right) \left( \sigma_C \otimes \sigma_D \right)^{-1} - \identity_{CD}}_\infty & =     \norm{\left( \bigg( \underset{i=1}{\overset{m}{\bigotimes}} \, \sigma_{C_i}  \bigg) \sigma_C^{-1} \right) \otimes \left(  \bigg( \underset{i=1}{\overset{m}{\bigotimes}} \, \sigma_{D_i} \bigg) \sigma_D^{-1}  \right)  - \identity_{CD}}_\infty \\
    &= \norm{ \left( \underset{k=1}{\overset{m-1}{\prod}} \, \xi^C_k \right) \otimes \left( \underset{k=1}{\overset{m-1}{\prod}} \, \xi^D_k \right) - \identity_{CD}}_\infty
\end{align*}
Next, we need to separate the norms of the difference between each of the terms $\xi_k^{C,D}$ and identity from the others, as we did in the previous step. The main difference now lies in the fact that we need to do it for both $C$ and $D$. We first fix the terms in $D$ and work on the terms with support on $C$:
\begin{align*}
    & \norm{ \left( \underset{k=1}{\overset{m-1}{\prod}} \, \xi^C_k \right) \otimes \left( \underset{k=1}{\overset{m-1}{\prod}} \, \xi^D_k \right) - \identity_{CD}}_\infty \\
    & \phantom{asdasdasdasd}= \norm{ (\xi_1^C - \identity_{CD})\left( \underset{k=2}{\overset{m-1}{\prod}} \, \xi^C_k \right) \otimes \left( \underset{k=1}{\overset{m-1}{\prod}} \, \xi^D_k \right) +  \left( \underset{k=2}{\overset{m-1}{\prod}} \, \xi^C_k \right) \otimes \left( \underset{k=1}{\overset{m-1}{\prod}} \, \xi^D_k \right) - \identity_{CD}}_\infty \\
    & \phantom{asdasdasdasd}\leq \norm{\xi_1^C - \identity_{C^{(1)}} }_\infty   \norm{\underset{k=2}{\overset{m-1}{\prod}} \, \xi^C_k}_\infty \norm{\underset{k=1}{\overset{m-1}{\prod}} \, \xi^D_k }_\infty  + \norm{   \left( \underset{k=2}{\overset{m-1}{\prod}} \, \xi^C_k \right) \otimes \left( \underset{k=1}{\overset{m-1}{\prod}} \, \xi^D_k \right) - \identity_{\Lambda \setminus C^{(1)}}}_\infty \\
    & \phantom{asdasdasdasd}\leq \eta(\ell) \underset{k=2}{\overset{m-1}{\prod}} \norm{\xi^C_k}_\infty \underset{k=1}{\overset{m-1}{\prod}} \norm{\xi^D_k}_\infty + \norm{   \left( \underset{k=2}{\overset{m-1}{\prod}} \, \xi^C_k \right) \otimes \left( \underset{k=1}{\overset{m-1}{\prod}} \, \xi^D_k \right) - \identity_{\Lambda \setminus C^{(1)}}}_\infty \, ,
\end{align*}
where we have used \eqref{eq:recursion_joining_part_C} as well as triangle inequality and submultiplicativity for the operator norm. Therefore,  using now an analogue of \eqref{eq:bound_norm_xis_eta} for $\xi_k^{C,D}$, we obtain
\begin{equation*}
    \norm{ \left( \underset{k=1}{\overset{m-1}{\prod}} \, \xi^C_k \right) \otimes \left( \underset{k=1}{\overset{m-1}{\prod}} \, \xi^D_k \right) - \identity_{CD}}_\infty \leq   \eta(\ell) \left( \eta(\ell)  + 1 \right)^{2m -3 } + \norm{   \left( \underset{k=2}{\overset{m-1}{\prod}} \, \xi^C_k \right) \otimes \left( \underset{k=1}{\overset{m-1}{\prod}} \, \xi^D_k \right) - \identity_{\Lambda \setminus C^{(1)}}}_\infty \, .
\end{equation*}
By repeating the same procedure $m-2$ times on the terms with support on $C$, we get
\begin{equation*}
    \norm{\left( \underset{i=1}{\overset{m}{\bigotimes}} \left( \sigma_{C_i} \otimes \sigma_{D_i} \right) \right) \left( \sigma_C \otimes \sigma_D \right)^{-1} - \identity_{CD}}_\infty \leq  \eta(\ell) \underset{k=m-1}{\overset{2m-3}{\sum}} \left(  \eta(\ell) + 1 \right)^{k}   + \norm{   \left( \underset{k=1}{\overset{m-1}{\prod}} \, \xi^D_k \right) - \identity_{D}}_\infty \, .
\end{equation*}
Now, following the same idea for the terms on $D$, we can clearly conclude:
\begin{align}\label{eq:final_joining_step}
    \norm{\left( \underset{i=1}{\overset{m}{\bigotimes}} \left( \sigma_{C_i} \otimes \sigma_{D_i} \right) \right) \left( \sigma_C \otimes \sigma_D \right)^{-1} - \identity_{CD}}_\infty & \leq  \eta(\ell) \underset{k=m-1}{\overset{2m-3}{\sum}} \left(  \eta(\ell) + 1 \right)^{k}   + \eta(\ell) \underset{k=0}{\overset{m-2}{\sum}}  \left(  \eta(\ell) + 1 \right)^{k} \, \nonumber \\
    &=  \eta(\ell) \underset{k=0}{\overset{2m-3}{\sum}}  \left(  \eta(\ell) + 1 \right)^{k} \nonumber \\
    & = \left(  \eta(\ell) + 1 \right)^{2m-2} -1 \, .
\end{align}

\subsection{Merging both steps}

To conclude the proof of Lemma \ref{lemma:globaltolocal}, we need to combine \eqref{eq:bound_product_xi} and \eqref{eq:final_joining_step}. We also need the following estimate on the norm of the difference of a product of observables $X_1, X_2$ and the identity:
\begin{equation}\label{eq:estimate_infty_norm_product}
\norm{X_1 X_2 - \identity}_\infty \leq \norm{X_1 - \identity}_\infty \norm{X_2 - \identity}_\infty + \norm{X_1 - \identity}_\infty +\norm{X_2 - \identity}_\infty \, .
\end{equation}
With this at hand, we can prove
\begin{align*}
    & \norm{\sigma_{CD} (\sigma_C^{-1} \otimes \sigma_D^{-1})- \identity_{CD}}_\infty \\
    & \phantom{asda}=  \Bigg\|\underbrace{\sigma_{CD} \left( \underset{i=1}{\overset{m }{\bigotimes}} (\sigma_{C_i} \otimes \sigma_{D_i})\right)}_{=: \Xi_1} \underbrace{\left( \underset{i=1}{\overset{m }{\bigotimes}} (\sigma_{C_i} \otimes \sigma_{D_i})\right)^{-1} (\sigma_C^{-1} \otimes \sigma_D^{-1})}_{=: \Xi_2}- \identity_{CD} \Bigg\|_\infty \\
    & \phantom{asda}\leq \norm{\Xi_1 - \identity_{CD}}_\infty \norm{\Xi_2 - \identity_{CD}}_\infty +\norm{\Xi_1 - \identity_{CD}}_\infty +\norm{\Xi_2 - \identity_{CD}}_\infty \\
    & \phantom{asda}\leq \left[ \left(  \eta(\ell) + 1 \right)^{2m-1} -1 \right] \left[ \left(  \eta(\ell) + 1 \right)^{2m-2} -1 \right] + \left[ \left(  \eta(\ell) + 1 \right)^{2m-1} -1 \right] + \left[ \left(  \eta(\ell) + 1 \right)^{2m-2} -1 \right] 
    \\  &\phantom{asda} \leq \left(  \eta(\ell) + 1 \right)^{4m-3} -1 \
    .
\end{align*}
Finally, let us recall that $\eta(\ell)= \mathcal{K} \operatorname{e}^{- \gamma \ell}$ for some constants $\mathcal{K}, \gamma >0 $ determined by \eqref{eq:mixing}. Thus, choosing $\ell= \mathcal{O}(\ln(n))$ for $n= \abs{\Lambda}$, we clearly have $\abs{A_i} = \abs{B_i} = \mathcal{O} (\ln(n))$ for every $i=1, \ldots , m$, due to the explicit form of the covering considered. Moreover, the number of small sub-regions considered is $m=  \mathcal{O}(n / \ln(n) )$. Hence, we can control the above upper bound. Indeed, note that, in such a case, the limit of the bound above with $m$ tending to infinity is a constant, that we can make smaller than $1/2$ by choosing properly the specific value of $\ell$.

Therefore,  we have proven that \eqref{ATmixingcondrel} holds with $\|h(\sigma_{A^cB^c})\|_\infty\le \mathcal{C}$ independent of $n$ as long as the regions $A_i$ and $B_j$ grow logarithmically with $n$, since those regions depend linearly on $\ell$.

\section{Proof of local control of the constant (\Cref{lemmatechnical1})}\label{sec:localcontrol}

This section is devoted to proving the other main technical tool in the proof of Theorem \ref{thm:main_result_formal}, namely \Cref{lemmatechnical1}.  We first state a slightly more general formulation of Lemma \ref{lemmatechnical1} before proving the latter. Recall that $\chi_\Phi$ denotes the module Choi operator for a bimodule map $\Phi$.
\begin{lem}\label{lemmatechnicalgeneral}
Let $\cN\subset \cM$ be finite-dimensional von Neumann subalgebras and let $E_\cN:\cM\to\cN$ be a conditional expectation. Suppose $\Phi:\cM\to \cM$ is a  $\cN$-bimodule map.
\begin{enumerate}
\item[$\operatorname{i)}$] If $\norm{\chi_\Phi-\chi_{E_\cN}}_{\cB((l_2^n))\ten \cM}{}\le \eps\le 1$, then
\[ (1-\eps)E_{\cN}\le_{\operatorname{cp}} \Phi\le_{\operatorname{cp}} (1+\eps) E_{\cN} \ .\]
where the order $\Phi\le_{\operatorname{cp}}\Psi$ means $\Psi-\Phi$ is completely positive.
\item[$\operatorname{ii)}$] Assume that $\Phi$ is unital and self-adjoint with respect to the $\operatorname{KMS}$ inner product $\langle.,.\rangle_{\sigma}$ for an (arbitrary) invertible invariant state $\sigma=E_{\cN*}(\sigma)$ and $\lambda:=\norm{\Phi-E_\cN:L_2(\sigma)\to L_2(\sigma)}{}<1$. Then for $k>\frac{\ln C_{\tau,\operatorname{cb}}(\cM:\cN)}{-\ln \lambda}$,
\[ (1-\epsilon)E_\cN\le_{\operatorname{cp}} \Phi^k\le_{\operatorname{cp}} (1+\epsilon)E_\cN  \]
for $\epsilon=\lambda^kC_{\tau,\operatorname{cb}}(\cM:\cN)<1 $. In particular, one can choose $k$ such that $\eps=\lambda^k\mu_{\min}(\sigma)^{-1}C_{\operatorname{cb}}(\cM:\cN)<1$ for any invertible invariant state $\sigma=E_{\cN*}(\sigma)$.
\end{enumerate}
\end{lem}

The proof of \Cref{lemmatechnicalgeneral} is postponed to \Cref{prooflemmatechngen}. For now, we show how it implies \Cref{lemmatechnical1}:

\begin{proof}[Proof of Lemma \ref{lemmatechnical1}]
Specializing \Cref{lemmatechnicalgeneral} to the spin chain setting, we choose the map $\Phi$ to be $\prod_{i\in A} E_i$ and $E_\cN:=E_A$ for a given region $A\subseteq\Lambda$. Then it suffices to control the constants $\mu_{\min}(\sigma^{(A\partial) })$ and $C_{\operatorname{cb}}(\cB(\cH_{A\partial}):\cN_A)$, where $\sigma^{(A\partial)}$ is taken as the Gibbs state on region $A\partial$, which is a fixed point of $E_{A*}$. First of all
\begin{align*}
    \sigma^{(A\partial)}=e^{-\beta H_{A\partial}}/\tr e^{-\beta H_{A\partial}}\ge e^{-\beta \|H_{A\partial}\|_\infty}/\tr e^{\beta \|H_{A\partial}\|_\infty}\ge d^{-|A\partial|}e^{-2\beta |A\partial| J}\,,
\end{align*}
where $J$ denotes the interaction strength of $H$. Moreover, $C_{\operatorname{cb}}(\cB(\cH_{A\partial}):\cN_A)\le C_{\operatorname{cb}}(\cB(\cH_{A\partial}):\mathbb{C}\Id_{A\partial})\le d^{|A\partial|}$. Then, the conclusion of Lemma \ref{lemmatechnicalgeneral} holds true for $\eps = \lambda^k\,d^{2|A\partial|}e^{2\beta |A\partial| J}<1$. Equation \eqref{eq:kA} follows after taking $\eps=\frac{1}{2}$ and solving for $k$.

That $\lambda$ is independent of $n$ is a simple consequence of the uniform positivity of the gap:
\begin{align*}
\lambda(\cL_\Lambda^D):=\inf_\Lambda\inf_{X \subset \Lambda}\,\frac{-\langle X,\,\cL^D_\Lambda(X)\rangle_{\sigma}}{\|X-\tr(\sigma X)\Id\|_{L_2(\sigma)}^2}>0
\end{align*}
for 1D Davies generators associated to commuting Hamiltonians \cite[Proposition 29]{kastoryano2016quantum}, together with the detectability lemma \cite{aharonov2009detectability,anshu2016simple} which asserts that:
\begin{align*}
\lambda^2\le\frac{1}{\lambda(\cL_\Lambda^D)/g^2+1}<1\,,
\end{align*}
where $g$ denotes the maximum number of conditional expectations $E_i$ which do not commute with any given $E_j$. In the present framework of a finite-range commuting Hamiltonian, the conditional expectations $E_j$ are supported on finite regions of locality controlled by the range $r$ of $H$. Therefore, $g$ is finite and thus $\lambda<1$ independently of the system size.

% \textcolor{blue}{If we want to go noncommuting, instead of invoking the gap+detectability lemma, we might have a shot at directly controlling the L2 to L2 norm of the detectability lemma operator. Then we can control this by the L2 to L2 map of the recovery maps (although we still have to use $E_A$ for the englobing region) }

\end{proof}

In order to prove \Cref{lemmatechnicalgeneral}, we first need to derive some technical results on amalgamated $L_p$ spaces which might be of independent interest. These are gathered in \Cref{sec:amalgamatedLp}.

\subsection{Amalgamated $L_p$ spaces}\label{sec:amalgamatedLp}

In this section, we briefly review the operator space structure of the weighted amalgamated $L_p$ spaces introduced in \cite{junge2010mixed} (see also \cite{bardet2018hypercontractivity} for a recent account of theses spaces in finite dimensions) and connect them to the completely bounded subalgebra indices introduced in \Cref{sec:indices}. These spaces are generalizations of the non-commutative vector-valued $L_p$ spaces introduced by Pisier in \cite{Pisier98NCvector}. Given a full-rank state $\sigma\in \cD(\cH)$ and $p\ge 1$, we define the weighted $L_p(\sigma)$ space by the norm
\begin{align}\label{kosaki}
    \|x\|_{L_p(\sigma)}:=\Big(\tr\big|\si^{\frac{1}{2p}}x\si^{\frac{1}{2p}} \big|^p\Big)^{\frac{1}{p}}\,.
\end{align}
As expected, for $p=2$, $L_2(\sigma)$ is a Hilbert space associated with the so-called $\sigma$-KMS inner product:
\begin{align}\label{eq:sigmakms}
    \langle x,\,y\rangle_{\sigma}:=\tr\big(x^\dagger\sigma^{\frac{1}{2}} y\,\sigma^{\frac{1}{2}} \big)\,.
\end{align}
Let $\cM$ be a finite dimensional von Neumann algebra equipped with trace $\Tr$. Let $\cN\subset \cB(\cH)$ be a subalgebra and let $E_{\cN}:\cB(\cH)\to \cN$ be a conditional expectation. We define the invariant state
\begin{align*}
    \sigma_{\Tr}:=E_{\cN*}\Big(\frac{\Id_\cM}{\Tr(\Id_\cM)}\Big)\,.
\end{align*}
Note that $\sigma_{\Tr}$ is explicit from \eqref{cd}, and $\sigma_{\Tr}$ and $E_\cN$ determine each other. Moreover $\sigma_{\Tr}\in \cN'$ because for any $x\in \cN,y\in\cM$
\begin{align*}\Tr(\sigma_{\Tr}xy)
&=\Tr\Big( E_{\cN*}\Big(\frac{\Id}{\Tr(\Id_\cM)}\Big)xy\Big)
=\frac{1}{\Tr(\Id_\cM)}\Tr\Big(E_{\cN}(xy)\Big)
=\frac{1}{\Tr(\Id_\cM)}\Tr\Big(xE_{\cN}(y)\Big)\\
&=
\frac{1}{\Tr(\Id_\cM)}\Tr\Big(E_{\cN}(y)x\Big)=\frac{1}{\Tr(\Id_\cM)}\Tr\Big(E_{\cN}(yx)\Big)
=\Tr\Big(E_{\cN*}\Big(\frac{\Id}{\Tr(\Id_\cM)}\Big)yx\Big)\\
&=\Tr\Big(xE_{\cN*}\Big(\frac{\Id}{\Tr(\Id_\cM)}\Big)y\Big)=\Tr\Big(x\sigma_{\Tr}y\Big)\,.
\end{align*}
In other words, $\si_\tr$ restricted to $\cN$ is a trace.
Let $1\leq q,p\le\infty$ and fix $\frac{1}{r}=|\frac{1}{q}-\frac{1}{p}|$. The amalgamated $L_p$ spaces are defined via the following norms: for $p\ge q$,
\begin{align}
  	\norm{x}_{L_q^p(\Ncal\subset \cM)}:=\inf_{a,b\in\DF,\, y\in\cM, x=ayb}\,
   	\norm{a}_{L_{2r}(\sigma_{\tr})}\label{eq_def_DFnorms}\,
   	 \norm{b}_{L_{2r}(\sigma_{\tr})}\,\norm{y}_{L_p(\sigma_{\tr})}\, ,
  \end{align}
where the infimum is over all factorizations $x=ayb$ for $a,b\in \cN$ and $y\in \cM$. For $p\le q$,
  \begin{align}
 	\norm{y}_{L_q^p(\cN\subset\cM)}:=\sup_{a,b\in\DF}\,\frac{\norm{ayb}_{L_p(\sigma_{\Tr})}}{\norm{a}_{L_{2r}(\sigma_{\tr})}\norm{b}_{L_{2r}(\sigma_{\tr})}}\,\label{linfty1pq}
\end{align}
where the supremum is over all $a,b\in \cN$.
For any $1\le q,p\le \infty$, we denote by ${L}_q^p(\cN\subset \cM)$ the space $\cM$ equipped with the above norms. We gather some basic properties of amalgamated $L_p$-norms before discussing their operator space structure. In the following, we fix $p'$ and $q'$ to be the H\"{o}lder conjugate of $p$ and $q$ respectively such as $\frac{1}{p}+\frac{1}{p'}=1$.
\begin{prop}\label{prop_duality}
Let $\cN\subset\cM$ be finite dimensional von Neumann algebras. Let $1\le q,p\le \infty$. Then
 \begin{itemize}
 	\item[$\operatorname{(i)}$] H\"{o}lder's inequality: for any $x\in{L}_p^q(\cN\subset \cM)$ and $y\in{L}_{p'}^{q'}(\cN\subset \cM)$,
 	\begin{align}\label{eq:holberamal}
 		|\langle x,y\rangle_{\sigma_\tr}|\le \|x\|_{L_q^p(\cN\subset\cM)}\|y\|_{L_{q'}^{p'}(\cN\subset\cM)}\,.
 		\end{align}
 	\item[$\operatorname{(ii)}$] Duality: For any $x\in L_q^p(\cN\subset\cM)$,
	\begin{align}\label{eq:duality}
 		\|x\|_{L_q^p(\cN\subset\cM)}=\sup\Big\{ |\langle x,y\rangle_{\sigma_\tr}|:~\|y\|_{L_{q'}^{p'}(\cN\subset\cM)}= 1\Big\}\,. \end{align}
 \item[$\operatorname{(iii)}$] The infimum (resp. $supremum$) in the definition  \eqref{eq_def_DFnorms} (resp. \eqref{linfty1pq}) can be restricted to the set of positive semidefinite operators $a,b\ge 0$. Furthermore, for all positive semidefinite $x$,
 	\begin{align}\label{eq_prop_normp_positive}
 		&\norm{x}_{L_{q}^p(\cN\subset \cM)}=\underset{\substack{a\in\cN,~a>0\\~\|a\|_{L_1(\sigma_\tr)}=1}}{\inf}\,\norm{a^{-\frac{1}{2r}}\,x\,a^{-\frac{1}{2r}}}_{L_{p}(\sigma_\tr)}\,\,  &\text{if $p\geq q$} \\
 \label{eq_prop_normq_positive}
 		&\norm{x}_{L_p^q(\cN\subset \cM)}=\underset{\substack{
 		a\in\cN,~a>0 \\\|a\|_{L_1(\sigma_\tr)}=1}}{\sup}\,\norm{a^{\frac{1}{2r}}\,x\,a^{\frac{1}{2r}}}_{L_q(\sigma_{\Tr})}\, \, &\text{if $p\leq q$}
 	\end{align}
\item[$\operatorname{(iv)}$] The following complex interpolation relation holds (see.
e.g. \cite{bergh2012interpolation} for an introduction to interpolation spaces):
 		\begin{align}\label{eq:interp}
 		    L_q^p(\cN\subset\cM)\cong \big[L_{q_0}^{p_0}(\cN\subset\cM),\,L_{q_1}^{p_1}(\cN\subset\cM)\big]_{\theta}\,,
 		\end{align}
 		where $0\le \theta\le 1$, $(1-\theta)/p_0+\theta/p_1=1/p$, $(1-\theta)/q_0+\theta/q_1=1/q$ and $(p_1-q_1)(p_0-q_0)\ge 0$.
 			%	\item[$\operatorname{(iv)}$] For any bimodule map $\Phi:\cB(\cH)\to \cB(\cH)$ and $1\le p,q,s\le \infty$
 	%	\begin{align}\label{eq:bimodule}
 %	\|\Phi:\,L_\infty^p(\cN\subset \cB(\cH))\to L_\infty^p(\cN\subset \cB(\cH))\|=\|\Phi:L_s^p(\cN\subset\cB(\cH))\to L_s^q(\cN\subset \cB(\cH))\|\,.
 	%	\end{align}
 \item[$\operatorname{(v)}$] Relation with $L_p(\sigma_\tr)$ norms: if $q\le p$, then for any $x\in L_p(\sigma_\tr)$,
	\begin{align}\label{eq24}
 		& \|x\|_{L_q(\sigma_\tr)}\le\|x\|_{L_{q}^p(\cN\subset \cM)}\le \|x\|_{L_p(\sigma_\tr)},\\
 		& \|x\|_{L_q(\sigma_\tr)}\le \|x\|_{L_p^q(\cN\subset \cM)}\le \|x\|_{L_p(\sigma_\tr)}\,.\label{eq25}
 	\end{align}	
In particular, the equality holds for $p=q$, which is referred as Fubini's Theorem.

% \item[(iv)] The hierarchy of norms: for $1\le q_1\le q_2, p_1\le p_2\le +\infty$, and any $X\in\cB(\cH)$,
% \begin{align*}
% 	\|X\|_{(q_1,p_1),\,\cN}\le \|X\|_{(q_2,p_2),\,\cN}\,.
% \end{align*}	
	
%	\item[$\operatorname{(vi)}$] For all $1\le q\le p\le +\infty$, $\|x\|_{L_q^p(\cN\subset\cB(\cH))}=\|x\|_{L_q(\sigma_\tr)}$ whenever $x\in\cN$.
 \end{itemize}	
\end{prop}
\begin{proof}The property (i-iii) and (v) are proved in \cite[Proposition 3.1]{bardet2018hypercontractivity}. The complex interpolation (iv) is proved in \cite[Theorem 3.2 \& 4.6]{junge2010mixed} (for the reduction from the Haagerup $L_p$-norm to the Kosaki $L_p$-norm we used in \eqref{kosaki}, see \cite[Proposition 2.4]{GYZ19interpolation}). For completeness, we include the proof of \eqref{eq_prop_normp_positive} $p=\infty$, which will be sufficient for our discussion.
%We first consider the special case $q=\infty$ and then argue general case $q\le p$ by interpolation.
%In the following, we will use the short notation $\norm{\cdot}_{p}\neq \norm{\cdot}_{L_{p}(\sigma_\tr)}$
We show that when $x=x^\dag$ is self-adjoint, it suffices to consider $a=b>0$ in the infimum of the definition of the $L_q^\infty$-norm \eqref{eq_def_DFnorms}. Without losing generality, we assume $\norm{x}_{L_q^\infty(\cN\subset \cM)}<1$ with $x=ayb$ such that
\[\ \norm{a}_{L_{2q}(\sigma_\tr)}=\norm{b}_{L_{2q}(\sigma_\tr)}<1\ , \ \norm{y}_{\infty}<1\,.\]
By polar decomposition, we can further assume $a,b\ge 0$.
Denote $Q=(\frac{1}{2}a^{2}+\frac{1}{2}b^{2}+\delta \Id)^{1/2}\in \cN$ for some small $\delta>0$.
Since $x=\frac{1}{2}(x+x^\dag)=\frac{1}{2}(ayb+b y^\dag a)$, we have $x=QYQ$ for
\begin{align*} Y=\frac{1}{2}Q^{-1}(ayb+b y^\dag a)Q^{-1}
=
\frac{1}{2}Q^{-1}\left[\begin{array}{cc} a& b\end{array}\right]
\cdot\left[\begin{array}{cc} y&0\\ 0 &y^\dag \end{array}\right]\cdot
\left[\begin{array}{c} a\\ b\end{array}\right]
Q^{-1}\,.
\end{align*}
Note that $$\norm{Q}_{L_{2q}(\sigma_\tr)}=\norm{\frac{1}{2}a^{2}+\frac{1}{2}b^{2}+\delta \Id}_{L_{q}(\sigma_\tr)}^{1/2}
\le (\frac{1}{2}\norm{a^2}_{L_{q}(\sigma_\tr)}+\frac{1}{2}\norm{b^2}_{L_{q}(\sigma_\tr)}+\delta)^{1/2}=(1+\delta)^{1/2}$$
and
\begin{align*}\norm{\left[\begin{array}{cc} Q^{-1}a& Q^{-1}b
\end{array}\right]}_{\infty}=&\norm{Q^{-1}(a^2+b^2)Q^{-1}}_{\infty}\\ =&\norm{(a^2+b^2)^{1/2}(\frac{1}{2} a^{2}+\frac{1}{2}b^{2}+\delta \Id)^{-1} (a^2+b^2)^{1/2}}_{\infty}
\\ \le&\norm{(a^2+b^2)^{1/2}(\frac{1}{2} a^{2}+\frac{1}{2}b^{2})^{-1} (a^2+b^2)^{1/2}}_{\infty}
\le \sqrt{2}\,.
\end{align*}
Thus
\[ \norm{Y}_{\infty}\le \frac{1}{2}
\norm{\left[\begin{array}{cc} Q^{-1}a& Q^{-1}b
\end{array}\right]}_{\infty}\norm{
\left[\begin{array}{cc} y&0\\ 0 &y^\dag \end{array}\right]}_{\infty} \norm{
\left[\begin{array}{c} aQ^{-1}\\ bQ^{-1}\end{array}\right]}_{\infty}
\le \norm{y}_{\infty}=1\,.
\]
Take $A=\frac{1}{\norm{Q}_{L_{2q}(\si_\tr)}^2}Q^{2q}$. We have $A>0,\norm{A}_{L_1(\si_\tr)}=1$ and
\[\norm{A^{-1/2} xA^{-1/2}}_{\infty}\le \norm{Q}_{L_{2q}(\si_\tr)}^2\norm{Y}_{\infty}\le 1+\delta\ \]
Since $\delta>0$ is arbitrary, this proves \eqref{eq_prop_normp_positive} for $p=\infty$. %For general $1\le q\le p \infty$, we use the interpolation relation \ref{eq:interp}. Suppose $x\ge 0$ and $\norm{x}{L_q^p(\cN\subset\cM)}=1$.

\end{proof}

We will also need the following two factorization Lemmas. The first one is dual form of \cite[Theorem 3.19]{junge2010mixed} at $p=1$.
\begin{lem}\label{factorization}
Let $\cN\subset\cM$ be finite dimensional subalgebras and
$\si_\Tr$ be defined as above. Then for any $x\in \cM$,  \begin{align}\label{eq:factorization2} \norm{x}_{L_1^\infty(\cN\subset\cM)}=\inf_{x=yz}\norm{yy^\dagger}_{L_1^\infty(\cN\subset\cM)}^{1/2}\norm{z^\dagger z}_{L_1^\infty(\cN\subset\cM)}^{1/2}\end{align}
where the infimum is over all factorization $x=yz$ with $y,z\in \cM$.
\end{lem}
\begin{proof}In the following proof, we use the short norm notation $\norm{\cdot}_{}:=\norm{\cdot}_{\infty}$,
 $\norm{\cdot}_{p}:=\norm{\cdot}_{L_p(\si_\tr)}$ and $\norm{\cdot }_{L_1^\infty}:=\norm{\cdot }_{L_1^\infty(\cN\subset\cM)}$. First, for any factorization $x=yz$,
\begin{align*} \norm{x}_{L_1^\infty}=&\inf_{A,B>0, \norm{A}_{1}=\norm{B}_{1}=1}\norm{A^{-1/2}xB^{-1/2}}
\\ \le& \inf_{A,B>0, \norm{A}_{1}=\norm{B}_{1}=1}\norm{A^{-1/2}y}\norm{zB^{-1/2}}
\\ =& \inf_{A>0, \norm{A}_{1}=1}\norm{A^{-1/2}yy^\dagger A^{-1/2}}\inf_{B>0, \norm{B}_{1}=1}\norm{B^{-1/2}z^\dagger zB^{-1/2}}
\\ =&\norm{yy^\dagger}_{L_1^\infty}^{1/2}\norm{z^\dagger z}_{L_1^\infty}^{1/2}\ ,
\end{align*}
where in the last step we used the property \eqref{eq_prop_normp_positive} for $yy^\dagger,z^\dagger z\ge 0$.

For the other direction,
 let us denote
 \[\norm{x}_{h}:=\inf_{x=yz}\norm{yy^\dagger}_{L_1^\infty}^{1/2}\norm{z^\dagger z}_{L_1^\infty}^{1/2}\]
 where the infimum takes over all factorization $x=yz$. We first show $\norm{\cdot}_{h}$ is a norm. To verify the triangle inequality, it suffices to show that for
 any $x_1=y_1z_1$, $x_2=y_2z_2$,
 and $\delta>0$, there exist $y_0,z_0$ such that $x_1+x_2=y_0z_0$ and

 \begin{align*} \norm{y_0y^\dagger_0}_{L_1^\infty}^{1/2}\norm{z^\dagger_0z_0}_{L_1^\infty}^{1/2}\le \norm{y_1y_1^\dagger}_{L_1^\infty}^{1/2}\norm{z^\dagger_1z_1}_{L_1^\infty}^{1/2}+\norm{y_2y_2^\dagger}_{L_1^\infty}^{1/2}\norm{z^\dagger_2z_2}_{L_1^\infty}^{1/2}+\delta.
 \end{align*}

By rescaling $x_1=ty_1\cdot t^{-1}z_1$, we can assume
\begin{align} &\norm{y_1y_1^\dagger}_{L_1^\infty}^{1/2}\norm{z^\dagger_1z_1}_{L_1^\infty}^{1/2}+
\norm{y_2y_2^\dagger}_{L_1^\infty}^{1/2}\norm{z^\dagger_2z_2}_{L_1^\infty}^{1/2}
\nonumber\\ &\qquad \qquad\qquad =\Big(\norm{y_1y_1^\dagger}_{L_1^\infty}+\norm{y_2y_2^\dagger}_{L_1^\infty}\Big)^{1/2}
\Big(\norm{z^\dagger_1z_1}_{L_1^\infty}+\norm{z^\dagger_2z_2}_{L_1^\infty}\Big)^{1/2} \label{eq:scale}
\end{align}
Take $\delta>0$ and
\begin{align*} &y=(y_1y_1^\dagger +y_2y_2^\dagger +\delta\Id)^{1/2}\ ,   z=(z_1^\dagger z_1+z_2^\dagger z_2+\delta\Id)^{1/2}\ ,\\ \  & X=y^{-1} (y_1z_1+y_2z_2)z^{-1}=y^{-1}(x_1+x_2)z^{-1}\ .
\end{align*}
We have
$\norm{\left[\begin{array}{cc}y^{-1}y_1& y^{-1}y_2\\ 0&0
                               \end{array}\right]}_{}\le 1$
  and $\norm{\left[\begin{array}{cc}z_1z^{-1}& 0 \\ z_2z^{-1} & 0\end{array}\right]}{}\le 1$. Thus \[\norm{X}=\norm{\left[\begin{array}{cc}X&0\\
                          0&0     \end{array}\right]}{}\  \le\   \norm{\left[\begin{array}{cc}y^{-1}y_1& y^{-1}y_2\\ 0&0
                               \end{array}\right]}{}\norm{ \left[\begin{array}{cc}z_1z^{-1}& 0\\ z_2z^{-1}& 0
                               \end{array}\right]}{}\ \le \ 1\ .
                               \]
Then we have $x=x_1+x_2=yXz$ and
\begin{align*}
&\norm{yy^\dagger }_{L_1^\infty}\norm{z^\dagger X^\dagger Xz}_{L_1^\infty}
\\ &\qquad\qquad\qquad\le \norm{yy^\dagger }_{L_1^\infty}\norm{z^\dagger z}_{L_1^\infty}
\\ &\qquad\qquad\qquad\le\norm{y_1y_1^\dagger +y_2y_2^\dagger +\delta}_{L_1^\infty}\norm{z_1^\dagger z_1+z_2^\dagger z_2+\delta}_{L_1^\infty}
\\ &\qquad\qquad\qquad\le \Big(\norm{y_1y_1^\dagger }_{L_1^\infty}+\norm{y_2y_2^\dagger}_{L_1^\infty}+\delta\Big)\Big(\norm{z_1^\dagger z_1}_{L_1^\infty}+\norm{z_2^\dagger z_2}_{L_1^\infty}+\delta\Big)\
\end{align*}
Since $\delta>0$ is arbitrary, we have by \eqref{eq:scale}
\begin{align*}
&\norm{yy^\dagger }_{L_1^\infty}^{1/2}\norm{z^\dagger X^\dagger Xz}_{L_1^\infty}^{1/2}\\ &\qquad\qquad\qquad\le \norm{y_1y_1^\dagger}_{L_1^\infty}^{1/2}\norm{y_2y_2^\dagger}_{L_1^\infty}^{1/2}+\norm{z_1^\dagger z_1}_{L_1^\infty}^{1/2}\norm{z_2^\dagger z_2}_{L_1^\infty}^{1/2}
\end{align*}
This proves the triangular inequality and also
\[\norm{x}_{h}=\inf_{x=\sum_{i}y_iz_i}\norm{\sum_{i}y_iy_i^\dagger }_{L_1^\infty}^{1/2}\norm{\sum_{i}z^\dagger_iz_i}_{L_1^\infty}^{1/2} ,\]
where the supremum is over all finite families $\{y_i\}$ and $\{z_i\}$ such that $\sum_{i=1}^ky_iz_i=x$. We now use a standard Grothendieck-Pietsch factorization to show $\norm{\cdot}_{h}=\norm{\cdot}_{L_1^\infty}$.
Suppose $\norm{x}_{h}=1$. By Hahn-Banach Theorem, there exists a linear functional $\phi:\cM\to \mathbb{C}$ such that $\phi(x)=\norm{x}_{h}=1$ and for any finite families $\{y_i\}$ and $\{z_i\}$,
\begin{align*} |\sum_{i=1}^k\phi(y_iz_i)|\le \sup_{ \norm{a}_{L_\infty^1}=1}\langle a,\sum_{i}y_iy_i^\dagger \rangle_{\si_\tr}^{1/2}\sup_{\norm{b}_{L_\infty^1}=1}\langle b, \sum_{i}z_i^\dagger z_i\rangle_{\si_\tr}^{1/2}\,.
\end{align*}
Here we use the duality $L_1^\infty(\cN\subset\cM)^*=L_\infty^1(\cN\subset\cM)$ and for positive $Y\ge 0$,
\[ \norm{Y}_{L_1^\infty}=\sup_{\norm{a}_{L_\infty^1}=1}|\langle Y, a\rangle_{\si_\tr}|=\sup_{a\ge 0, \norm{a}_{L_\infty^1}=1}\langle Y, a\rangle_{\si_\tr}\ .\]
By modifying the phase factor and arithmetic-geometric mean inequality, we have
 \begin{align} \label{eq:positivity} \sup_{a\ge 0, \norm{a}_{L_\infty^1}=1}\langle a,\sum_{i}y_iy_i^\dagger \rangle_{\si_\tr}+\sup_{b\ge 0, \norm{b}_{L_\infty^1}=1}\langle b, \sum_{i}z_i^\dagger z_i\rangle_{\si,\tr}-2\sum_{i=1}^k \text{Re }\phi(y_iz_i)\ge 0\,.
\end{align}
Denote $B_{+}$ as the positive unit ball of $L_\infty^1(\cN\subset\cM)$ and
$C(B_{+}\times B_{+})$ as the real continuous function space. For each pair of finite families ${\bf y}=\{y_i\}$ and ${\bf z}=\{z_i\}$, we define the function \[f_{{\bf y}, {\bf z}}:B_+\times B_+\to \mathbb{R}\ , f_{{\bf y}, {\bf z}}(a,b)=\langle a,\sum_{i}y_iy_i^\dagger \rangle_{\si_\tr}+\langle b, \sum_{i}z_i^\dagger z_i\rangle_{\si,\tr}-\sum_{i=1}^k \text{Re} \phi(y_iz_i)\,.\  \]
We define the cones in $C(B_{+}\times B_{+})$
\begin{align*}
    &C=\{f_{{\bf y}, {\bf z}} \ | \ \{y_i\}, \{z_i\}\subset \cM  \}
    \\ &C_-=\{f\in C(B_{+}\times B_{+},\mathbb{R})  \ | \sup f<0  \}\,.
\end{align*}
Note that both $C$ and $C_-$ are convex and $C_-$ is open. Moreover, $C\cap C_-=\emptyset$ because of \eqref{eq:positivity}. By Hahn-Banach separation Theorem, there exists a linear functional
$\psi: C(B_{+}\times B_{+})\to \mathbb{R}$ such that
\[ \psi(f_-) \le \lambda \le \psi(f_{{\bf y}, {\bf z}})\]
for any $f_-\in C_-$ and $f_{{\bf y}, {\bf z}}\in C$. Since $C_-$ is a cone, $\lambda\ge 0$ and hence $\psi$ is a positive linear function. Up to normalization, there exists a probability measure $\mu$ on $B_+\times B_+$ such that $\psi(f)=\int_{B_+\times B_+}f(a,b)\ d\mu(a,b)$. Take
\[a_0=\int_{B_+\times B_+} a\  d\mu(a,b)\ ,\  b_0=\int_{B_+\times B_+} b\  d\mu(a,b)\,.\]
By convexity of $B_+$, we have $a_0,b_0\in B_+$ and moreover for every $y,z\in \cM$,
\begin{align*}
\psi(f_{\{y\}, \{z\}})=&\int_{B_+\times B_+}f(a,b)d\mu(a,b)\\ =&
\int_{B_+\times B_+}\langle a,yy^\dagger \rangle_{\si_\tr}d\mu(a,b)+\int_{B_+\times B_+}\langle b z^\dagger z\rangle_{\si,\tr}d\mu(a,b)-2 \text{Re} \phi(yz)
\\=&
\langle a_0,yy^\dagger \rangle_{\si_\tr}+\langle b_0,z^\dagger z\rangle_{\si,\tr}- \text{Re} \phi(yz)\ge 0\,.
\end{align*}
Rescaling $y_i$ and $z_i$ again, we have
\begin{align*}
|\phi(yz)|\le \tr(a_0\si_\tr^{1/2}yy^\dagger \si_\tr^{1/2})^{1/2} \tr(b_0\si_\tr^{1/2}z^\dagger z\si_\tr^{1/2})^{1/2}=\norm{a_0^{1/2}\si_\tr^{1/2}y}_{2,\tr}\norm{z\si_\tr^{1/2}b_0^{1/2}}_{2,\tr}\,,
\end{align*}
where $\|.\|_{2,\tr}$ denotes the Hilbert Schmidt norm. One can further find invertible $a_1,b_1\in (1+\eps)B_+$ such that
\begin{align*}
|\phi(yz)|\le \norm{a_1^{1/2}\si_\tr^{1/2}y}_{2,\tr}\norm{z\si_\tr^{1/2}b_1^{1/2}}_{2,\tr}
\end{align*}
Because of the invertibility of $a_1,b_1$, there exists a contraction $u$ such that
\[ \phi(yz)=\tr( ua_1^{1/2}\si_\tr^{1/2}yz\si_\tr^{1/2}b_1^{1/2})=\langle a_1^{1/2}u^\dag b_1^{1/2}, yz \rangle_{\si_\tr}\]
Note that by H\"older inequality,
\begin{align*} \norm{a_1^{1/2}u^\dag b_1^{1/2}}_{L_\infty^1}=&\sup_{\norm{X}_{2}=\norm{Y}_{2}=1}\norm{Xa_1^{1/2}u^\dag b_1^{1/2}Y}_{1}\\
\\=&\sup_{X,Y}\norm{\si_\tr^{1/2}Xa_1^{1/2}}_{2,\tr}
\norm{u^\dag b_1^{1/2}Y \si_\tr^{1/2}}_{2,\tr}
\\ =&\sup_{X,Y}\tr(\si_\tr Xa_1X^\dag)^{1/2}
\tr(Y^\dag b_1^{1/2}uu^\dag b_1^{1/2}Y\si_\tr)
\\ \le &\sup_{X,Y}\norm{Xa_1X^\dag}_{1}^{1/2}
\norm{Y^\dag b_1 Y}_{1}^{1/2}\le \norm{a_1}_{L_\infty^1}^{1/2}\norm{b_1}_{L_\infty^1}^{1/2}=1+\eps.
\end{align*}
Therefore, for any factroization $x=yz$,
$$\norm{x}_{h}=\phi(x)= \langle a_1^{1/2}u^\dag b_1^{1/2}, x \rangle_{\si_\tr}\le \norm{a_1^{1/2}u^\dag b_1^{1/2}}_{L_\infty^1}\norm{x}_{L_1^\infty}\le (1+\eps) \norm{x}_{L_1^\infty}.$$
Since $\eps$ is arbitrary, that concludes the proof. \end{proof}

The second factorization lemma is generalization of \cite[Theorem 1.5]{Pisier98NCvector} for subalgebra.
\begin{lem}\label{factorization2}
For any $x\in L_p(\sigma_\tr)$,
 \begin{align}\label{eq:factorization}
     \|x\|_{L_p(\sigma_\tr)}=\inf_{\substack{x=ayb\\a,b\in\cN}}\,\|a\|_{L_{2p}(\sigma_\tr)}\,\|y\|_{L_\infty^p(\cN\subset\cM)}\,\|b\|_{L_{2p}(\sigma_\tr)}\,.
 \end{align}
 where the infimum is over all factorizations $x=ayb$ with $a,b\in \cN$ and $y\in \cM$.
 \end{lem}

\begin{proof}Once again, we use the shorter notations  $\norm{\cdot}_{}:=\norm{\cdot}_{\infty}$,
 $\norm{\cdot}_{p}:=\norm{\cdot}_{L_p(\si_\tr)}$ and $\norm{\cdot }_{L_p^q}:=\norm{\cdot }_{L_p^q(\cN\subset\cM)}$. Denote
\[\norm{x}_{h}:=\inf_{x=ayb}\norm{a}_{2p}\norm{y}_{L_{\infty}^{p}}\norm{b}_{2p}\ \]
where the infimum is over all factorizations $x=ayb$ with $a,b\in \cN$ and $y\in \cM$.
We first show that $\norm{\cdot}_{h}$ is a norm. Let $x_1=a_1y_1b_1$ and $x_2=a_2y_2b_2$ with $\norm{y_1}_{L_\infty^p}=\norm{y_2}_{L_\infty^p}=1$.
Take small $\delta>0$ and denote $a=(a_1a_1^\dag+a_2a_2^\dag+\delta \Id)^{1/2}, b=(b_1^\dag b_1+b_2^\dag b_2+\delta \Id)^{1/2}$.
We have $x_1+x_2= ayb$ for
\begin{align*}
y= a^{-1}(x_1+x_2)b^{-1}=a^{-1}(a_1y_1b_1+a_2y_2b_2)b^{-1}
\end{align*}
We show that $$\norm{y}_{L_\infty^p} \le \max\{ \norm{y_1}_{L_\infty^p}, \norm{y_1}_{L_\infty^p}\}\le 1
$$
Indeed, we have
\begin{align*}
\left[\begin{array}{cc}y&0\\ 0&0\end{array}\right]=\left[\begin{array}{cc} a^{-1}a_1&a^{-1}a_2\\ 0&0\end{array}\right]\cdot \left[\begin{array}{cc}y&0\\ 0&0\end{array}\right]\cdot\left[\begin{array}{cc}b_1b^{-1}&0\\ b_2 b^{-1}& 0\end{array}\right]
\end{align*}
Denote $A=\left[\begin{array}{cc} a^{-1}a_1&a^{-1}a_2\\ 0&0\end{array}\right]$ and $B=\left[\begin{array}{cc}b_1b^{-1}&0\\ b_2 b^{-1}& 0\end{array}\right]$. We see that $A,B\in \mathbb{M}_2(\cN)$ and
\begin{align*}
&\norm{A}_{\infty}=\norm{a^{-1}(a_1a_1^\dag+a_2a_2^\dag)a^{-1}}_{\infty}\le 1
\\ &\norm{B}_{\infty}=\norm{b^{-1}(b_1^\dag b_1+b_2^\dag b_2)b^{-1}}_{\infty}\le 1
\end{align*}
Then by the operator space structure of $L_\infty^p$,
\begin{align*}
\norm{y}_{L_\infty^p} =&\norm{\left[\begin{array}{cc}y&0\\ 0&0\end{array}\right]}_{L_\infty^p(\mathbb{M}_2(\cN)\subset\mathbb{M}_2(\cM))}\\
=&\norm{A \cdot \left[\begin{array}{cc}y_1&0\\ 0&y_2\end{array}\right]\cdot B}_{L_\infty^p(\mathbb{M}_2(\cN)\subset\mathbb{M}_2(\cM))}
\\ \le &\norm{A}_{\infty}\cdot \norm{\left[\begin{array}{cc}y_1&0\\ 0&y_2\end{array}\right]}_{L_\infty^p(\mathbb{M}_2(\cN)\subset\mathbb{M}_2(\cM))}\cdot \norm{B}_{\infty}
\\ \le  &\max\{ \norm{y_1}_{L_\infty^p}, \norm{y_1}_{L_\infty^p}\}\le 1\,.
\end{align*}
Therefore, we have $x_1+x_2= ayb$ for
\begin{align*}
&\norm{a}_{2p}=\norm{a_1a_1^\dag+a_2a_2^\dag+\delta \Id}_{p}^{1/2}\le \Big( \norm{a_1a_1^\dag}_{{p}}+\norm{a_2a_2^\dag}_{{p}}+\delta\Big)^{1/2}\\
&\norm{b}_{2p}=\norm{b_1^\dag b_1+b_2^\dag b_2+\delta \Id}_{p}^{1/2}\le \Big( \norm{b_1^\dag b_1}_{{p}}+\norm{b_2^\dag b_2}_{{p}}+\delta\Big)^{1/2}\,.\\
\end{align*}
Thus
\begin{align*}
\norm{x_1+x_2}_{h}\le &\norm{a}_{{2p}}\norm{y}_{L_\infty^p} \norm{b}_{{2p}}
\\ \le &\Big( \norm{a_1a_1^\dag}_{{p}}+\norm{a_2a_2^\dag}_{{p}}+\delta\Big)^{1/2}\Big( \norm{b_1^\dag b_1}_{{p}}+\norm{b_2^\dag b_2}_{{p}}+\delta\Big)^{1/2}\,.
\end{align*}
Taking $\delta\to 0$ and rescaling $a_1,b_1$, we obtain
\begin{align*}
\norm{x_1+x_2}_{h}
\le \norm{a_1}_{{2p}}\norm{b_1}_{{2p}}+\norm{a_2}_{2p}\norm{b_2}_{{2p}}
\end{align*}
This proves the triangle inequality.

We now show that $\norm{\cdot}_{h}=\norm{\cdot}_{p}$ coincide. First, note that by definition of $L_\infty^p$
\[ \norm{y}_{L_{\infty}^{p}}\ge \frac{\norm{ayb}_{{p}}}{\norm{a}_{{2p}}\norm{b}_{2p}}\]
for any $a,b \in \cN, y\in \cM$. This implies $\norm{\cdot}_{h}\ge \norm{\cdot}_{p}$.
To see the converse direction, we consider the dual norm $\norm{\cdot}_{h^*}$ given by the sesquilinear pairing $\langle \cdot,\cdot\rangle_{\si_\tr}$. Then
\begin{align*}
\norm{z}_{h^*}=&\sup_{\norm{x}_{h}=1}|\langle z,x \rangle_{\si_\tr} |
\\ =&\underset{{\substack{a,b\in \cN, y\in \cM \\
\norm{a}_{{2p}}=\norm{b}_{{2p}}=\norm{y}_{L_\infty^{p}} =1}}}{\sup} |\langle z,ayb \rangle_{\si_\tr}|
\\ =&\underset{{\substack{a,b\in \cN, y\in \cM \\
\norm{a}_{{2p}}=\norm{b}_{{2p}}=\norm{y}_{L_\infty^{p}} =1}}}{\sup}  |\langle a^\dag zb^\dag,y  \rangle_{\si_\tr} |
\\ =&\underset{{\substack{a,b\in \cN, y\in \cM \\
\norm{a}_{{2p}}=\norm{b}_{{2p}}=1}}}{\sup}  \norm{a^\dag zb^\dag}_{L_1^{p'}}
\\ =&\underset{\substack{a_1,b_1>0, a_1,b_1\in \cN \\ \norm{a}_{{2p}}=\norm{b}_{{2p}}=1\\ }}{\inf}\underset{{\substack{a,b\in \cN, y\in \cM \\
\norm{a}_{{2p}}=\norm{b}_{{2p}}=1}}}{\sup}\norm{b_1^{-1}b^\dag za^\dag a_1^{-1}}_{{p'}}
\\ \ge &\underset{{\substack{a,b\in \cN, y\in \cM \\
\norm{a}_{{2p}}=\norm{b}_{{2p}}=1}}}{\sup}
\underset{\substack{a_1,b_1>0, a_1,b_1\in \cN\\ \norm{a}_{{2p}}=\norm{b}_{{2p}}=1\\ }}{\inf}\norm{b_1^{-1}b^\dag za^\dag a_1^{-1}}_{{p'}}
\\ \ge & \norm{z}_{{p'}}
\end{align*}
where the last inequality follows from choosing $a=a_1,b=b_1$. This proves $\norm{z}_{h^*}\ge \norm{z}_{{p'}}$ which by duality gives $\norm{x}_{h}\le \norm{x}_{{p}}$. That completes the proof.

\end{proof}

We now discuss the operator space structure of amalgamated $L_p$ spaces (see the analogous treatment of the symmetric case in \cite[Appendix A.2]{gaoindex}). We first define the operator space structure of $L_\infty^1(\mathcal{N}\subset\cM)$, then extend the structure to $L_\infty^p(\mathcal{N}\subset\cM)$ for other $p$ by interpolation.
% , and finally argue the case of $L_{p'}^1(\mathcal{N}\subset\cM$ by duality.
We define the matrix norm of $L_\infty^1(\mathcal{N}\subset\cM)$ via the following isometry,
\begin{align}\label{eq:opspacestruc}
    \mathbb{M}_n(L_\infty^1(\cN\subset\cM))\cong L_\infty^1(\mathbb{M}_n(\cN)\subset \mathbb{M}_n(\cM))\,,
\end{align}
where for each $n\in\mathbb{N}$ the enlarged amalgamated space on $L_\infty^1(\mathbb{M}_n(\cN)\subset \mathbb{M}_n(\cM))$ is defined with respect to the state $\sigma_\tr^{(n)}:=n^{-1}\Id_n\otimes\sigma_\tr $. We verify that these norms satisfy Ruan's axioms recalled in \Cref{sec:opspace}. Let $e_1$ (resp. $e_2$) be projection onto $\mathbb{M}_m(\cM)$ (resp. $\mathbb{M}_n(\cM)$).
 Given $v\in \mathbb{M}_m(\cM)$ and $w\in \mathbb{M}_n(\cM)$, we have $x:=v\oplus w$ where $v=e_1xe_1$ and $x=e_2xe_2$. Then,
\begin{align}
    \|x\|_{L_\infty^1(\mathbb{M}_{m+n}(\cN)\subset \mathbb{M}_{m+n}(\cM))}&=\sup_{\substack{\|a\|_{L_2(\sigma_\tr^{(m+n)})}=1\nonumber\\ \|b\|_{L_2(\sigma_\tr^{(m+n)})}=1}}\,\|axb\|_{L_1(\sigma_\tr^{(m+n)})}\\
   & \le\sup_{\substack{\|a\|_{L_2(\sigma_\tr^{(m+n)})}=1\nonumber\\ \|b\|_{L_2(\sigma_\tr^{(m+n)})}=1}}\,\|ae_1xe_1b\|_{L_1(\sigma_\tr^{(m+n)})}+\|ae_2xe_2b\|_{L_1(\sigma_\tr^{(m+n)})}\nonumber\\
   &\le\sup_{\substack{\|a\|_{L_2(\sigma_\tr^{(m+n)})}=1\nonumber\\ \|b\|_{L_2(\sigma_\tr^{(m+n)})}=1}}\,\|ae_1ve_1b\|_{L_1(\sigma_\tr^{(m+n)})}+\|ae_2we_2b\|_{L_1(\sigma_\tr^{(m+n)})}\nonumber\\
  &\le\sup_{\substack{\|a\|_{L_2(\sigma_\tr^{(m+n)})}=1\nonumber\\ \|b\|_{L_2(\sigma_\tr^{(m+n)})}=1}}\,\||ae_1|\,v\,|(e_1b)^\dagger|\|_{L_1(\sigma_\tr^{(m)})}   +\||ae_2|\,w\,|(e_2b)^\dagger|\|_{L_1(\sigma_\tr^{(n)})}\label{eq:taking||}\,. \\
\end{align}
In \eqref{eq:taking||} above, we used that
\begin{align*}
    \|ae_1ve_1b\|_{L_1(\sigma_\tr^{(m+n)})}&\overset{(1)}{=}\tr\Big|(\sigma_\tr^{(m+n)})^{\frac{1}{2}} ae_1ve_1b(\sigma_\tr^{(m+n)})^{\frac{1}{2}}\Big|\overset{(1)}{=}\tr\Big|(\sigma_\tr^{(m+n)})^{\frac{1}{2}} U|ae_1|\,v \,|(e_1b)^\dagger|U'(\sigma_\tr^{(m+n)})^{\frac{1}{2}}\Big|\,\\
    &\overset{(2)}{=}\tr\Big|U(\sigma_\tr^{(m+n)})^{\frac{1}{2}} |ae_1|\,v \,|(e_1b)^\dagger|(\sigma_\tr^{(m+n)})^{\frac{1}{2}}U'\Big|=\||ae_1|\,v\,|(e_1b)^\dagger|\|_{L_1(\sigma_\tr^{(m+n)})}
\end{align*}
where in $(1)$, $U$ and $U'$ respectively are given by the polar decomposition $ae_1:=U|ae_1|$ and $(e_1b)^\dagger=|(e_1b)|(U')^\dag$. In $(2)$, we use the fact that both $U'$ and $U$ belong to $\mathbb{M}_{m+n}(\cN)$ and hence commute with $\sigma_\tr^{(m+n)}$. The same trick applies to $ae_2$ and $e_2b$. Moreover, by the fact that both $|ae_1|$ and $|(e_1b)^\dagger|$ belong to $\mathbb{M}_{m}(\cN)$ (resp. $|ae_2|, \,|(e_2b)^\dagger|\in\mathbb{M}_n(\cN)$), \eqref{eq:taking||} becomes
\begin{align}
     \|x\|_{L_\infty^1(\mathbb{M}_{m+n}(\cN)\subset \mathbb{M}_{m+n}(\cM))}&\le \sup_{\substack{\|a\|_{L_2(\sigma_\tr^{(m+n)})}=1\nonumber\\ \|b\|_{L_2(\sigma_\tr^{(m+n)})}=1}}\,\||ae_1|\|_{L_2(\sigma_\tr^{(m)})}\|v\|_{L_\infty^1(\mathbb{M}_m(\cN)\subset \mathbb{M}_m(\cM))}\||(e_1b)^\dagger|\|_{L_2(\sigma_\tr^{(m)})}   \\
     &\qquad +\||ae_2|\|_{L_2(\sigma_\tr^{(n)})}\|w\|_{L_\infty^1(\mathbb{M}_n(\cN)\subset \mathbb{M}_n(\cM))}\||(e_2b)^\dagger|\|_{L_2(\sigma_\tr^{(n)})}\,.\label{eq:almostthere}
\end{align}
Now, we observe that
 \begin{align*}
   \| a\|_{L_2(\sigma_\tr^{(m+n)})}^2&=\|a(e_1+e_2)\|_{L_2(\sigma_\tr^{(m+n)})}^2\\ &=\tr\big((\sigma_\tr^{(m+n)})^{\frac{1}{2}}\,(e_1+e_2)^\dagger a^\dagger (\sigma_\tr^{(m+n)})^{\frac{1}{2}}\, a(e_1+e_2)\big)\\
&\overset{(1)}{=}\tr\big((\sigma_\tr^{(m+n)})^{\frac{1}{2}}
   \,(e_1+e_2)(e_1+e_2)^\dagger a^\dagger {(\sigma_{\tr}^{(m+n)})}^{\frac{1}{2}} \,a\big)\\
   &=\tr\big((\sigma_\tr^{(m+n)})^{\frac{1}{2}}
   \,(e_1+e_2)(e_1+e_2)^\dagger a^\dagger {(\sigma_{\tr}^{(m+n)})}^{\frac{1}{2}} \,a\big)\\
   &=\tr\big((\sigma_\tr^{(m+n)})^{\frac{1}{2}}
   \,e_1 a^\dagger {(\sigma_{\tr}^{(m+n)})}^{\frac{1}{2}} \,a\big)+\tr\big((\sigma_\tr^{(m+n)})^{\frac{1}{2}}
   \,e_2 a^\dagger {(\sigma_{\tr}^{(m+n)})}^{\frac{1}{2}} \,a\big)\\
   &\overset{(2)}{=}\| ae_1\|_{L_2(\sigma_\tr^{(m+n)})}^2+\| ae_2\|_{L_2(\sigma_\tr^{(m+n)})}^2\\
   &=\| |ae_1|\|_{L_2(\sigma_\tr^{(m)})}^2+\| |ae_2|\|_{L_2(\sigma_\tr^{(n)})}^2\,.
\end{align*}
 In $(1)$ and $(2)$, we used the tracial property and the fact that both projections $e_1$ and $e_2$ commute with $\sigma_\tr^{(m+n)}$ since the
  $e_i$'s act trivially on $\cM$ whereas $\sigma_\tr^{(m+n)}$ acts trivially on $\mathbb{M}_{m+n}$. Similarly,
 \begin{align*}
     \| b\|_{L_2(\sigma_\tr^{(m+n)})}^2=  \| |(e_1b)^\dagger|\|_{L_2(\sigma_\tr^{(m)})}^2+\| |(e_2b)^\dagger|\|_{L_2(\sigma_\tr^{(n)})}^2\,.
 \end{align*}
Applying H\"{o}lder's inequality to \eqref{eq:almostthere}, we have that
 \begin{align*}
 &  \|x\|_{L_\infty^1(\mathbb{M}_{m+n}(\cN)\subset \mathbb{M}_{m+n}(\cM))}\\
& \qquad\le \Big(\| |ae_1|\|_{L_2(\sigma_\tr^{(m)})}^2+\| |ae_2|\|_{L_2(\sigma_\tr^{(n)})}^2\Big)^{\frac{1}{2}}\cdot \max\{\|v\|_{L_\infty^1(\mathbb{M}_m(\cN)\subset \mathbb{M}_m(\cM))},\|w\|_{L_\infty^1(\mathbb{M}_n(\cN)\subset \mathbb{M}_n(\cM))}\} \\
   &\qquad\qquad  \cdot \Big (\|(e_1b)^\dagger\|_{L_2(\sigma_\tr^{(m)})}^2+ \| |(e_2b)^\dagger|\|_{L_2(\sigma_\tr^{(n)})}^2\Big)^{\frac{1}{2}}\\
   &\qquad =\max\{\|v\|_{L_\infty^1(\mathbb{M}_m(\cN)\subset \mathbb{M}_m(\cM))},\|w\|_{L_\infty^1(\mathbb{M}_n(\cN)\subset \mathbb{M}_n(\cM))}\}\,.
\end{align*}
This verifies (i) in Ruan's axioms. For the second axiom (ii), consider $x\in \mathbb{M}_m(\cM)$ and $A\in \mathbb{M}_{n,m}$, $B\in \mathbb{M}_{m,n}$. We have
\begin{align*}
    \|AxB\|_{L_\infty^1(\mathbb{M}_n(\cN)\subset \mathbb{M}_n(\cM))}&=\sup_{\substack{\|a\|_{L_2(\sigma_\tr^{(n)})}=1\\ \|b\|_{L_2(\sigma_\tr^{(n)})}=1}} \|aAxBb\|_{L_1(\sigma_\tr^{(n)})}\\
    &\overset{(1)}{\le} \sup_{\substack{\|a\|_{L_2(\sigma_\tr^{(n)})}=1\\ \|b\|_{L_2(\sigma_\tr^{(n)})}=1}} \|aA\|_{L_2(\sigma_\tr^{(n)})}\,\|x\|_{L_\infty^1(\mathbb{M}_m(\cN)\subset \mathbb{M}_m(\cM))}\,\|Bb\|_{L_2(\sigma_\tr^{(n)})}\\
    &\overset{(2)}{\le} \sup_{\substack{\|a\|_{L_2(\sigma_\tr^{(n)})}=1\\ \|b\|_{L_2(\sigma_\tr^{(n)})}=1}} \|a\|_{L_2(\sigma_\tr^{(n)})}\,\|A\|_{\infty}\,\|x\|_{L_\infty^1(\mathbb{M}_m(\cN)\subset \mathbb{M}_m(\cM))}\,\|B\|_\infty\,\|b\|_{L_2(\sigma_\tr^{(n)})}\\
     &=  \,\|A\|_{\infty}\,\|x\|_{L_\infty^1(\mathbb{M}_m(\cN)\subset \mathbb{M}_m(\cM))}\,\|B\|_\infty\,.
\end{align*}
Here, equation (1) above follows again from the definition of $L_\infty^1(\mathbb{M}_m(\cN)\subset \mathbb{M}_m(\cM))$.
In $(2)$, we used H\"{o}lder's inequality and the fact that $\norm{A}_{\infty}:=\norm{A}_{\mathbb{M}_{n,m}}=\norm{A\ten {\bf 1}_\cM}_{L_\infty(\sigma_\tr^{(n)})}$, and similarly for $B$. This verifies Ruan's axiom (ii), which implies that \Cref{eq:opspacestruc} indeed provides an operator space structure for $L_\infty^1(\cN\subset\cM)$. As mentioned, this observation can be extended to the case of $L_\infty^p(\mathcal{N}\subset\cM)$ by complex interpolation:
\begin{prop}\label{opspaceLqp}
For any $1\le p\le \infty$, the identification
\begin{align*}
    \mathbb{M}_n(L_\infty^p(\cN\subset \cM))\cong
     L_\infty^p(\mathbb{M}_n(\cN)\subset\mathbb{M}_n(\cM))\,
\end{align*}
defines an operator space structure on $L_\infty^p(\cN\subset \cM)$.
\end{prop}
\begin{proof}
We recall that $L_\infty^\infty(\cN\subset \cM)\cong L_\infty(\cM)$ by definition and its operator space structure is naturally given by $\mathbb{M}_n(L_\infty(\cM))\cong L_\infty(\mathbb{M}_n(\cM))$. Then by the interpolation \eqref{eq:interp}:
\begin{align*}
    L_\infty^p(\mathbb{M}_{n}(\cN)\subset \mathbb{M}_n(\cM))&\cong \big[L_\infty(\mathbb{M}_n(\cM)),\,L_\infty^1(\mathbb{M}_n(\cN)\subset \mathbb{M}_n(\cM)) \big]_{\frac{1}{p}}\\
    &\cong \big[\mathbb{M}_n(L_\infty(\cM)),\,\mathbb{M}_n(L_\infty^1(\cN\subset \cM)) \big]_{\frac{1}{p}}\\
    &\cong \mathbb{M}_n(L_\infty^p(\cN\subset\cM))\,.
\end{align*}
\end{proof}

Combining \Cref{opspaceLqp} with \eqref{eq:factorization}, we find (see also \cite[Lemma 3.12]{gao2020fisher} in the tracial case):
\begin{prop}
 For any $\cN$-bimodule map $\Phi:\cM\to \cM$ and $1\le p,q,s\le \infty$,
 		\begin{align}\label{eq:bimodulecb}
 	\|\Phi:\,L_\infty^p(\cN\subset \cM)\to L_\infty^q(\cN\subset \cM)\|=\|\Phi:L_s^p(\cN\subset\cM)\to L_s^q(\cN\subset \cM)\|\,.
 		\end{align}
In the case $p=2$, assuming that $\Phi$ is self-adjoint with respect to some invariant state $\sigma=E_{\cN*}(\sigma)$, we further have
\begin{align}\label{eqgapcb}
\norm{\Phi:L_2(\rho)\to L_2(\rho)}{}&=\norm{\Phi:L_2(\sigma_\tr)\to L_2(\sigma_\tr)}{}\\
&= \norm{\Phi:L_\infty^2(\cN\subset\cM)\to L_\infty^2(\cN\subset \cM)}{} \,,\nonumber
\end{align}
for any other invariant state $\rho=E_{\cN*}(\rho)$.
\end{prop}
\begin{proof}
We first prove the direction "$\ge$" in \eqref{eq:bimodulecb} for $s=1$: let $x\in L_1^p(\cN\subset\cM)$ with $\norm{x}_{L_1^p(\cN\subset\cM)}<1$. By definition, there exists a decomposition $x=ayb$ of $x$ with $a,b\in\cN$ and $\|a\|_{L_{2p'}(\sigma_\tr)},\,\|b\|_{L_{2p'}(\sigma_\tr)}< 1$ and $\|y\|_{L_p(\sigma_\tr)}\le 1$, where $p'$ is the H\"{o}lder conjugate of $p$. Moreover, by the factorization property \eqref{eq:factorization}, $y=czd$ with $c,d\in\cN$, $\|c\|_{L_{2p}(\sigma_\tr)},\,\| d\|_{L_{2p}(\sigma_\tr)}< 1$ and $\|z\|_{L_\infty^p(\cN\subset\cM)}< 1$. Therefore, $x=(ac)z(db)$ and by H\"older inequality
\begin{align}
    \|ac\|_{L_2(\sigma_\tr)}\le \|a\|_{L_{2p'}(\sigma_\tr)}\|c\|_{L_{2p}(\sigma_\tr)}< 1\label{eq:ac}
\end{align}
where we used the fact $\sigma_\tr$ is a trace on $\cN$. Similarly, we have $\|db\|_{L_2(\sigma_\tr)}< 1$. Therefore, one can find $a',b',c',d'\in \cN$ such that $ac=a'c'$, $db=d'b'$ with $$\|a'\|_{L_{2q}(\sigma_\tr)}, \|c'\|_{L_{2q'}(\sigma_\tr)}, \|b'\|_{L_{2q'}(\sigma_\tr)}, \|d'\|_{L_{2q}(\sigma_\tr)}\le 1 \ .$$
Indeed, using the polar decomposition $ac=U|ac|=U|ac|^{1/q}|ac|^{1/q'}$, one can choose $a'=U|ac|^{1/q}$ and $c'=|ac|^{1/q'}$. We then use the module property of the map $\cN$, so that $\Phi(x)=\Phi(a'c'zd'b')=a'c'\Phi(z)d'b'$.
Since $\|z\|_{L_\infty^p(\cN\subset\cM)}\le 1$, we have
\begin{align*}
    \|\Phi(z)\|_{L_\infty^q(\cN\subset\cM)}&\le \|\Phi:L_\infty^p(\cN\subset\cM)\to L_\infty^q(\cN\subset\cM)\|\,.
\end{align*}
By definition of the $L_1^q$ and $L_\infty^q$ norms,
\begin{align*}\|\Phi(x)\|_{L_1^q(\cN\subset\cM)}\le \|c'\Phi(z)d'\|_{L_q(\cN\subset\cM)}\le&  \|c'\|_{L_{2q}(\sigma_\tr)}\,\|\Phi(z)\|_{L_\infty^q(\cN\subset\cM)}\,\|d'\|_{L_{2q}(\sigma_\tr)}\\ \le&
\|\Phi:L_\infty^p(\cN\subset\cM)\to L_\infty^q(\cN\subset\cM)\|\,,\end{align*}
%using once again the factorization property \eqref{eq:factorization} together with the norm estimates found above,
%\begin{align*}
 %   \|c'\Phi(z)d'\|_{L_q(\sigma_\tr)}&\le \|c'\|_{L_{2q}(\sigma_\tr)}\,\|\Phi(z)\|_{L_\infty^q(\cN\subset\cM)}\,\|d'\|_{L_{2q}(\sigma_\tr)}\\
  %  &\le \|\Phi:L_\infty^p(\cN\subset\cM)\to L_\infty^q(\cN\subset\cM)\|\,.
%\end{align*}
%Therefore, by definition of $L_1^q$, we conclude that $$\|\Phi(x)\|_{L_1^q(\cN\subset\cM)}\le \|c'\Phi(z)d'\|_{L_1^q(\cN\subset\cM)}\le \|\Phi:L_\infty^p(\cN\subset\cM)\to L_\infty^q(\cN\subset\cM)\|\,,$$
%and
Therefore
$$\|\Phi:L_1^p(\cN\subset\cM)\to L_1^q(\cN\subset\cM)\|\le  \|\Phi:L_\infty^p(\cN\subset\cM)\to L_\infty^q(\cN\subset\cM)\|\,.$$ By complex interpolation \cite{bergh2012interpolation}, we can prove that this last claim holds true for all $1\le s\le \infty$:
\begin{align}\label{eq:interpola}
    \|\Phi:L_s^p(\cN\subset\cM)\to L_s^q(\cN\subset\cM)\|\le  \|\Phi:L_\infty^p(\cN\subset\cM)\to L_\infty^q(\cN\subset\cM)\|\,.
\end{align}
We denote by $\Phi^{\operatorname{KMS}}$ the dual of $\Phi$ with respect to the $\sigma_\tr$-KMS inner product \eqref{eq:sigmakms}. Then, applying duality \eqref{eq:duality} to \eqref{eq:interpola}, we get
\begin{align*}
    \|\Phi:L_s^p(\cN\subset\cM)\to L_s^q(\cN\subset\cM)\|&=    \|\Phi^{\operatorname{KMS}}:L_{s'}^{q'}(\cN\subset\cM)\to L_{s'}^{p'}(\cN\subset\cM)\|\\
    &\le \|\Phi^{\operatorname{KMS}}:L_{\infty}^{q'}(\cN\subset\cM)\to L_{\infty}^{p'}(\cN\subset\cM)\|\\
    &= \|\Phi:L_{1}^{p}(\cN\subset\cM)\to L_{1}^{q}(\cN\subset\cM)\|\,.
\end{align*}
Using duality one last time, we get
\begin{align*}
    \|\Phi: L_\infty^p(\cN\subset\cM)\to L_\infty^q(\cN\subset\cM)\|&=\|\Phi^{\operatorname{KMS}}: L_1^{q'}(\cN\subset\cM)\to L_1^{p'}(\cN\subset\cM)\|\\
    &\le \|\Phi^{\operatorname{KMS}}: L_{s'}^{q'}(\cN\subset\cM)\to L_{s'}^{p'}(\cN\subset\cM)\|\\
    &= \|\Phi: L_{s}^{p}(\cN\subset\cM)\to L_{s}^{q}(\cN\subset\cM)\|\\
    &  \le \|\Phi: L_\infty^p(\cN\subset\cM)\to L_\infty^q(\cN\subset\cM)\|\,,
\end{align*}
and hence all the norms coincide. This concludes the proof of \eqref{eq:bimodulecb}. %Thanks to Proposition \ref{opspaceLqp}, we apply this identity to $\mathbb{M}_n(\cN)\subset \mathbb{M}_n(\cM)$ for all $n$
%\begin{align*}
  %  \norm{\Phi:L_\infty^p(\cN\subset \cM)\to L_\infty^q(\cN\subset \cM)}_{\operatorname{cb}} =&\sup_n \|\Phi: L_\infty^{p}(\mathbb{M}_n(\cN)\subset\mathbb{M}_n(\cM))\to L_\infty^{q}(\mathbb{M}_n(\cN)\subset\mathbb{M}_n(\cM))\|
   % \\ =&\sup_n \|\Phi: L_s^{p}(\mathbb{M}_n(\cN)\subset\mathbb{M}_n(\cM))\to L_s^{q}(\mathbb{M}_n(\cN)\subset\mathbb{M}_n(\cM))\|
  %   \\ =& \|\Phi: L_s^{p}(\cN\subset\cM)\to L_s^{q}(\cN\subset\cM)\|_{\operatorname{cb}}
%\end{align*}
%where we used the fact that for $1\le s\le \infty$, the $\operatorname{CB}$-norm of $\Phi$ can be equivalently given by (see  \cite[Lemma 1.7]{Pisier98NCvector})
%\begin{align*} \|\Phi: L_s^{p}(\cN\subset\cM)\to L_s^{q}(\cN\subset\cM)\|_{\operatorname{cb}}
%=\sup_n \|\Phi: L_s^{p}(\mathbb{M}_n(\cN)\subset\mathbb{M}_n(\cM))\to L_s^{q}(\mathbb{M}_n(\cN)\subset\mathbb{M}_n(\cM))\|
%\end{align*}
In the case $p=2$, it is easy to see that all matrix norms are equal to each others and correspond to the spectral radius of $\Phi$ as a self-adjoint operator on $L_2$, from which \eqref{eqgapcb} follows. (see \cite[Lemma 2.6]{gao2021spectral} for independence of invariant state $\rho$).
\end{proof}

We end this section with a Lemma that connects up the notion of $\operatorname{cb}$-index to the operator space structure of amalgamated $L_p$ norms (see \cite[Theorem 3.9]{gaoindex} for the tracial case).
\begin{lem}\label{lemmaindexlpspace}
Let $\cN\subset \cM$ be finite dimensional von Neumann subalgebras and $E_{\cN}:\cM\to \cN$ be a conditional expectation. Then
\begin{align}
&C_{\tau}(\cM:\cN)\ge  \|\id:L_\infty^1(\cN\subset\cM)\to \cM\|_{} \label{eq:indexnorm}\\
&C_{\tau,\operatorname{cb}}(\cM:\cN)\ge  \|\id:L_\infty^1(\cN\subset\cM)\to \cM\|_{\operatorname{cb}}\,.
\label{eq:cbindexnorm}
\end{align}
\end{lem}
\begin{proof}
First by duality, we have
\begin{align}
    \|\id:L_1(\si_\Tr)\to L_1^\infty(\cN\subset\cM)\|=\|\id: L_\infty^1(\cN\subset \cM)\to \cM\|.
\end{align}
We show that it suffices to consider positive elements in the optimization of the norm on the left hand side. Denote
\[\|\id:L_1(\si_\Tr)\to L_1^\infty(\cN\subset\cM)\|_{+}=\underset{\substack{x\ge 0\\
 	\tr(\si x)=1}}{\sup} \norm{x}_{L_1^\infty(\cN\subset\cM)} . \]
For general $x\in \cM$, we have
\begin{align*}
\norm{x}_{L_1(\si_\Tr)}=\norm{\si_\Tr^{1/2}x\si_\Tr^{1/2}}_{L_1(\tr)}=\inf_{y'z'=\si_\Tr^{1/2}x\si_\Tr^{1/2}}\norm{y'}_{L_2(\tr)} \norm{z'}_{L_2(\tr)}\,,
\end{align*}
where $\|.\|_{L_p(\tr)}$ denotes the usual $p$-Schatten norm. Take $y=\si_\Tr^{-1/2}y'$ and $z=z'\si_\Tr^{-1/2}$. We have
\begin{align*}
\norm{x}_{L_1(\si_\Tr)}=&\inf_{yz=x}\norm{\si_\Tr^{1/2}y}_{L_2(\tr)} \norm{z\si_\Tr^{1/2}}_{L_2(\tr)}
\\=&\inf_{yz=x}\norm{\si_\Tr^{1/2}yy^\dagger \si_\Tr^{1/2}}_{L_1(\tr)}^{1/2} \norm{\si_\Tr^{1/2}z^\dagger z\si_\Tr^{1/2}}_{L_1(\tr)}^{1/2}
\\=&\inf_{yz=x}\norm{yy^\dagger }_{L_1(\si_\tr)}^{1/2} \norm{z^\dagger z}_{L_1(\si_\tr)}^{1/2}\,.
\end{align*}
Then, we have by the factorization property \eqref{eq:factorization2} that
\begin{align*}\norm{x}_{L_1^\infty(\cN\subset\cM)}\le &\inf_{x=yz} \norm{yy^\dagger }_{L_1^\infty(\cN\subset\cM)}^{1/2}\norm{z^\dagger z}_{L_1^\infty(\cN\subset\cM)}^{1/2} \\ \le&
\norm{\id:L_1(\si_\Tr)\to L_1^\infty(\cN\subset\cM)}_{+}\inf_{x=yz} \norm{yy^\dagger }_{L_1(\si_\tr)}^{1/2}\norm{z^\dagger z}_{L_1(\si_\tr)}^{1/2}
\\ =&
\norm{\id:L_1(\si_\Tr)\to L_1^\infty(\cN\subset\cM)}_{+}\norm{x}_{L_1(\si_\Tr)}.
\end{align*}
which proves $$\norm{\id:L_1(\si_\Tr)\to L_1^\infty(\cN\subset\cM)}_{+}=\norm{\id:L_1(\si_\Tr)\to L_1^\infty(\cN\subset\cM)}{}\,.$$
Therefore, by definition of the $L_1^\infty(\cN\subset\cM)$ norm
\begin{align*}
    \norm{\id:L_1(\si_\Tr)\to L_1^\infty(\cN\subset\cM)}_{+}=& \underset{\substack{x\ge 0\\
 	\tr(\si x)=1}}{\sup} \norm{x}_{L_1^\infty(\cN\subset\cM)}
    \\ =& \underset{\substack{x\ge 0\\
 	\tr(\si x)=1}}{\sup} \underset{\substack{a>0, a\in \cN \\  \tr(\si_\tr a)=1}}{\inf} \norm{a^{-1/2}x{a^{-1/2}}}_{\infty}
    \\ =& \underset{\substack{x\ge 0\\
 	\tr(\si x)=1}}{\sup} \underset{\substack{a>0, a\in \cN \\  \tr(\si_\tr a)=1}}{\inf} \inf \{c>0\ |\ x\le ca\} \
     \\ \overset{(1)}{=}& \underset{\substack{\rho\ge 0 \\ \tr(\rho)=1}}{\sup} \underset{\substack{\omega>0, E_{\cN*}(\omega)=\omega \\  \tr(\omega)=1}}{\inf} \inf \{c>0\ |\ \rho\le c\omega\}
     \\ \overset{(2)}{\le }& \underset{\substack{\rho\ge 0 \\ \tr(\rho)=1}}{\sup}  \inf \{c>0\ |\ \rho\le cE_{\cN *}(\rho)\}  := C_{\tau}(\cM:\cN)\ .
\end{align*}
Here, the equality (1) uses the substitution $\rho=\si_{\tr}^{\frac{1}{2}}x \si_{\tr}^{\frac{1}{2}}$ and $\omega=\si_{\tr}^{\frac{1}{2}}a \si_{\tr}^{\frac{1}{2}}$. The inequality (2) follows by choosing $\omega=E_{\cN *}(\rho)$.  This proves the inequality \eqref{eq:indexnorm}. Applying \eqref{eq:indexnorm} to the inclusion $\mathbb{M}_n(\cN)\subset \mathbb{M}_n(\cM)$ for all $n$ yields the CB-version \eqref{eq:cbindexnorm}.
%  Then,
% \begin{align*}
%     \|\id:L_1(\sigma_\tr)\to L_1^\infty(\cN\subset \cM)\|&=\sup_{x\ge0,\,\|x\|_{L_1(\sigma_\tr)}=1}\,\|x\|_{L_1^\infty(\cN\subset\cM)}\\
%     &=\sup_{x\ge0,\,\|x\|_{L_1(\sigma_\tr)}=1}\,\sup_{y\ge 0 ,\, \|y\|_{L_\infty^1(\cN\subset\cM)\le 1}}\,\langle y,\,x\rangle_{\sigma_\tr}\\
%     &=\sup_{y\ge 0 ,\, \|y\|_{L_\infty^1(\cN\subset\cM)\le 1}}\,\|y\|_\infty
% \end{align*}
%By Proposition \ref{opspaceLqp}, we then have
%\begin{align*}
 %   \sup_{n}\,\|\id:L_1(\sigma_\tr^{(n)})\to L_1^\infty(\mathbb{M}_n(\cN)\subset \mathbb{M}_n(\cM))\|=\|\id:L_\infty^1(\cN\subset\cM)\to \cM\|_{\operatorname{cb}}\,.
%\end{align*}
%Therefore,
%\begin{align*}
 %   C_{\tau,\operatorname{cb}}(\cM:\cN)&=\sup_n C_{\tau\otimes \Id_n}(\mathbb{M}_n(\cM):\mathbb{M}_n(\cN))\\
  %  &=\sup_n\inf_c\{c>0|\,\rho\le c\,(\id_n\otimes E_{\cN*})(\rho)\,\forall \rho\in\cD(\mathbb{C}^n\otimes \cH)\}\\
  %&  =\sup_{n}\,\|\id:L_1(\sigma_\tr^{(n)})\to L_1^\infty(\mathbb{M}_n(\cN)\subset \mathbb{M}_n(\cM))\|\,,
%\end{align*}
%where the last identity follows from Proposition \ref{prop_duality} (vi). The result follows.
\end{proof}

\subsection{Proof of \Cref{lemmatechnicalgeneral}}\label{prooflemmatechngen}

The proof of i) is identical to that of \cite[Lemma B.2]{gao2021spectral} in the trace-symmetric case. ii) is a consequence of i) and \Cref{thm:ruan} (see also \cite[Lemma B.2]{gao2021spectral} and \cite[Lemma 3.14]{gao2020fisher}). Recall that for a $\cN$-bimodule map $\Phi$, its module Choi operator $\chi_\Phi=\sum_{i,j}\ket{i}\bra{j}\ten \Phi(\xi_i^\dag\xi_j)\in \cB(l_2^n)\ten \cM$ where $\xi_1,\ldots,\xi_n$ is a module basis for the conditional expectation $E_\cN$ such that
\[ E_\cN(\xi^\dag\xi_j)=\delta_{ij}p_i\ \]
with some projections $p_i\in \cN$. We first show the following inequality (which is actually an equality but we only need one direction here, for the other direction see \cite{gao2020fisher})
\begin{align}\label{eq:inequality} \norm{\Phi:L_\infty^1(\cN\subset \cM)\to \cM}_{\operatorname{cb}}\ge \norm{\chi_\Phi}_{\cB(l_2^n)\ten \cM}\,.\ \end{align}
Define the map $$\Psi:\cT_1(l_2^n)\to \cM, \Psi(|i\rangle\langle j|)=\xi_i^\dagger\xi_j.$$
Its (standard) Choi matrix is $J_{\Psi}=\sum_{i,j}\ket{i}\bra{j}\ten \xi_i^\dag\xi_j\in \cB(l_2^n)\ten \cM$. Then, by \Cref{thm:ruan} and the duality $\cT_1(l_2^n)^*=\cB(\cH)$, we have
\begin{align}
    \|\Psi:\,\cT_1(l_2^n)\to L_\infty^1(\cN\subset \cM)\|_{\operatorname{cb}}&=\|J_\Psi\|_{\cB(l_2^n)\otimes_{\min} L_\infty^1(\cN\subset \cM)}\nonumber\\
    & \overset{(1)}{=}\|J_\Psi\|_{\mathbb{M}_n(L_\infty^1(\cN\subset \cM))}\nonumber\\
&\overset{(2)}{=} \|J_\Psi\|_{L_\infty^1(\mathbb{M}_n(\cN)\subset\label{eq:chipsi} \mathbb{M}_n(\cM))}\,
\end{align}
where (1) is a consequence of the second part of \Cref{thm:ruan} and (2) comes from the operator space structure of $L_\infty^1$ in \Cref{opspaceLqp}. Because $J_\Psi$ is a positive element in $\mathbb{M}_n(\cM)$, we have by \eqref{eq_prop_normq_positive}:
\begin{align*}
    \|J_\Psi\|_{L_\infty^1(\mathbb{M}_n(\cN)\subset \mathbb{M}_n(\cM))}&=\underset{\substack{
 		a\in\mathbb{M}_n(\cN),~a>0 \\\|a\|_{L_1(\sigma_\tr^{(n)})}=1}}{\sup}\,\norm{a^{1/2}\,J_\Psi\,a^{1/2}}_{L_1(\sigma_{\Tr}^{(n)})}\\
 		&\overset{(1)}{=}\underset{\substack{
 			a\in\mathbb{M}_n(\cN),~a>0 \\\|a\|_{L_1(\sigma_\tr^{(n)})}=1}}{\sup}\,\langle a,\,J_\Psi\rangle_{\sigma_\tr^{(n)}}\\
 		&\overset{(2)}{=}\underset{\substack{
 			a\in\mathbb{M}_n(\cN),~a>0 \\\|a\|_{L_1(\sigma_\tr^{(n)})}=1}}{\sup}\,\langle a,\,(\id_n\otimes E_\cN)(J_\Psi)\rangle_{\sigma_\tr^{(n)}}\\
 			&\overset{(3)}{=}\|(\id_n\otimes E_\cN)(J_\Psi)\|_{\mathbb{M}_n(\cN)}\\
 			&=\big\|\sum_{i}|i\rangle\langle i|\otimes p_i\big\|_\infty \\
 			&\le 1\,,
\end{align*}
where in $(1)$ we use that $[\sigma_\tr^{(n)},\mathbb{M}_n(\cN)]=0$, whereas $(2)$ follows by $(\id_n\otimes E_\cN)(a)=a$ and self-adjointness of $\id_n\otimes E_\cN$ with respect to $\langle .,.\rangle_{\sigma_\tr^{(n)}}$. Moreover, $(3)$ follows by duality. Combining this bound with \eqref{eq:chipsi}, we obtain
\begin{align}\label{psicompletecontraction}
    \|\Psi:\,\cT_1(l_2^n)\to L_\infty^1(\cN\subset \cM)\|_{\operatorname{cb}}\le 1\,.
\end{align}
Note that the standard Choi matrix of $\Phi\circ \Psi$ is the module Choi matrix of $\Phi$, i.e.
$J_{\Phi\circ \Psi}=\chi_{\Phi}$.
Then, by Proposition \ref{thm:ruan} again
\begin{align*}
\norm{\chi_\Phi}_{\cB(l_2^n)\ten \cM}
= \|J_{\Phi\circ \Psi}\|_{\cB(l_2^n)\ten \cM}
= \norm{\Phi\circ \Psi:\cT_1(l_2^n)\to \cM}_{\operatorname{cb}}
\le \norm{\Phi:L_\infty^1(\cN\subset \cM)\to \cM}_{\operatorname{cb}} \, ,
\end{align*}
where the last inequality uses \eqref{psicompletecontraction}. This proves \eqref{eq:inequality}. Now consider $\Phi$ to be a unital $\cN$-bimodule map that is self-adjoint w.r.t. the KMS inner product $\langle \cdot ,\cdot \rangle_{\si_\Tr}$. It follows that $\Phi\circ E_\cN=E_\cN\circ \Phi=E_\cN$. Indeed, for any $x\in \cM,y\in \cN$:
\[ \Phi(y)=y \ , \ \Phi(\Id)=\Id\ , \ \langle y, E_\cN\circ \Phi(x)\rangle_{\si_\Tr}=\langle \Phi\circ E_\cN(y), x\rangle_{\si_\Tr}=\langle E_\cN(y), x\rangle_{\si_\Tr}=\langle y, E_\cN(x)\rangle_{\si_\Tr}\,.\]
Thus, $(\Phi-E_{\cN})^k=\Phi^k-E_{\cN}$ and we control the norm of the difference between the module Choi operators of $\Phi^k$ and $E_\cN$ as follows:
\begin{align*}
&\norm{\chi_{\Phi^k}-\chi_{E_{\cN}}}_{\cB(l_2^n)\ten \cM}\\
&~~~~~\overset{(1)}{\le}\norm{\Phi^k-E_{\cN}:L_\infty^1(\cN\subset \cM)\to \cM }_{\operatorname{cb}}\\
&~~~~~=\norm{(\Phi-E_{\cN})^k:L_\infty^1(\cN\subset \cM)\to \cM }_{\operatorname{cb}}
\\
&~~~~~=\norm{\id:L_\infty^1(\cN\subset \cM)\to L_\infty^2(\cN\subset \cM)}_{\operatorname{cb}}
\cdot
\norm{(\Phi-E_{\cN})^k:L_\infty^2(\cN\subset \cM)\to L_\infty^2(\cN\subset \cM)}_{\operatorname{cb}}\\ &~~~~~\quad\quad\,\qquad\qquad\qquad\qquad\qquad\qquad\qquad \qquad\, \qquad\cdot
\norm{\id:L_\infty^2(\cN\subset \cM)\to  \cM}_{\operatorname{cb}}
\\&~~~~~\overset{(2)}{=} \norm{\id:L_\infty^1(\cN\subset\cM)\to \cM}_{\operatorname{cb}}\,\lambda^k.
\end{align*}
Here, (1) is a consequence of the inequality \eqref{eq:inequality} proved above. $(2)$ follows the duality
\begin{align*}
    \norm{\id:L_\infty^1(\cN\subset \cM)\to L_\infty^2(\cN\subset \cM)}_{\operatorname{cb}} & =\norm{\id:L_\infty^2(\cN\subset \cM)\to  \cM}_{\operatorname{cb}} \\
    & =\norm{\id:L_\infty^1(\cN\subset \cM)\to \cM}_{\operatorname{cb}}^{1/2}\, 
\end{align*}
and by \Cref{eq:bimodulecb}, we have for the $L_\infty^2\to L_\infty^2$ norm
\begin{align*}
&\norm{(\Phi-E_{\cN})^k:L_\infty^2(\cN\subset \cM)\to L_\infty^2(\cN\subset \cM)}_{\operatorname{cb}}\\
=& \sup_{n}\norm{\id_{M_n}\ten(\Phi-E_{\cN})^k:L_\infty^2(M_n(\cN)\subset M_n(\cM))\to L_\infty^2(M_n(\cN)\subset M_n(\cM)}_{}\\
=& \sup_{n}\norm{\id_{M_n}\ten(\Phi-E_{\cN})^k:L_2( M_n(\cM))\to L_2(M_n( M_n(\cM)}_{}\\
=& \norm{(\Phi-E_{\cN})^k:L_2( \cM)\to L_2( \cM)}_{}\le \lambda^k
\end{align*}
Therefore, the assertion ii) follows from i) and Lemma \ref{lemmaindexlpspace}.
\begin{flushright}
   \qedsymbol
\end{flushright}

\section*{Acknowledgements}
IB is supported by French A.N.R. grant: ANR-20-CE47-0014-01 ``ESQuisses''. AC was partially supported by an MCQST Distinguished Postdoc and the Seed Funding Program of the MCQST (EXC-2111/Projekt-ID: 390814868). AL acknowledges support from the BBVA Fundation and the Spanish ``Ramón y Cajal'' Programme (RYC2019-026475-I / AEI / 10.13039/501100011033). This project has received funding from the European Research Council (ERC) under the European Union’s Horizon 2020 research and innovation programme (grant agreement No 648913).  DPG acknowledges support from MINECO (grant MTM2017-88385-P) and from Comunidad de Madrid (grant QUITEMAD-CM, ref. P2018/TCS-4342). CR acknowledges financial support from a Junior Researcher START Fellowship from the MCQST, and AC and CR also acknowledge financial support by the DFG cluster of excellence 2111 (Munich Center for Quantum Science and Technology).

\bibliographystyle{abbrv}
\bibliography{references}

\begin{thebibliography}{10}

\bibitem{aharonov2009detectability}
D.~Aharonov, I.~Arad, Z.~Landau, and U.~Vazirani.
\newblock The detectability lemma and quantum gap amplification.
\newblock In {\em Proceedings of the forty-first annual ACM symposium on Theory
  of computing}, pages 417--426, 2009.

\bibitem{Alicki_2009}
R.~Alicki, M.~Fannes, and M.~Horodecki.
\newblock On thermalization in {K}itaev{\textquotesingle}s 2{D} model.
\newblock {\em Journal of Physics A: Mathematical and Theoretical},
  42(6):065303, jan 2009.

\bibitem{Alicki2010}
R.~Alicki, M.~Horodecki, P.~Horodecki, and R.~Horodecki.
\newblock On thermal stability of topological qubit in
  {K}itaev{\textquotesingle}s 4{D} model.
\newblock {\em Open Systems {\&} Information Dynamics}, 17(01):1--20, Mar.
  2010.

\bibitem{anshu2016simple}
A.~Anshu, I.~Arad, and T.~Vidick.
\newblock Simple proof of the detectability lemma and spectral gap
  amplification.
\newblock {\em Physical Review B}, 93(20):205142, 2016.

\bibitem{araki1969gibbs}
H.~Araki.
\newblock Gibbs states of a one dimensional quantum lattice.
\newblock {\em Communications in Mathematical Physics}, 14(2):120--157, 1969.

\bibitem{baillet1988indice}
M.~Baillet, Y.~Denizeau, and J.-F. Havet.
\newblock Indice d'une esp{\'e}rance conditionnelle.
\newblock {\em Compositio mathematica}, 66(2):199--236, 1988.

\bibitem{Bardet-NonCommFunctInequalities-2017}
I.~Bardet.
\newblock Estimating the decoherence time using non-commutative functional
  inequalities.
\newblock {\em arXiv preprint arXiv:1710.01039}, 2017.

\bibitem{bardet2021MLSIDavies1Dshort}
I.~Bardet, {\'A}.~Capel, L.~Gao, A.~Lucia, D.~P{\'e}rez-Garc{\'\i}a, and
  C.~Rouz{\'e}.
\newblock Rapid thermalization of spin chain commuting hamiltonians.
\newblock {\em in preparation}, 2021.

\bibitem{bardet2019modified}
I.~Bardet, {\'{A}}.~Capel, A.~Lucia, D.~P{\'e}rez-Garcia, and C.~Rouz{\'e}.
\newblock On the modified logarithmic {S}obolev inequality for the heat-bath
  dynamics for 1{D} systems.
\newblock {\em Journal of Mathematical Physics}, 62(6):061901, 2021.

\bibitem{bardet2020approximate}
I.~Bardet, {\'{A}}.~Capel, and C.~Rouz{\'{e}}.
\newblock Approximate tensorization of the relative entropy for noncommuting
  conditional expectations.
\newblock {\em Annales Henri Poincar{\'{e}}}, July 2021.

\bibitem{bardet2018hypercontractivity}
I.~Bardet and C.~Rouz{\'e}.
\newblock Hypercontractivity and logarithmic {S}obolev inequality for
  non-primitive quantum {M}arkov semigroups and estimation of decoherence
  rates.
\newblock {\em arXiv preprint arXiv:1803.05379}, 2018.

\bibitem{BeigiDattaRouze-ReverseHypercontractivity-2018}
S.~Beigi, N.~Datta, and C.~Rouzé.
\newblock Quantum reverse hypercontractivity: its tensorization and application
  to strong converses.
\newblock {\em Communications in Mathematical Physics}, 376(2):753--794, 2018.

\bibitem{bergh2012interpolation}
J.~Bergh and J.~L{\"o}fstr{\"o}m.
\newblock {\em Interpolation spaces: an introduction}, volume 223.
\newblock Springer Science \& Business Media, 2012.

\bibitem{bhatia2013matrix}
R.~Bhatia.
\newblock {\em Matrix analysis}, volume 169.
\newblock Springer Science \& Business Media, 2013.

\bibitem{bluhm2021exponential}
A.~Bluhm, {\'{A}}.~Capel, and A.~Pérez-Hernández.
\newblock Exponential decay of mutual information for {G}ibbs states of local
  {H}amiltonians.
\newblock {\em arXiv preprint arXiv:2104.04419}, 2021.

\bibitem{brannan2021complete2}
M.~Brannan, L.~Gao, and M.~Junge.
\newblock Complete logarithmic {S}obolev inequality via {R}icci curvature
  bounded below {II}.
\newblock {\em Journal of Topology and Analysis}, pages 1--54, 2021.

\bibitem{Briegel2001}
H.~J. Briegel and R.~Raussendorf.
\newblock Persistent entanglement in arrays of interacting particles.
\newblock {\em Physical Review Letters}, 86(5):910--913, Jan. 2001.

\bibitem{capel2019thesis}
{\'{A}}.~Capel.
\newblock {\em Quantum Logarithmic {S}obolev Inequalities for Quantum Many-Body
  Systems: An approach via Quasi-Factorization of the Relative Entropy}.
\newblock Ph.D. thesis at Universidad Aut{\'o}noma de Madrid, 2019.

\bibitem{capel2017superadditivity}
{\'A}.~Capel, A.~Lucia, and D.~P{\'e}rez-Garcia.
\newblock Superadditivity of quantum relative entropy for general states.
\newblock {\em IEEE Transactions on Information Theory}, 64(7):4758--4765,
  2017.

\bibitem{capel2018quantum}
{\'A}.~Capel, A.~Lucia, and D.~P{\'e}rez-Garcia.
\newblock Quantum conditional relative entropy and quasi-factorization of the
  relative entropy.
\newblock {\em Journal of Physics A: Mathematical and Theoretical},
  51(48):484001, 2018.

\bibitem{capel2020MLSI}
{\'{A}}.~Capel, C.~Rouzé, and D.~Stilck~França.
\newblock {The modified logarithmic {S}obolev inequality for quantum spin
  systems: classical and commuting nearest neighbour interactions}.
\newblock {\em arXiv preprint, arXiv:2009.11817}, 2020.

\bibitem{CarboneMartinelli_LogSobolev_2015}
R.~Carbone and A.~Martinelli.
\newblock {Logarithmic {S}obolev inequalities in non-commutative algebras}.
\newblock {\em Infinite Dimensional Analysis, Quantum Probability and Related
  Topics}, 18(02):1550011, 2015.

\bibitem{cesi2001quasi}
F.~Cesi.
\newblock Quasi-factorization of the entropy and logarithmic {S}obolev
  inequalities for {G}ibbs random fields.
\newblock {\em Probability Theory and Related Fields}, 120(4):569--584, 2001.

\bibitem{Chen2011}
X.~Chen, Z.-C. Gu, and X.-G. Wen.
\newblock Classification of gapped symmetric phases in one-dimensional spin
  systems.
\newblock {\em Physical Review B}, 83(3), Jan. 2011.

\bibitem{coser2019classification}
A.~Coser and D.~P{\'e}rez-Garc{\'\i}a.
\newblock Classification of phases for mixed states via fast dissipative
  evolution.
\newblock {\em Quantum}, 3:174, 2019.

\bibitem{CubittLuciaMichalakisPerezGarcia_StabilityRapidMixing_2015}
T.~S. Cubitt, A.~Lucia, S.~Michalakis, and D.~P{\'e}rez-Garc{\'\i}a.
\newblock Stability of local quantum dissipative systems.
\newblock {\em Communications in Mathematical Physics}, 337(3):1275--1315, Apr.
  2015.

\bibitem{daipra2002classicalMLSI}
P.~{Dai Pra}, A.~M. Paganoni, and G.~Posta.
\newblock {Entropy inequalities for unbounded spin systems}.
\newblock {\em The Annals of Probability}, 30(4):1959--1976, 10 2002.

\bibitem{datta2009min}
N.~Datta.
\newblock Min-and max-relative entropies and a new entanglement monotone.
\newblock {\em IEEE Transactions on Information Theory}, 55(6):2816--2826,
  2009.

\bibitem{Davies1976}
E.~Davies.
\newblock {\em Quantum theory of open systems}.
\newblock London, New York: Academic Press, 1976.

\bibitem{Davies1979}
E.~Davies.
\newblock Generators of dynamical semigroups.
\newblock {\em Journal of Functional Analysis}, 34(3):421--432, Dec. 1979.

\bibitem{archbold_1983}
E.~B. Davies.
\newblock One-parameter semigroups (academic press, london, 1980), viii 230 pp.
\newblock {\em Proceedings of the Edinburgh Mathematical Society},
  26(1):115–116, 1983.

\bibitem{de2021quantum}
G.~De~Palma and C.~Rouz{\'e}.
\newblock Quantum concentration inequalities.
\newblock {\em arXiv preprint arXiv:2106.15819}, 2021.

\bibitem{effros2000operator}
E.~Effros and Z.~Ruan.
\newblock {\em Operator Spaces}.
\newblock London Mathematical Society monographs. Clarendon Press, 2000.

\bibitem{frigerio1982long}
A.~Frigerio and M.~Verri.
\newblock Long-time asymptotic properties of dynamical semigroups on
  {W}$^*$-algebras.
\newblock {\em Mathematische Zeitschrift}, 180(3):275--286, 1982.

\bibitem{gao2020fisher}
L.~Gao, M.~Junge, and N.~LaRacuente.
\newblock Fisher information and logarithmic {S}obolev inequality for
  matrix-valued functions.
\newblock In {\em Annales Henri Poincar{\'e}}, volume~21, pages 3409--3478.
  Springer, 2020.

\bibitem{gaoindex}
L.~Gao, M.~Junge, and N.~LaRacuente.
\newblock Relative entropy for von {N}eumann subalgebras.
\newblock {\em International Journal of Mathematics}, 31(06):2050046, 2020.

\bibitem{gao2021geometric}
L.~Gao, M.~Junge, and H.~Li.
\newblock Geometric approach towards complete logarithmic {S}obolev
  inequalities.
\newblock {\em arXiv preprint arXiv:2102.04434}, 2021.

\bibitem{gao2021spectral}
L.~Gao and C.~Rouz{\'e}.
\newblock Complete entropic inequalities for quantum {M}arkov chains.
\newblock {\em arXiv preprint arXiv:2102.04146}, 2021.

\bibitem{GYZ19interpolation}
J.~Gu, Z.~Yin, and H.~Zhang.
\newblock Interpolation of quasi noncommutative $ {L}_p $-spaces.
\newblock {\em arXiv preprint arXiv:1905.08491}, 2019.

\bibitem{holley1989uniform}
R.~A. Holley and D.~W. Stroock.
\newblock {Uniform and $L_2$ convergence in one dimensional stochastic Ising
  models}.
\newblock {\em Communications in Mathematical Physics}, 123(1):85--93, 1989.

\bibitem{junge2019stability}
M.~Junge, N.~LaRacuente, and C.~Rouz{\'e}.
\newblock Stability of logarithmic {S}obolev inequalities under a
  noncommutative change of measure.
\newblock {\em arXiv preprint arXiv:1911.08533}, 2019.

\bibitem{junge2010mixed}
M.~Junge and J.~Parcet.
\newblock {\em Mixed-norm inequalities and operator space $ L_p $ embedding
  theory}.
\newblock American Mathematical Soc., 2010.

\bibitem{kastoryano2016quantum}
M.~J. Kastoryano and F.~G. Brandao.
\newblock Quantum {G}ibbs samplers: the commuting case.
\newblock {\em Communications in Mathematical Physics}, 344(3):915--957, 2016.

\bibitem{KastoryanoTemme-LogSobolevInequalities-2013}
M.~J. Kastoryano and K.~Temme.
\newblock Quantum logarithmic {S}obolev inequalities and rapid mixing.
\newblock {\em Journal of Mathematical Physics}, 54(5):052202, May 2013.

\bibitem{laracuente2019quasi}
N.~LaRacuente.
\newblock Quasi-factorization and multiplicative comparison of
  subalgebra-relative entropy.
\newblock {\em arXiv preprint arXiv:1912.00983}, 2019.

\bibitem{lucia2021thermalization}
A.~Lucia, D.~P{\'e}rez-Garc{\'i}a, and A.~P{\'e}rez-Hern{\'a}ndez.
\newblock {Thermalization in Kitaev's quantum double models via Tensor Network
  techniques}.
\newblock {\em arXiv preprint arXiv:2107.01628}, 2021.

\bibitem{mcginley2019interacting}
M.~McGinley and N.~R. Cooper.
\newblock Interacting symmetry-protected topological phases out of equilibrium.
\newblock {\em Physical Review Research}, 1(3):033204, 2019.

\bibitem{mcginley2020fragility}
M.~McGinley and N.~R. Cooper.
\newblock Fragility of time-reversal symmetry protected topological phases.
\newblock {\em Nature Physics}, 16(12):1181--1183, 2020.

\bibitem{MullerHermesFrancaWolf-EntropyProduction-2016}
A.~Müller-Hermes, D.~S. França, and M.~M. Wolf.
\newblock Entropy production of doubly stochastic channels.
\newblock {\em Journal of Mathematical Physics}, 57:022203, 2016.

\bibitem{MullerHermesFrancaWolf-DepolarizingChannels-2016}
A.~Müller-Hermes, D.~S. França, and M.~M. Wolf.
\newblock Relative entropy convergence for depolarizing channels.
\newblock {\em Journal of Mathematical Physics}, 57:022202, 2016.

\bibitem{nacu2003glauber}
{\c{S}}.~Nacu.
\newblock Glauber dynamics on the cycle is monotone.
\newblock {\em Probability theory and related fields}, 127(2):177--185, 2003.

\bibitem{DePalma2021}
G.~D. Palma, M.~Marvian, D.~Trevisan, and S.~Lloyd.
\newblock The quantum {W}asserstein distance of order 1.
\newblock {\em {IEEE} Transactions on Information Theory}, pages 1--1, 2021.

\bibitem{paschke1973inner}
W.~L. Paschke.
\newblock Inner product modules over ${B}^*$-algebras.
\newblock {\em Transactions of the American Mathematical Society},
  182:443--468, 1973.

\bibitem{popapimser}
M.~Pimsner and S.~Popa.
\newblock Entropy and index for subfactors.
\newblock {\em Annales scientifiques de l'\'Ecole Normale Sup\'erieure}, Ser.
  4, 19(1):57--106, 1986.

\bibitem{Pisier98NCvector}
G.~Pisier.
\newblock {\em Non-commutative vector valued Lp-spaces and completely p-summing
  maps}.
\newblock Soci{\'e}t{\'e} math{\'e}matique de France, 1998.

\bibitem{Pisier03}
G.~Pisier.
\newblock {\em Introduction to operator space theory}.
\newblock Number 294. Cambridge University Press, 2003.

\bibitem{Preskill2018}
J.~Preskill.
\newblock Quantum computing in the {NISQ} era and beyond.
\newblock {\em Quantum}, 2:79, Aug. 2018.

\bibitem{Rouz2019}
C.~Rouz{\'{e}} and N.~Datta.
\newblock Concentration of quantum states from quantum functional and
  transportation cost inequalities.
\newblock {\em Journal of Mathematical Physics}, 60(1):012202, Jan. 2019.

\bibitem{rouze2021learning}
C.~Rouz{\'e} and D.~S. Fran{\c{c}}a.
\newblock Learning quantum many-body systems from a few copies.
\newblock {\em arXiv preprint arXiv:2107.03333}, 2021.

\bibitem{Ruan88}
Z.-J. Ruan.
\newblock Subspaces of ${C}^*$-algebras.
\newblock {\em Journal of Functional Analysis}, 76(1):217--230, 1988.

\bibitem{Son2011}
W.~Son, L.~Amico, R.~Fazio, A.~Hamma, S.~Pascazio, and V.~Vedral.
\newblock Quantum phase transition between cluster and antiferromagnetic
  states.
\newblock {\em {EPL} (Europhysics Letters)}, 95(5):50001, Aug. 2011.

\bibitem{SL78}
H.~Spohn and J.~L. Lebowitz.
\newblock {Irreversible thermodynamics for quantum systems weakly coupled to
  thermal reservoirs}.
\newblock {\em Advances in Chemical Physics}, 38:109--142, 1978.

\bibitem{Zegarlinski1990}
B.~Zegarlinski.
\newblock Log-{S}obolev inequalities for infinite one dimensional lattice
  systems.
\newblock {\em Communications in Mathematical Physics}, 133(1):147--162, Sept.
  1990.

\end{thebibliography}

\end{document}